\documentclass[10pt,reqno]{amsart}
\usepackage{latexsym,amsmath}
\usepackage{amsfonts}
\usepackage{amssymb}
\usepackage{fixmath}

\usepackage{fullpage}
\usepackage{graphicx}
\usepackage[colorlinks=true,citecolor=blue]{hyperref}
\usepackage[dvipsnames]{xcolor}

\usepackage{enumerate}

\usepackage[T1]{fontenc}
\usepackage{fourier}
\usepackage{bbm}

\usepackage{tikz}
\usetikzlibrary{arrows.meta,calc,decorations.pathreplacing,positioning}

\usepackage[boxed, linesnumbered, noend, noline]{algorithm2e}
\SetKwInput{KwData}{Input}
\SetKwInput{KwResult}{Output}

\numberwithin{equation}{section}

\renewcommand\root{\mathfrak{x}}
\newcommand\pulp{\ensuremath{\mathtt{PULP}}}
\newcommand\duniq{d_{\mathrm{uniq}}}
\newcommand\dMS{d_{\mathrm{MS}}}
\newcommand\dpure{d_{\mathrm{pure}}}
\newcommand\PM{\{\pm1\}}
\newcommand\ism{\cong}

\newcommand\sign{\mathrm{sign}}

\newcommand\disteq{\sim}
\newcommand\dist{\mathrm{dist}}
\newcommand\MU{\vec\mu}
\newcommand{\BP}{\mathrm{BP}}

\newcommand{\fU}{\mathfrak U}
\newcommand{\fB}{\mathfrak B}
\newcommand{\fT}{\mathfrak T}

\newcommand{\vDelta}{\vec\Delta}

\newcommand{\vd}{\vec d}

\newcommand{\vN}{\vec N}

\newcommand{\vs}{\vec s}

\newcommand{\vX}{\vec X}
\newcommand{\vL}{\vec L}

\renewcommand{\epsilon}{\eps}
\renewcommand{\subset}{\subseteq}

\newcommand\vY{\vec Y}

\newcommand\vr{\vec r}
\newcommand\vm{\vec m}

\newcommand\ETA{{\vec\eta}}

\newcommand\PHI{\vec\Phi}

\newcommand\nix{\,\cdot\,}

\newcommand\dd{{\mathrm d}}

\renewcommand{\vec}[1]{\boldsymbol{#1}}

\newcommand\SIGMA{\vec\sigma}

\newcommand\TAU{\vec\tau}

\newtheorem{definition}{Definition}[section]
\newtheorem{claim}[definition]{Claim}

\newtheorem{remark}[definition]{Remark}
\newtheorem{theorem}[definition]{Theorem}
\newtheorem{lemma}[definition]{Lemma}
\newtheorem{proposition}[definition]{Proposition}
\newtheorem{corollary}[definition]{Corollary}

\newtheorem{fact}[definition]{Fact}

\newcommand\fG{\mathfrak{G}}

\newcommand\fE{\mathfrak{E}}

\newcommand\fF{\mathfrak{F}}

\newcommand\cA{\mathcal{A}}

\newcommand\cC{\mathcal{C}}

\newcommand\cF{\mathcal{F}}

\newcommand\cU{\mathcal{U}}

\newcommand\cT{\mathcal{T}}

\newcommand\cL{\mathcal{L}}

\newcommand\cP{\mathcal{P}}
\newcommand\cX{\mathcal{X}}

\newcommand\cV{\mathcal{V}}

\def\cR{{\mathcal R}}

\newcommand\vx{\vec x}
\newcommand\vZ{\vec Z}

\newcommand\THETA{\vec\theta}

\newcommand\eps{\varepsilon}

\newcommand\NN{\mathbb{N}}
\newcommand\TT{\mathbb{T}}

\newcommand\Erw{\mathbb{E}}
\newcommand{\vecone}{\mathbb{1}}

\newcommand{\Po}{{\rm Po}}

\newcommand\bc[1]{\left({#1}\right)}
\newcommand\cbc[1]{\left\{{#1}\right\}}

\newcommand\brk[1]{\left\lbrack{#1}\right\rbrack}

\newcommand\abs[1]{\left|{#1}\right|}

\newcommand\RR{\mathbb{R}}

\def\?#1{}
\makeatletter
\def\whp{w.h.p\@ifnextchar-{.}{\@ifnextchar.{.\?}{\@ifnextchar,{.}{\@ifnextchar){.}{\@ifnextchar:{.:\?}{.\ }}}}}}
\makeatother

\makeatletter
\def\Whp{W.h.p\@ifnextchar-{.}{\@ifnextchar.{.\?}{\@ifnextchar,{.}{\@ifnextchar){.}{\@ifnextchar:{.:\?}{.\ }}}}}}
\makeatother

\newcommand{\Lovasz}{Lov\'asz}

\newcommand{\Chvatal}{Chv\'{a}tal}

\newcommand\pr{\mathbb{P}}
\renewcommand\Pr{\pr}

\newcommand\Lem{Lemma}
\newcommand\Prop{Proposition}
\newcommand\Cl{Claim}
\newcommand\Thm{Theorem}
\newcommand\Fact{Fact}

\newcommand\Cor{Corollary}
\newcommand\Sec{Section}
\newcommand\Chap{Chapter}

\newcommand{\dsat}{d_{\mathrm{sat}}}
\newcommand{\drsb}{d_{\mathrm{rsb}}}
\newcommand{\dalg}{d_{\mathrm{alg}}}
\newcommand{\dgiant}{d_{\mathrm{giant}}}

\usepackage{upgreek}

\newcommand{\fun}{\uppsi}
\newcommand{\sfun}{\phi}
\newcommand{\Pfun}{\Upgamma}

\usepackage{wasysym}
\usepackage{dsfont}
\usepackage{paralist}

\newcommand*{\NoKid}{\mathrel{\scalebox{0.6}{$\Circle$}}}
\newcommand*{\AllKid}{\mathrel{\scalebox{0.6}{$\CIRCLE$}}}
\newcommand*{\PureP}{\mathrel{\scalebox{0.8}{$\oplus$}}}
\newcommand*{\PureM}{\mathrel{\scalebox{0.8}{$\ominus$}}}

\newcommand{\rNoKid}{\mathrel{\raisebox{0.3pt}{\scalebox{0.8}{$\Circle$}}}}
\newcommand{\rAllKid}{\mathrel{\raisebox{0.3pt}{\scalebox{0.8}{$\CIRCLE$}}}}
\newcommand{\rPureP}{\mathrel{\raisebox{0pt}{$\oplus$}}}
\newcommand{\rPureM}{\mathrel{\raisebox{0pt}{$\ominus$}}}

\newcommand{\diffr}{\mathcal{D}}

\DeclareMathOperator{\type}{tp}

\newcommand{\vdpc}{\vd_{+}^{\star}}
\newcommand{\vdmc}{\vd_{-}^{\star}}

\newcommand{\LArg}{\pmb{\Upxi}}
\newcommand{\rvec}[1]{#1}

\newcommand{\fratio}{{\gamma}}
\newcommand{\fprop}{{\gamma^{-1}}}

\newcommand{\Ex}[1]{E_{#1}}
\newcommand{\Ind}{\vecone}

\newcommand{\dours}{\ensuremath{d_{\mathrm{con}}}}

\newcommand{\LDELit}{{\mathrm{LL}}}

\newcommand{\LDELitS}[1]{{\mathrm{LL}}^{\star}_{#1}}

\newcommand{\LDEP}{\mathrm{LL}^+}

\definecolor{aoEng}{rgb}{0, 0.5,0}

\newcommand{\fn}{\mathfrak{h}}

\usepackage{pifont}

\newcounter{kcomcount}
\def\ex{{\mathbb E}}
\def\pr{{\mathbb P}}

\newcommand\RSA{Random Structures and Algorithms}

\newcommand\CPC{Combinatorics, Probability and Computing}

\begin{document}

\title{The random $k$-SAT Gibbs uniqueness threshold revisited}

\author{Arnab Chatterjee, Amin Coja-Oghlan, Catherine Greenhill, Vincent Pfenninger, Maurice Rolvien, Pavel Zakharov, Kostas
Zampetakis}

\address{Arnab Chatterjee, {\tt arnab.chatterjee@tu-dortmund.de}, TU Dortmund, Faculty of Computer Science, 12 Otto-Hahn-St, Dortmund 44227, Germany.}
\address{Amin Coja-Oghlan, {\tt amin.coja-oghlan@tu-dortmund.de}, TU Dortmund, Faculty of Computer Science and Faculty of Mathematics, 12 Otto-Hahn-St, Dortmund 44227, Germany.}
\address{Catherine Greenhill, {\tt c.greenhill@unsw.edu.au}, School of Mathematics and Statistics, UNSW Sydney, NSW 2052, Australia.}
\address{Vincent Pfenninger, {\tt pfenninger@math.tu-graz.at},TU Graz, Institute of Discrete Mathematics, Steyrergasse 30, 8010 Graz, Austria.}
\address{Maurice Rolvien,
{\tt maurice.rolvien@uni-hamburg.de}, University of Hamburg, Faculty of Mathematics, Informatics and Natural Sciences, Department of Informatics, Vogt-K\"olln-Str.\ 30, 22527 Hamburg, Germany.}
\address{Pavel Zakharov, {\tt pavel.zakharov@tu-dortmund.de},TU Dortmund, Faculty of Computer Science and Faculty of Mathematics, 12 Otto-Hahn-St, Dortmund 44227, Germany.}
\address{Kostas Zampetakis, {\tt konstantinos.zampetakis@tu-dortmund.de},TU Dortmund, Faculty of Computer Science, 12 Otto-Hahn-St, Dortmund 44227, Germany.}

\begin{abstract}%
We prove that for any $k\geq3$ for clause/variable ratios up to the Gibbs uniqueness threshold of the corresponding Galton-Watson tree, the number of satisfying assignments of random $k$-SAT formulas is given by the `replica symmetric solution' predicted by physics methods [Monasson, Zecchina: Phys.\ Rev.\ Lett.\ {\bf76} (1996)].
Furthermore, while the Gibbs uniqueness threshold is still not known precisely for any $k\geq3$, we derive new lower bounds on this threshold that improve over prior work [Montanari and Shah: SODA (2007)].
The improvement is significant particularly for small~$k$.
  \hfill {\em MSc:~68Q87, 60C05, 68R07}
\end{abstract}

\maketitle

\section{Introduction}\label{sec_intro}

\subsection{Background and motivation}\label{sec_motiv}

\noindent
Going back to experimental work from the 1990s, the most prominent question concerning random $k$-SAT has been to pinpoint the satisfiability threshold, defined as the largest density $m/n$ of clauses $m$ to variables $n$ up to which satisfying assignments likely exist~\cite{ANP,Cheeseman}.
Currently, the satisfiability threshold is known precisely in the case of $k=2$~\cite{CR,Goerdt} and for $k\geq k_0$ with $k_0$ an undetermined (large) constant~\cite{DSS3}.
The latter result confirms `predictions' based on an analytic but non-rigorous physics technique called the `cavity method'.
Indeed, the cavity method predicts the satisfiability threshold for every $k\geq3$~\cite{MPZ}, but random $k$-SAT for `small' $k\geq3$ appears to be a 
particularly hard nut to crack.
Additionally, according to the cavity method several phase transitions precede the satisfiability threshold and are expected to impact, among other things, the performance of algorithms~\cite{pnas}.
One of these phase transitions, the {\em Gibbs uniqueness transition}, pertains to a spatial mixing property that also plays a pivotal role in the computational complexity of counting and sampling~\cite{Sly}.

From a statistical physics viewpoint, the satisfiability threshold is 
only the second most important quantity associated with random $k$-SAT.
The first place firmly belongs to the typical number of satisfying assignments, known as the {\em partition function} in physics parlance~\cite{MM}.
All the other predictions, including the location of the satisfiability threshold, ultimately derive from the formula for the number of satisfying assignments or closely related variables~\cite{MMZ}.
Yet there has been little progress on confirming the physics formula for the number of satisfying assignments rigorously.

Three prior contributions stand out.
First, a proof technique called the `interpolation method' turns the physics prediction into a rigorous upper bound~\cite{FranzLeone,Guerra,PanchenkoTalagrand}.%
\footnote{Strictly speaking, the contributions \cite{FranzLeone,Guerra,PanchenkoTalagrand} deal with the `random $k$-SAT model at positive tempertature', see \Sec~\ref{sec_discussion}. In \Cor~\ref{cor_interpolation} below we combine the interpolation method with a concentration argument to bound the number of actual satisfying assignments.}
Second, in the case $k=2$, conceptually much simpler than $k\geq3$, the physics formula has been proved correct~\cite{2sat}.
Third, Montanari and Shah~\cite{MS} proved that also for $k\geq3$ for certain clause/variable densities the `replica symmetric solution' from physics correctly approximates the number of `good' assignments that satisfy all but $o(n)$ clauses.
However, it seems difficult to estimate the gap between the number of such `good' assignments and the number of actual satisfying assignments.
A rigorous method to this effect would likely imply the existence of uniform satisfiability thresholds for all $k\geq3$, thereby resolving a long-standing conundrum~\cite{BGT,Friedgut}.
The proof of Montanari and Shah is based on the aforementioned Gibbs uniqueness property.

The aim of the present paper is to determine the number of {\em actual} satisfying assignments of random $k$-SAT formulas for clause/variable densities up to the Gibbs uniqueness threshold.
Specifically, we verify that the `replica symmetric solution' from~\cite{MZ,MZ2} yields the correct answer for any $k\geq3$ right up to the Gibbs uniqueness threshold, 
even though the precise value of this threshold is not currently known.
Additionally, we derive a new lower bound on the Gibbs uniqueness threshold.
The improvement is particularly significant for `small' $k\geq3$.
Combining these two results, we obtain the first rigorous formula for the number of satisfying assignments of random $k$-SAT formula for a non-trivial regime of clause/variable densities.
Crucially, the result covers meaningful clause/variable densities even for small $k\geq3$.

\subsection{Results}\label{sec_results}

\noindent
Let $\PHI=\PHI_{d,k}(n)$ be the random $k$-CNF on $n$ Boolean variables $x_1,\ldots,x_n$ with $\vm=\vm_n\disteq\Po(dn/k)$ clauses $a_1,\ldots,a_{\vm}$.
The clauses $a_i$ are drawn independently and uniformly from the set of all $2^k \binom{n}{k}$ possible clauses with $k$ distinct variables.
Hence, the parameter $d$ prescribes the expected number of clauses in which a given variable appears.
Let $S(\PHI)$ be the set of satisfying assignments of $\PHI$ and let $Z(\PHI)=|S(\PHI)|$.
We encode the Boolean values `true' and `false' by $+1$ and $-1$, respectively.
Since right up to the satisfiability threshold $Z(\PHI)$ is of order $\exp(\Theta(n))$ \whp\ for trivial reasons%
\footnote{For example, \whp\ there are $\Omega(n)$ variables that do not appear in any clause.},
our objective is to study the random variable $n^{-1}\log Z(\PHI)$ as $n\to\infty$.

\subsubsection{The number of satisfying assignments up to the Gibbs uniqueness threshold}\label{sec_results_num_sol}

The first main result vindicates the `replica symmetric solution' for values of $d$ up to the Gibbs uniqueness threshold of the Galton-Watson tree that mimics the local topology of $\PHI$.
Let us define these concepts precisely.

We begin with the Galton-Watson tree $\TT=\TT_{d,k}$, which is generated by a two-type branching process.
The two types are {\em variable nodes} and {\em clause nodes}.
The process starts with a single root variable node $\root$.
The offspring of any variable node is a $\Po(d)$ number of clause nodes, while every clause node begets precisely $k-1$ variable nodes.
Additionally, independently for each clause node $a$ and every variable node $x$ that is either a child or the parent of $a$ a {\em sign}, denoted $\sign(x,a)\in\PM$, 
is chosen uniformly at random.
The resulting random tree $\TT$ models the local structure of the random formula $\PHI$ in the sense of local weak convergence~\cite{AldousSteele,Lovasz}.%
\footnote{\Cor~\ref{cor_lcwk} below provides a precise statement to this effect.}

Next, we define the Gibbs uniqueness property on the tree $\TT$.
For an integer $\ell\geq0$ let $\TT^{(\ell)}$ be the finite tree obtained by removing all variable and clause nodes at a distance greater than $2\ell$ from the root $\root$.
We identify the finite tree $\TT^{(\ell)}$ with a Boolean formula whose variables/clauses are precisely the variable/clause nodes of $\TT^{(\ell)}$.
Let $S(\TT^{(\ell)})\neq\emptyset$ be the set of satisfying assignments of this formula and let $\TAU^{(\ell)}\in S(\TT^{(\ell)})$ be a uniformly random satisfying assignment.
Moreover, let $\partial^{2\ell}\root$ be the set of variable nodes of $\TT^{(\ell)}$ at distance precisely $2\ell$ from the root $\root$.
Then for given $d,k$ the tree $\TT=\TT_{d,k}$ has the {\em Gibbs uniqueness property} if
\begin{align}\label{eqTreeUniq}
  \lim_{\ell\to\infty}\ex\brk{\max_{\tau\in S(\TT^{(\ell)})}\abs{ \pr\brk{\TAU^{(\ell)}(\root)=1\mid\TT} -\pr\brk{\TAU^{(\ell)}(\root)=1\mid\TT,\,\forall x\in\partial^{2\ell}\root:\TAU^{(\ell)}(x)=\tau(x)} }}&=0&&\mbox{(see~\cite{pnas})}.
\end{align}
In words, in the limit of large $\ell$ the truth value $\TAU^{(\ell)}(\root)$ of the root $\root$ is asymptotically independent of the truth values $\{\TAU^{(\ell)}(x)\}_{x\in\partial^{2\ell}\root}$ of the variables at distance $2\ell$ from $\root$.
In light of the above, for any $k \ge 2$ we further define $\duniq(k)$ as 
\begin{align}
  \duniq(k) = \inf \{ d>0: \text{ condition \eqref{eqTreeUniq} fails to hold for } {d,k}\} \enspace.
\end{align}
It is easy to see that $\duniq(k)$ is strictly positive and finite for any $k\geq2$.
Indeed, in \Thm~\ref{thm_uniq} we will derive explicit lower bounds on $\duniq(k)$.
However, the exact value of $\duniq(k)$ is not currently known for any $k\geq3$.

As a final preparation we need to spell out the `replica symmetric solution' from~\cite{MZ}.
This prediction comes in terms of a distributional fixed point problem, i.e., a fixed point problem on the space $\cP(0,1)$ of probability measures on the open unit interval.
Specifically, consider the {\em Belief Propagation operator}
\begin{align}\label{eqBPop}
  \BP_{d,k}&:\cP(0,1)\to\cP(0,1),&\pi&\mapsto\hat\pi=\BP_{d,k}(\pi)
\end{align}
defined as follows.
Let $\vd^+,\vd^-\disteq\Po(d/2)$ be Poisson variables with expectation $d/2$.
Moreover, let $(\MU_{\pi,i,j})_{i,j\geq1}$ be a sequence of i.i.d. random variables, each following distribution $\pi$.
All these random variables are mutually independent. Further, let
\begin{align}\label{eqhat}
  \MU_{\pi,i}&=1-\prod_{j=1}^{k-1}\MU_{\pi,i,j}\quad\mbox{ for $i\geq1$,}&&\mbox{and}&
  \hat\MU_{\pi}&=\frac{\prod_{i=1}^{\vd^-}\MU_{\pi,2i-1}}{\prod_{i=1}^{\vd^-}\MU_{\pi,2i-1}+\prod_{i=1}^{\vd^+}\MU_{\pi,2i}}.
\end{align}
Then $\hat\pi$ is the distribution of $\hat\MU_{\pi}$.
Furthermore, for a probability measure $\pi\in\cP(0,1)$ define the {\em Bethe free entropy}%
\footnote{Throughout the paper $\log$ refers to the natural logarithm.}
\begin{align}\label{eqBethe}
  \fB_{d,k}(\pi)&=\ex\brk{\log\bc{\prod_{i=1}^{\vd^-}\MU_{\pi,2i-1}+\prod_{i=1}^{\vd^+}\MU_{\pi,2i}}-\frac{d(k-1)}{k}\log\bc{1-\prod_{j=1}^k\MU_{\pi,1,j}}},
\end{align}
provided that the expectation on the r.h.s.\ exists.
Finally, let $\delta_{1/2}\in\cP(0,1)$ be the atom at $1/2$ and let us write $\BP_{d,k}^\ell$ for the $\ell$-fold application of the operator $\BP_{d,k}$.

\begin{theorem}\label{thm_main}
  Let $k\geq3$ and assume that $0<d<\duniq(k)$.
  Then the weak limit
  \begin{align}\label{eqpidk}
    \pi_{d,k}=\lim_{\ell\to\infty}\BP_{d,k}^{\ell}(\delta_{1/2})\in\cP(0,1)
  \end{align}
  exists and
  \begin{align}\label{eqBFE}
    \lim_{n\to\infty}\frac1n\log Z(\PHI)&=\fB_{d,k}(\pi_{d,k})&&\mbox{in probability}.
  \end{align}
\end{theorem}

The formula~\eqref{eqBFE} matches the prediction from~\cite{MZ} precisely.
Of course, part of the assertion of \Thm~\ref{thm_main} is that the Bethe free entropy $\fB_{d,k}(\pi_{d,k})$ is well defined.
Admittedly, the formula~\eqref{eqBFE} is not `explicit'.
But the proof of \Thm~\ref{thm_main} evinces that the convergence~\eqref{eqpidk} occurs rapidly.
Therefore, a randomised algorithm called `population dynamics'~\cite{MM} can be used to approximate~\eqref{eqBFE} within any desired numerical accuracy.

\begin{figure}
  \includegraphics[height=8.5cm]{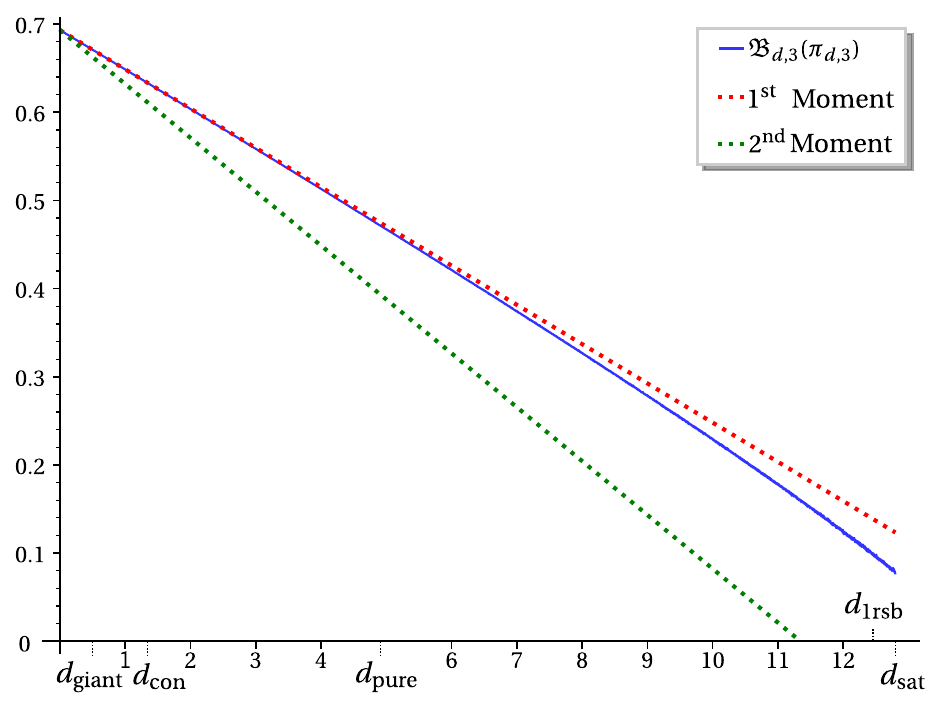}
  \caption{Comparison of $\fB_{d,k}(\pi_{d,k})$ with known bounds for $\lim_{n\to\infty}\frac1n\log Z(\PHI)$ for $k=3$.
    The red dotted line depicts the first moment upper bound~\eqref{eq1mmt}, while the green dotted line represents the lower bound provided by \eqref{eq2mmt}. 
    The blue line displays a numerical approximation of $\fB_{d,3}(\pi_{d,3})$. To obtain our values, 
    we generated $10^{6}$ samples from $\pi \approx \BP^{25}_{d,3}(\delta_{1/2})$ and then evaluated the
    corresponding empirical average of the expression in \eqref{eqBethe}. }\label{fig_maurice}
\end{figure}

\subsubsection{An improved lower bound on Gibbs uniqueness}\label{sec_results_imp_unq}
The obvious next task is to determine the Gibbs uniqueness threshold $\duniq(k)$.
Currently, its value is known precisely only in the case $k=2$, where $\duniq(2)=2$ coincides with the random 2-SAT satisfiability threshold~\cite{2sat,CR,Goerdt}.
Furthermore, Montanari and Shah~\cite{MS} proved that the {\em pure literal threshold}\footnote{This marks the threshold up to which the pure literal 
algorithm--which repeatedly assigns the preferred value to all variables appearing with a single sign--produces a satisfying assignment  \whp}
$\dpure(k)$ upper bounds $\duniq(k)$ for all $k\geq2$.%
\footnote{To be precise, Montanari and Shah established an upper bound on the Gibbs uniqueness threshold that turns out to coincide with the pure
  literal threshold, albeit without pointing out this identity. }
The value of $\dpure(k)$ admits a neat formula~\cite{BFU,Molloy}:
\begin{align}\label{eqpure}
  \duniq(k)\leq\dpure(k)&=\min_{z>0}\frac{z}{(1-\exp(-z/2))^{k-1}}.
\end{align}

Complementing the upper bound~\eqref{eqpure}, Montanari and Shah derived a lower bound $\dMS(k)$:
\begin{align}\label{eqMS}
  \dMS(k)&=\sup\cbc{d>0: d(k-1)\bc{1-\exp(-d/2)/4}\bc{1-\exp(-d/2)/2}^{k-2} <1}\leq\duniq(k).
\end{align}
Unfortunately, the bound \eqref{eqMS} is tight not even in the case $k=2$, where $\duniq(2)=2$ while $\dMS(2)\approx1.16$.
That said, the lower and upper bounds $\dMS(k)$ and $\dpure(k)$ match asymptotically in the limit of large $k$, as
\begin{align}\label{eqpureasymptotic}
  \dMS(k),\dpure(k)&=(2+o_k(1))\log k,
\end{align}
with $o_k(1)$ hiding a term that vanishes as $k\to\infty$.
The following theorem yields an improved lower bound $\dours(k)$ on $\duniq(k)$.

\begin{theorem}\label{thm_uniq}
  For all $k\geq3$ we have
  \begin{align}\label{eqmain}
    \duniq(k)\geq\dours(k):=\sup\cbc{d>0:\frac{d(k-1)}2\bc{1-\exp(-d/2)/2}^{k-2}<1}.
  \end{align}
\end{theorem}

\noindent
An easy calculation reveals that
\begin{align}\label{eqmainMS}
  \dMS(k)<\dours(k)&&\mbox{for every }k\geq2.
\end{align}
Moreover, it is satisfactory that the formula \eqref{eqmain} reproduces the correct (previously known) threshold $\duniq(2)=\dours(2)=\dpure(2)=2$.
That said, we have no reason to believe that \eqref{eqmain} is tight for any $k\geq3$.

\begin{table}[h]
  \begin{center}
    \begin{tabular}{|l||l|l|l|l|l|l|l} \hline
      $k$ & $2$ & $3$ & $4$ & $5$ 
      \\ \hline\hline
      $\dgiant$ & $1.0000$ & $0.5000$ & $0.3333$ & $0.2500$
      \\ \hline
      $\dMS$ & $1.1625$ & $0.8792$ & $0.8695$ & $0.9236$   
      \\ \hline
      $\dours$ & $2.0000$ & $1.3431$ & $1.2451$ & $1.2635$ 
      \\ \hline
      $\dpure$ & $2.0000$ & $4.9108$ & $6.1782$ & $7.0178$ 
      \\ \hline
      $\dsat$ & $2.0000$ & $12.801$ & $39.724$ & $105.585$ 
      \\ \hline
    \end{tabular}
  \end{center}
  \smallskip
  \caption{The values of $\dMS(k), \dours(k)$, and $\dpure(k)$ for $2\leq k\leq 5$.
    Additionally, $\dgiant(k) = 1/(k-1)$ marks the giant component threshold of the hypergraph induced by the random $k$-CNF formula. 
    Moreover, $\dsat(k)$ is the satisfiability threshold according to physics predictions~\cite{MMZ}. It is not hard to show that
  $\dgiant(k) \le \dMS(k) \le \dours(k) \le \duniq(k) \le \dpure(k) \le \dsat(k) $, for all $k \ge 2$.}
  \label{tab_main}
\end{table}

Combining \Thm s~\ref{thm_main} and~\ref{thm_uniq}, we obtain the following.

\begin{corollary}\label{cor_uniq}
  Let $k\geq3$.
  If $d<\dours(k)$ then~\eqref{eqBFE} holds.
\end{corollary}

\noindent
\Cor~\ref{cor_uniq} constitutes the first rigorous result to determine the precise asymptotic value of $\log Z(\PHI)$ for a non-trivial regime of $d$ for any $k\geq3$.
To elaborate, the formula~\eqref{eqBFE} is trivially true for $d<1/(k-1)$ because for such $d$ the $k$-uniform hypergraph induced by the clauses of $\PHI$ has no giant component and Belief Propagation is exact on acyclic graphical models~\cite{MM}.
But \Cor~\ref{cor_uniq} applies to $d$ well beyond this threshold, as displayed in Table~\ref{tab_main}.
In particular, in contrast to much of the prior work on random $k$-SAT, \Cor~\ref{cor_uniq} applies to a non-trivial regime of $d$ even for `small' $k\geq3$.

Although Table~\ref{tab_main} contains the values $\dMS(k)$ from~\cite{MS} for comparison, we emphasise that Montanari and Shah's result only yields the number of `good' assignments satisfying all but $o(n)$ clauses, rather than of actual satisfying assignments.
In fact, the best prior rigorous bounds on the number of satisfying assignments for $d>1/(k-1)$ derive from the first and the second moment methods.
Specifically, the folklore first moment bound reads
\begin{align}\label{eq1mmt}
  \frac1n\log Z(\PHI)&\leq\log 2+\frac dk\log(1-2^{-k})+o(1)\quad\mbox\whp
\end{align}
Furthermore, 
Achlioptas and Peres~\cite{yuval} perform a second moment argument on the number of {\em balanced} satisfying assignments, i.e., satisfying assignments that enjoy a peculiar additional condition required to keep the second moment under control.
They show that \whp
\begin{align}\label{eq2mmt}
  \frac1n\log Z(\PHI)&\geq(1-d)\log 2+\frac dk\log\brk{\bc{\lambda^{1/2}+\lambda^{-1/2}}^k-\lambda^{-k/2}}+o(1),&&\mbox{where}&(1-\lambda)(1+\lambda)^{k-1}=1,\,\lambda>0.
\end{align}

Figure~\ref{fig_maurice} illustrates the bounds~\eqref{eq1mmt}--\eqref{eq2mmt} along with \eqref{eqBFE} for $k=3$.
As the figure shows, the correct value \eqref{eqBFE} is quite close to the first moment bound.
That said, the first moment bound strictly exceeds $\fB_{d,k}(\pi_{d,k})$ for all $d>0$, $k\geq3$~\cite{CKM}.
On the other hand, Figure~\ref{fig_maurice} demonstrates that the `balanced second moment bound'~\eqref{eq2mmt} significantly undershoots $\fB_{d,3}(\pi_{d,3})$.
Recall that Figure~\ref{fig_maurice} is on a logarithmic scale; thus, even small differences translate into exponentially large errors.

\subsection{Preliminaries and notation}\label{sec_prelim}
Let $\Phi$ be a Boolean expression in conjunctive normal such that no clause contains the same variable twice.
We write $V(\Phi)$ for the set of Boolean variables of $\Phi$ and $F(\Phi)$ for the set of clauses.
The formula $\Phi$ gives rise to a bipartite graph $G(\Phi)$ on the vertex set $V(\Phi)\cup F(\Phi)$ in which a variable $x$ and a clause $a$ are adjacent iff variable $x$ appears in clause $a$ (either positively or negatively).
Let $E(\Phi)$ denote the edge set of the graph $G(\Phi)$.
Furthermore, for a vertex $v\in V(\Phi)\cup F(\Phi)$ let $\partial_\Phi v$ be the set of neighbours of $v$; where the reference to $\Phi$ is self-evident, we just write $\partial v$.

The graph $G(\Phi)$ induces a metric on $V(\Phi)\cup F(\Phi)$ by letting $\dist_\Phi(v,w)$ equal the length of the shortest path from $v$ to $w$.
For a vertex $v$ and an integer $\ell\geq0$ let $\partial^\ell_\Phi v=\partial^\ell v$ be the set of all vertices $w$ at distance precisely $\ell$ from $v$.

For a clause $a$ and a variable $x\in\partial a$ we define $\sign_\Phi(x,a)=1$ if $a$ contains $x$ as a positive literal, and $\sign_\Phi(x,a)=-1$ if $a$ contains the negation $\neg x$.
(This is unambiguous because clause $a$ is not allowed to contain both $x$ and $\neg x$.)
For a variable $x\in V(\Phi)$ and $s\in\PM$ we let $\partial^s_\Phi x=\partial^sx$ be the set of clauses $a\in\partial_\Phi x$ such that $\sign_\Phi(x,a)=s$.
Where convenient we use the shorthand $\partial^\pm x=\partial^{\pm1}x$.
We say that a variable $x$ is {\em pure} in $\Phi$ if $\sign_\Phi(x,a)=\sign_\Phi(x,b)$ for all $a,b\in\partial x$.
More specifically, say that $x$ is a {\em pure literal} of $\Phi$ if $\partial^-x=\emptyset$.
Similarly, $\neg x$ is called a pure literal if $\partial^+x=\emptyset$.
A variable or literal that fails to be pure is called {\em mixed}.

For a literal $l\in\{x,\neg x:x\in V(\Phi)\}$ we let $|l|$ denote the underlying variable; thus, $|x|=|\neg x|=x$ for $x\in V(\Phi)$.
Moreover, we define $\sign(x)=1$ and $\sign(\neg x)=-1$.
Further, for a literal $l$ we define $1\cdot l=l$ and $(-1)\cdot l=\neg l$.

If $\Phi$ is satisfiable, then $\SIGMA_\Phi=(\SIGMA_\Phi(x))_{x\in V(\Phi)}$ denotes a uniformly random satisfying assignment of $\Phi$.
Where the reference to $\Phi$ is obvious we just write $\SIGMA$.

Let $\mu,\nu$ be two probability measures on $\RR^h$, let $q\geq1$ and assume that $\int_{\RR^h}\|x\|_q^q\dd\mu(x),\int_{\RR^h}\|x\|_q^q\dd\nu(x)<\infty$.
We recall that the {\em $L_q$-Wasserstein distance} of $\mu,\nu$ is defined as
\begin{align*}
  W_q(\mu,\nu)&=\inf_{(\vec\xi,\vec\zeta)}\ex\brk{\|\vec\xi-\vec\zeta\|_q^q}^{1/q},
\end{align*}
where the infimum is taken over all pairs $(\vec\xi,\vec\zeta)$ of random variables defined on the same probability space $\Omega$ such that $\vec\xi$ has distribution $\mu$ and $\vec\zeta$ has distribution $\nu$.
If $\vX,\vY$ are random variables with distributions $\mu,\nu$, it is convenient to use the shorthand $W_q(\vX,\vY)=W_q(\mu,\nu)$, provided that $\ex[\|\vX\|_q^q],\ex[\|\vY\|_q^q]<\infty$.

For two random variables $\vX,\vY$ we write $\vX\disteq\vY$ if $\vX,\vY$ are identically distributed.
Moreover, for a probability distribution $\mu$ and a random variable $\vX$ we write $\vX\disteq\mu$ if $\vX$ has distribution $\mu$.

We will make repeated use of the following tail bound for Poisson variables.

\begin{lemma}[Bennett's inequality~{\cite[\Thm~2.9]{LugCon}}]\label{bennett}
  Suppose that $\vX\disteq\Po(\lambda)$ with $\lambda>0$ and let $\varphi(x)=(1+x)\log(1+x)-x$ for $x>-1$.
  Then
  \begin{align*}
    \pr\brk{\vX\geq\lambda+t}&\leq\exp(-\lambda\varphi(t/\lambda))&&\mbox{ for any }t>0,\\
    \pr\brk{\vX\leq\lambda-t}&\leq\exp(-\lambda\varphi(-t/\lambda))&&\mbox{ for any }0<t<\lambda.
  \end{align*}
\end{lemma}

For reals $a,b$ we write
\begin{align*}
  a\vee b&=\max\{a,b\},&a\wedge b&=\min\{a,b\}.
\end{align*}

Unless specified otherwise asymptotic notation $o(\nix),\,O(\nix)$, etc.\ is understood to refer to the limit $n\to\infty$.
The symbol $\tilde O(\nix)$ is understood to swallow $\mathrm{polylog}(n)$ terms.
Throughout we tacitly assume that $n$ is sufficiently large so that the various estimates are valid.
We use the conventions $\log 0=-\infty$ and $\log\infty=\infty$.
Finally, throughout the paper we assume that $k\geq3$ is a fixed integer.

\section{Overview}\label{sec_outline}

\noindent
In this section we survey the proofs of the main results.
Subsequently, we discuss further related work.
The proof details are deferred to the remaining sections; see \Sec~\ref{sec_org} for pointers.
We assume throughout that $k\geq3$.

\subsection{Existence of the fixed point and upper bound}\label{sec_de}
As a first step towards the proof of \Thm~\ref{thm_main} we prove that the limit~\eqref{eqpidk} exists for $d<\duniq(k)$.
More precisely, we will establish the following statement.

\begin{proposition}\label{prop_arnab}
  For every $k \ge 3$ and every $d < \duniq(k)$, the $W_1$-limit $\pi_{d,k}=\lim_{\ell\to\infty}\BP_{d,k}^{\ell}(\delta_{1/2})$ exists and
  \begin{align}\label{eq_prop_arnab_bound}
    \ex\brk{\log^2\MU_{\pi_{d,k},1,1}}+\ex\abs{\log\bc{\prod_{i=1}^{\vd^-}\MU_{\pi_{d,k},2i}+\prod_{i=1}^{\vd^+}\MU_{\pi_{d,k},2i-1}}}+\ex\abs{\log\bc{1-\prod_{j=1}^k\MU_{\pi_{d,k},1,j}}}&<\infty.
  \end{align}
  In addition, $\MU_{\pi_{d,k},1,1}$ and $1-\MU_{\pi_{d,k},1,1}$ are identically distributed.
\end{proposition}

The existence of the limit $\pi_{d,k}$ is an easy consequence of the Gibbs uniqueness property.
As an aside, the limit $\pi_{d,k}=\lim_{\ell\to\infty}\BP_{d,k}^{\ell}(\delta_{1/2})$ is a fixed point of the Belief Propagation operator, i.e.,
\begin{align}\label{eqBPfixedPoint}
  \pi_{d,k}&=\BP_{d,k}(\pi_{d,k}).
\end{align}
The proof of the bound~\eqref{eq_prop_arnab_bound} is a bit more subtle and requires a few preparations, but we will come to that.
The upshot of \eqref{eq_prop_arnab_bound} is that the Bethe free entropy $\fB_{d,k}(\pi_{d,k})$ is well defined.

With the fixed point $\pi_{d,k}$ in hand we can bring to bear the `interpolation method' to the upper bound the likely value of $\log Z(\PHI)$.

\begin{corollary}\label{cor_interpolation}
  If $d < \duniq(k)$  then  \whp\ we have $\frac1n\log Z(\PHI)\leq \fB_{d,k}(\pi_{d,k})+o(1).$
\end{corollary}

\noindent
The interpolation method is a mainstay of the study of disordered systems in mathematical physics and has also been used to investigate random constraint satisfaction problems.
In particular, the variant of the interpolation method from~\cite{PanchenkoTalagrand} (in combination with \Prop~\ref{prop_arnab}) easily implies that
\begin{align*}
  \limsup_{n\to\infty}\frac1n\ex\brk{\log(Z(\PHI)\vee1)}&\leq\fB_{d,k}(\pi_{d,k})\enspace;
\end{align*}
taking the logarithm of $Z(\PHI)\vee1$ ensures that the expectation is well defined, as it is possible (albeit unlikely for $d<\duniq(k)$) that $\PHI$ is unsatisfiable.
The added value of \Cor~\ref{cor_interpolation} is that we obtain a bound that holds {\em with high probability}, rather than just a bound on the expectation.
The interpolation method was used in~\cite{2sat} in a similar fashion to prove a `with high probability' bound on the number of satisfying assignments of random 2-CNFs.
The proof of \Cor~\ref{cor_interpolation} is an adaptation of that argument to $k\geq3$.

\subsection{A matching lower bound}\label{sec_lower}
The key step towards \Thm~\ref{thm_main} is to establish a lower bound on $\log Z(\PHI)$ that matches the upper bound from \Cor~\ref{cor_interpolation}.
To accomplish this task we employ a coupling argument known as the `Aizenman-Sims-Starr scheme' in mathematical physics.
Its original version was intended to estimate the partition function of the Sherrington-Kirkpatrick model, a spin glass model~\cite{Aizenman}.
But the technique has since been employed in probabilistic combinatorics (e.g.,~\cite{CKM,CKPZ,COP}).
By comparison to the mathematical physics context, the crucial difference is that here our objective is to count actual satisfying assignments where every single clause imposes a hard constraint, whereas in spin glass theory constraints are soft.
The same issue occurred in previous work on the random 2-SAT problem~\cite{2sat}.
However, in that case a relatively simple percolation argument was sufficient to deal with the ensuing complications.
As we will see, for $k\geq3$ considerably more care is needed.

But first things first.
The basic idea behind the Aizenman-Sims-Starr argument is to perform a kind of induction.
Translated to random $k$-SAT this means that we couple the random $k$-CNF $\PHI_{d,k}(n)$ with $n$ variables with the random $k$-CNF $\PHI_{d,k}(n+1)$ with $n+1$ variables.
Recall that $\PHI_{d,k}(n)$ comprises $\vm_n\disteq\Po(dn/k)$ independent random clauses.
Ultimately \Thm~\ref{thm_main} is going to be a consequence of \Cor~\ref{cor_interpolation} and the following statement.

\begin{proposition}\label{prop_aizenman}
  If $d<\duniq(k)$ then $\ex\brk{\log (Z(\PHI_{d,k}(n+1))\vee 1)}-\ex\brk{\log (Z(\PHI_{d,k}(n))\vee 1)}{=}\fB_{d,k}(\pi_{d,k})+o(1).$
\end{proposition}

\noindent
Once again we work with $Z(\PHI_{d,k}(n))\vee 1$ and $Z(\PHI_{d,k}(n+1))\vee 1$ to ensure that the expectations are well defined.

To prove \Prop~\ref{prop_aizenman} we couple the random formulas $\PHI_{d,k}(n+1)$ and $\PHI_{d,k}(n)$ as follows.
\begin{description}
  \item[CPL1] Let $\PHI'$ be a random $k$-CNF with variables $x_1,\ldots,x_n$ and $\vm'\disteq\Po(d(n-k+1)/k)$ clauses.
  \item[CPL2] Obtain $\PHI''$ from $\PHI'$ by adding another $\vec\Delta''\disteq\Po(d(k-1)/k)$ independent random clauses.
  \item[CPL3] Obtain $\PHI'''$ from $\PHI'$ by adding one new variable $x_{n+1}$ and $\vec\Delta'''\disteq\Po(d)$ independent random clauses that each contain $x_{n+1}$ and $k-1$ other variables from $\{x_1,\ldots,x_n\}$.
\end{description}
Observe that $\PHI''$ ultimately has variables $x_1,\ldots,x_n$ and a total of $\vm_n\disteq\Po(dn/k)$ random clauses.
Thus, $\PHI''$ is identical to the random formula $\PHI_{d,k}(n)$.
Similarly, $\PHI'''$ has the same distribution as $\PHI_{d,k}(n+1)$.
Consequently, we obtain the following.

\begin{fact}\label{fact_coupling}
  For any $d>0$ we have $Z(\PHI_{d,k}(n))\disteq Z(\PHI'')$ and $Z(\PHI_{d,k}(n+1))\disteq Z(\PHI''')$.
\end{fact}

The coupling {\bf CPL1--CPL3} reduces the proof of \Prop~\ref{prop_aizenman} to getting a handle on the differences $\log(Z(\PHI'')\vee 1)-\log(Z(\PHI')\vee 1)$ and $\log(Z(\PHI''')\vee 1)-\log(Z(\PHI')\vee 1)$.
More precisely, recalling~\eqref{eqhat}--\eqref{eqBethe}, we see that \Prop~\ref{prop_aizenman} is a consequence of the following two statements.

\begin{proposition}\label{lem_PHI''}
  If $d<\duniq(k)$ then
  \begin{align}\label{eqlem_PHI''}
    \ex\brk{\log\frac{Z(\PHI'')\vee 1}{Z(\PHI')\vee 1}}&=\frac{d(k-1)}{k}\ex\brk{\log\bc{1-\prod_{j=1}^k\MU_{\pi_{d,k},1,j}}}+o(1).
  \end{align}
\end{proposition}

\begin{proposition}\label{lem_PHI'''}
  If $d<\duniq(k)$ then
  \begin{align}\label{eqlem_PHI'''}
    \ex\brk{\log\frac{Z(\PHI''')\vee 1}{Z(\PHI')\vee 1}}&=\ex\brk{\log\bc{\prod_{i=1}^{\vd^-}\MU_{\pi_{d,k},2i}+\prod_{i=1}^{\vd^+}\MU_{\pi_{d,k},2i-1}}}+o(1).
  \end{align}
\end{proposition}

To prove \Prop s~\ref{lem_PHI''}--\ref{lem_PHI'''} we effectively need to trace the impact that \emph{local} changes have on the number of satisfying assignments.
Indeed, under the coupling {\bf CPL1--CPL3}, the formula $\PHI''$ is obtained from the `base formula' $\PHI'$ by adding just a bounded expected number of random clauses.
Thus, if we imagine that, as both the first moment upper bound~\eqref{eq1mmt} and the balanced second moment lower bound~\eqref{eq2mmt} suggest, each additional random clause typically reduces the number of satisfying assignments by a constant factor, then the quantity $|\log(Z(\PHI'')/Z(\PHI'))|$ should be bounded with probability close to one.
Similar reasoning applies to $\PHI'''$.

Yet while with high probability the local changes that turn $\PHI'$ into $\PHI''$ or $\PHI'''$ are indeed benign, because we are dealing with hard constraints there is a non-negligible probability that $\log(Z(\PHI'')/Z(\PHI'))$ and $\log(Z(\PHI''')/Z(\PHI'))$ could be large.
Indeed, a single extra clause might wipe out all satisfying assignments of $\PHI'$, in which case
\begin{align*}
  \log\frac{Z(\PHI'')\vee 1}{Z(\PHI') \vee 1}=-\log Z(\PHI')=-\Omega(n).
\end{align*}
Hence, we need to argue that such drastic changes are sufficiently rare.
The following statement furnishes the necessary tail bound.

\begin{proposition}\label{prop_bad}
  For $d<\duniq(k)$ we have
  \begin{align}\label{eq_prop_bad}
    \ex\brk{\left|\log\frac{Z(\PHI'')\vee 1}{Z(\PHI')\vee 1}\right|^{3/2}+\left|\log\frac{Z(\PHI''')\vee 1}{Z(\PHI')\vee1}\right|^{3/2}}=O(1).
  \end{align}
\end{proposition}

\subsection{Pure literal pursuit}\label{sec_pulp}
The proof of \Prop~\ref{prop_bad} constitutes the main technical challenge towards the proof of \Thm~\ref{thm_main}.
The linchpin of the proof is an algorithm that we call {\em Pure Literal Pursuit} (`\pulp').
Its purpose is to trace the repercussions of setting a relatively small number of variables to specific truth values.
More precisely, \pulp\ will allow us to compare the number of satisfying assignments that set a few chosen variables to specific values to the total number of satisfying assignments.

To this end \pulp\ attempts to solve the following optimisation task.
Suppose we are given a $k$-CNF $\Phi$ and a set $\cL$ of literals of $\Phi$ that we deem to be set to `true'.
We would like to identify a superset $\bar\cL\supseteq\cL$ of literals with the following properties; think of $\bar\cL$ as a `closure' of $\cL$.
\begin{description}
  \item[PULP1] every clause $a$ of $\Phi$ that contains a literal from $\neg\bar\cL=\{\neg l:l\in\bar\cL\}$ also contains a literal from $\bar\cL$.
  \item[PULP2] there is no literal $l$ such that $l,\neg l\in\bar\cL$.
\end{description}
Of course, it may be impossible to satisfy {\bf PULP1} and {\bf PULP2} simultaneously.
In this case we ask \pulp\ to report a `contradiction'.
But if {\bf PULP1}--{\bf PULP2} can be satisfied, we aim to find a closure $\bar\cL$ of as small size $|\bar\cL|$ as possible.

The combinatorial idea behind {\bf PULP1}--{\bf PULP2} is as follows.
Deeming the literals from the initial set $\cL$ `true', our goal is to reconcile this assumption with the formula $\Phi$.
To this end we enhance the set $\cL$.
Clearly, any clause that contains the negation $\neg l$ of a literal $l$ that we deem true also needs to contain another literal $l'$ that is set to true.
This is what {\bf PULP1} asks.
Furthermore, it would be contradictory to deem both $l$ and its negation $\neg l$ true; this is {\bf PULP2}.

The size of the closure $\bar\cL$ yields a bound on the reduction in the number of satisfying assignments if we indeed insist on all literals $l\in\cL$ being set to true.
Formally, let $S(\Phi,\cL)$ be the set of all satisfying assignments $\sigma\in S(\Phi)$ under which all literals $l\in\cL$ evaluate to `true'.
Also set $Z(\Phi,\cL)=|S(\Phi,\cL)|$.

\begin{lemma}\label{lem_Ichi}
  For any $\Phi,\cL$ and any $\bar\cL \supseteq \cL$ that satisfies {\bf PULP1}--{\bf PULP2} we have $Z(\Phi)\leq 2^{|\bar\cL|}Z(\Phi,\cL)$.
\end{lemma}

In order to identify a `small' closure $\bar\cL$ the \pulp\ algorithm resorts to {\em pure literal elimination}, a simple trick commonplace to satisfiability 
algorithms.
A variable $x$ is {\em pure} in a CNF formula $\Phi$ if $\sign(x,a)=\sign(x,b)$ for any two clauses $a,b\in\partial x$.
Clearly, if our objective is to construct a satisfying assignment, we might as well set all pure variables $x$ to the value that satisfies all clauses $a\in\partial x$ and disregard these clauses henceforth.
In light of this observation, pure literal elimination repeatedly removes all clauses that contain a pure variable.
Naturally, every round of clause removals may create new pure variables, and thus more clauses may be ripe for removal in the next round.
For a clause $a$ of the original formula $\Phi$ let $\fn_a(\Phi)\geq1$ be the number of the round at which pure literal elimination removes $a$.
If $a$ is never removed then we set $\fn_a(\Phi)=\infty$.

The \pulp\ algorithm invokes a slightly modified version of pure literal elimination to accommodate the initial set $\cL$ of literals.
Specifically, for a variable $x$ of a CNF $\Phi$ and $s\in\PM$ let $\Phi[x\mapsto s]$ be the CNF obtained by removing all clauses
$a\in\partial x$ with $\sign(x,a)=s$ and removing the literal $-s\cdot x$ from all $a\in\partial x$ with $\sign(x,a)=-s$.
The definition reflects that if we set $x$ to value $s$, all $a\in\partial^s x$ will be satisfied, while all $a\in\partial^{-s} x$ will have to be satisfied by one of their other constituent literals.
Further, let
\begin{align}\label{eqheightdef}
  \fn_x(s,\Phi)&=
  \begin{cases}
    0&\mbox{ if }\partial_\Phi^{-s} x=\emptyset,\\
    \max\cbc{\fn_a(\Phi[x\mapsto s]):a\in\partial_\Phi^{-s} x}&\mbox{ otherwise.}
  \end{cases}
  \qquad	\in[0,\infty].
\end{align}
We refer to $\fn_x(s,\Phi)$ as the {\em height of literal $s\cdot x$ in $\Phi$}.

The \pulp\ algorithm, displayed as Algorithm~\ref{fig_pulp}, harnesses the heights as follows.
In its attempt to precipitate {\bf PULP1} and {\bf PULP2} the algorithm iteratively enhances the set $\cL$ of literals deemed to be `true'.
For any clause $a$ that violates {\bf PULP1} and that contains a literal $l\not\in\neg\cL$ the algorithm adds one such literal $l$ {\em of minimum height} to $\cL$.
This choice is intended to keep the ultimate size of the closure small; one could say that \pulp\ uses height as a proxy of `size'.
If at any point the algorithm encounters a clause $a$ that consists of literals from $\neg\cL$ only, the algorithm reports a contradiction and aborts.

\begin{algorithm}[h!]\label{alg_ws}
  \KwData{A $k$-CNF $\Phi$ and a set $\cL$ of literals.}
  Let $\bar\cL=\cL$\;
  \While{there is a clause $a$ that contains a literal from $\neg\bar\cL$ but no literal from $\bar\cL$}
  {Pick such a clause $a$ that minimises the distance from the initial set $\cL=\{|l|:l\in\cL\}$\;
    \If{$a$ consists of literals $l\in\neg\bar\cL$ only}
    {return `contradiction' and halt\;}
    \Else{choose $x\in\partial a$ with $x,\neg x\not\in\bar\cL$ that minimises $\fn_{x}(\sign(x,a),\Phi)$ and add $\sign(x,a)\cdot x$ to $\bar\cL$ \;}
  }
  \Return\ $\bar\cL$
  \caption{The \pulp\ algorithm}\label{fig_pulp}
\end{algorithm}

\begin{remark}\label{rem_ws}
To break ties that may occur in the execution of Steps~3 and~7 of \pulp\ we assume that the variables and clauses of $\Phi$ are numbered so that Steps~3 and~7 can choose the clause/variable with the smallest number that satisfies the respective requirements.
In due course we will run \pulp\ on (finite subtrees of) the Galton-Watson tree $\TT$.
To number the variables and clauses of $\TT$ we equip each of them with an independent Gaussian label.
Since $\TT$ comprises a countable number of clauses/variables, these labels will almost surely be pairwise distinct.
\end{remark}

From here on we write $\bar\cL$ for the set of literals returned by \pulp\ if the algorithm does not encounter a contradiction;
in the event of a a contradiction we let $\bar\cL=\{x,\neg x:x\in V(\Phi)\}$ be the set of all literals.
Where the reference to the formula $\Phi$ is not entirely obvious, we write $\bar\cL_{\Phi}$.
The analysis of \pulp\ on the random formula $\PHI'$ furnishes the following bound on $|\bar\cL|$ in terms of the size of the initial set $\cL$.
This bound is the key ingredient towards the proof of \Prop~\ref{prop_bad}.

\begin{lemma}\label{lem_EIchi}
  There exists $C=C(d,k)>0$ such that the following is true.
  Let $\cL$ be a set of literals of $\PHI'$ such that $1\leq|\cL|\leq\log^2n$ and such that $\{x_i,\neg x_i\}\not\subseteq\cL$ for all $1\leq i\leq n$.
  Then $\Erw[|\bar\cL|^{3/2}]\leq C|\cL|^{3/2}$.
\end{lemma}

The proof of \Lem~\ref{lem_EIchi} is one of the main technical challenges of the present work.
The difficulty stems from the stochastic dependencies that are inherent to the \pulp\ algorithm.
Specifically, in order to decide which literals to add to the set $\cL$, \pulp\ requires knowledge of the heights $\fn_{x}(\pm1,\PHI')$.
But these heights depend on the other variables $y\in\partial a\setminus\cbc x$, the clauses that these variables $y$ appear in, etc.
Furthermore, in its subsequent iterations the algorithm is apt to revisit some of these variables and clauses at a point when their heights have already been revealed.
These repetitions rule out an analysis of \pulp\ by way of routine techniques such as the principle of deferred decisions or the differential equations method.
The reason why we manage to cope with these complicated dependencies at all is that, remarkably, the heights $\fn_{x}(\pm1,\PHI')$ have only a tiny upper tail.
More precisely, as we will see the tails of these random variables decay at a {\em doubly} exponential rate.

\Prop~\ref{prop_bad} follows from the analysis of \pulp.
The basic idea is to apply the algorithm to an initial set $\cL$ of literals that contain one literal from each of the extra clauses that are present in $\PHI''$ or $\PHI'''$ but not in $\PHI'$.
With a bit of care the bounds from \Lem s~\ref{lem_Ichi} and~\ref{lem_EIchi} then imply~\eqref{eq_prop_bad}.
Finally, the analysis of \pulp\ that leads up to the proof of \Lem~\ref{lem_EIchi} also implies the necessary tail bounds to verify the bounds from~\eqref{eq_prop_arnab_bound}.
Specifically, the proof of \Lem~\ref{lem_EIchi} proceeds by way of analysing \pulp\ on the Galton-Watson tree $\TT_{d,k}$, and the bounds \eqref{eq_prop_arnab_bound} come out as a byproduct of that analysis.

\subsection{Completing the Aizenman-Sims-Starr scheme}\label{sec_aizenman_complete}
To obtain \Prop s~\ref{lem_PHI''}--\ref{lem_PHI'''} we combine \Prop~\ref{prop_bad} with an analysis of the quotients $Z(\PHI'')/Z(\PHI')$ and $Z(\PHI''')/Z(\PHI')$ on a likely `good' event.
On this good event the empirical distribution of the marginal probabilities $(\pr[\SIGMA_{\PHI'}(x_i)=1\mid\PHI'])_{1\leq i\leq n}$ of the different variables $x_i$ receiving the value `true' under a random satisfying assignment is `close' to the limiting distribution $\pi_{d,k}$ from \Prop~\ref{prop_arnab}.
Additionally, on the good event the joint distribution of the truth values assigned to a moderate number of variables is well approximated by a product measure.
Of course, to make this precise we need to investigate the empirical distribution
\begin{align}\label{eqempirical}
  \vec\pi_n'&=\frac1n\sum_{i=1}^n\delta_{\pr[\SIGMA_{\PHI'}(x_i)=1\mid\PHI']}\in\cP(0,1)
\end{align}
of the marginals $(\pr[\SIGMA_{\PHI'}(x_i)=1\mid\PHI'])_{1\leq i\leq n}$.

\begin{proposition}\label{prop_benign}
  Assume that $d<\duniq(k)$.
  Then $\ex\brk{W_1(\vec\pi_n',\pi_{d,k})}=o(1)$
  and for any $\ell = O(1)$ we have
  \begin{align*}
    \sum_{\sigma\in\PM^\ell}\ex\abs{\pr\brk{\forall 1\leq i\leq\ell:\SIGMA_{\PHI'}(x_i)=\sigma_i\mid\PHI'}-\prod_{i=1}^\ell\pr\brk{\SIGMA_{\PHI'}(x_i)=\sigma_i\mid\PHI'}}&=o(1).
  \end{align*}
\end{proposition}

The proof of \Prop~\ref{prop_benign} hinges on the Gibbs uniqueness property and the convergence of the local topology of the random formula $\PHI'$ to the Galton-Watson tree $\TT_{d,k}$.
Together with careful coupling arguments \Prop s~\ref{prop_bad}--\ref{prop_benign} imply \Prop s~\ref{lem_PHI''}--\ref{lem_PHI'''}.
Moreover, in combination with Fact~\ref{fact_coupling} these two propositions yield \Prop~\ref{prop_aizenman}.
We complete this paragraph by showing how \Thm~\ref{thm_main} follows from \Cor~\ref{cor_interpolation} and \Prop~\ref{prop_aizenman}.

\begin{proof}[Proof of \Thm~\ref{thm_main}]
  The existence of the limit~\eqref{eqpidk} follows from \Prop~\ref{prop_arnab}.
  With respect to~\eqref{eqBFE}, we 
  apply \Prop~\ref{prop_aizenman} to obtain 
  \begin{align}\nonumber
    \frac{1}{n}\Erw\brk{\log(1 \vee Z(\PHI_{d,k}(n)))}&=\frac{1}{n}\sum_{N=0}^{n-1} \left(\Erw\brk{\log(1 \vee Z(\PHI_{d,k}(N+1))} - \Erw\brk{\log(1 \vee Z(\PHI_{d,k}(N))}\right)\\
    &= \fB_{d,k}(\pi_{d,k}) + o(1)
    \enspace.\label{eq_exp_tele}
  \end{align}
  Since, conversely, \Cor~\ref{cor_interpolation} shows that $\frac1n\log Z(\PHI)\leq \fB_{d,k}(\pi_{d,k})+o(1)$ \whp\ and since $\log Z(\PHI)\leq n\log 2$ deterministically, the assertion follows from~\eqref{eq_exp_tele}.
\end{proof}

\subsection{Lower-bounding the uniqueness threshold}\label{sec_uniq_outline}
The proof of \Thm~\ref{thm_uniq} combines three ingredients.
From the work~\cite{2sat} on random 2-SAT we borrow the idea of constructing an explicit extremal boundary configuration.
In effect, in order to prove Gibbs uniqueness we just have to consider one single boundary configuration, instead of an enormous number of possible configurations $\tau$ that grows quickly with the height $\ell$ as in the original definition~\eqref{eqTreeUniq}.
Second, from the work~\cite{MS} of Montanari and Shah we borrow the idea of expressly considering the effect of pure literals.
As it turns out, without explicit consideration of pure literals it seems difficult to even recover the correct asymptotic order~\eqref{eqpure} of the Gibbs uniqueness threshold.
Third, and most importantly, the improvement over the bound from~\cite{MS} stems from a new subtle coupling argument that we will explain in due course.

\subsubsection{The extremal boundary condition}

An obvious challenge associated with establishing the Gibbs uniqueness property~\eqref{eqTreeUniq} seems to be that we need to estimate the marginal of the root variable given {\em any} possible boundary condition, i.e., given any assignment of the variables at distance $2\ell$ from $\root$.
As we expect to see $(d  (k-1))^\ell$ variables at distance $2\ell$ from $\root$, we thus face a doubly exponential number $2^{(d  (k-1))^\ell}$ of possible boundary conditions.
But fortunately, following~\cite{2sat} we may confine ourselves to just a single, explicit boundary configuration $\TAU^+$ that satisfies
\begin{align}
  \pr&\brk{\TAU^{(\ell)}(r)=1\mid\TT,\,\forall x\in\partial^{2\ell}r:\TAU^{(\ell)}(x)=\TAU^+(x)}
  =
  \max_{\tau\in S(\TT^{(\ell)})} \pr\brk{\TAU^{(\ell)}(r)=1\mid\TT,\,\forall x\in\partial^{2\ell}r:\TAU^{(\ell)}(x)=\tau(x)}.
  \label{eqfunnydiff}
\end{align}
Due to the inherent symmetry of the distribution of $\TT$ with respect to the signs of the clauses,
towards the proof of~\eqref{eqTreeUniq} it is sufficient to show that the difference~\eqref{eqfunnydiff} vanishes as $\ell\to\infty$.

The extremal boundary condition can be constructed explicitly.
Specifically, given ${\TT}^{(\ell)}$ we construct a satisfying assignment $\TAU^+\in S(\TT^{(\ell)})$ by working our way down the tree ${\TT}^{(\ell)}$.
We begin by setting $\TAU^+(\root) = 1$.
Now suppose that for $q \ge 1$, the values of the variables at distance $2(q-1)$ from $\root$ have been already determined.
Let $w$ be a variable at distance ${2q}$ from $\root$ with parent clause $a$ and grandparent variable $u$.
Then we define
\begin{align}\label{eq_extremalBoundDef}
  \TAU^+(w) = \sign(w,a)\cdot\Ind \{\sign(u,a)\neq\TAU^+(u)\}
  -\sign(w,a)\cdot\Ind \{\sign(u,a)=\TAU^+(u)\}
  \enspace.
\end{align}
The idea behind~\eqref{eq_extremalBoundDef} is for $\TAU^+(w)$ to ``nudge'' $u$ towards $\TAU^+(u)$ 
by making sure that $w$ satisfies clause $a$ if setting $u$ to $\TAU^+(u)$ does not, 
and conversely making sure that $w$ fails to satisfy clause $a$ if setting $u$ to $\TAU^+(u)$ does.
A simple induction on $\ell$ shows that $\TAU^+$ is a satisfying assignment for which~\eqref{eqfunnydiff} holds.

\begin{lemma}\label{lem:ExtremalBnd}
  For any integer $\ell\geq0$ the assignment $\TAU^+$ defined via~\eqref{eq_extremalBoundDef} satisfies~\eqref{eqfunnydiff}.
\end{lemma}

\noindent
Hence, proving \Thm~\ref{thm_uniq} reduces to establishing the following.

\begin{proposition}\label{prop:condMarg}
  For $d < \dours(k)$ we have that
  \begin{align}\label{eqprop:condMarg}
    \lim_{\ell\to\infty}\Erw\brk{\pr\brk{\TAU^{(\ell)}(\root)=1\mid\TT,\,\forall x\in\partial^{2\ell}\root:\TAU^{(\ell)}(x)=\TAU^+(x)}-\pr\brk{\TAU^{(\ell)}(\root)=1\mid\TT}}=0.
  \end{align}
\end{proposition}

The proof of Proposition~\ref{prop:condMarg} may seem delicate because the boundary condition $\TAU^+$ depends on the tree ${\TT}^{(\ell)}$.
To sidestep this problem, we generalise another idea from the work \cite{2sat} on random $2$-SAT to $k\geq3$ by introducing a quantity that allows us to prove~\eqref{prop:condMarg} but that behaves `Markovian' as we pass up and down the tree.
Specifically, for a variable $x$ of ${\TT}^{(\ell)}$ let $\TT_x^{(\ell)}$ be the sub-formula of ${\TT}^{(\ell)}$ comprising $x$ and its progeny.
Moreover, for a satisfying assignment $\tau\in S({\TT}^{(\ell)})$ let
\begin{align*}
  S({\TT}_x^{(\ell)},\tau)
  &=
  \cbc{\chi\in S({\TT}_x^{(\ell)}):
    \forall y\in V({\TT}_x^{(\ell)})\cap\partial^{2\ell}_{\TT}\root :
    \chi_y=\tau_y } \enspace,
  &
  Z({\TT}_x^{(\ell)},\tau)
  &=\abs{S({\TT}_x^{(\ell)},\tau)} \enspace.
\end{align*}
In words, $S({\TT}_x^{(\ell)},\tau)$ contains all satisfying assignments of ${\TT}_x^{(\ell)}$ that comply with the boundary condition induced by $\tau$.
Additionally, for $t=\pm1$ let
\begin{align*}
  S({\TT}_x^{(\ell)},\tau,t) &=\cbc{\chi\in S({\TT}_x^{(\ell)},\tau):\chi_x=t},&
  Z({\TT}_x^{(\ell)},\tau,t)&=\abs{S({\TT}_x^{(\ell)},\tau,t) }
\end{align*}
be the set and number of satisfying assignments of ${\TT}_x^{(\ell)}$ that agree with $\tau$ on the boundary and assign value $t$ to $x$.
Finally, let
\begin{align}\label{eta}
  \ETA_{x}^{(\ell)}
  &=
  \log\frac{Z({\TT}_x^{(\ell)},\TAU^+,\TAU^+(x))}
  {Z({\TT}_x^{(\ell)},\TAU^+,-\TAU^+(x))}
  \in\RR\cup\{\pm\infty\}
\end{align} %
be the log-likelihood ratio that gauges how likely a random satisfying assignment $\TAU$ of ${\TT}_x^{(\ell)}$ subject to the $\TAU^+$-boundary condition is to set $x$ to its designated value $\TAU^+(x)$ from \eqref{eq_extremalBoundDef}.
In terms of~\eqref{eta}, 
the proof of \Prop~\ref{prop:condMarg} comes down to showing that for $d<\dours(k)$,
\begin{align}\label{eqstrangeconv}
  \lim_{\ell\to\infty}\ETA_\root^{(\ell)}&=\log\left(\frac{\MU_{\pi_{d,k},1,1}}{1-\MU_{\pi_{d,k},1,1}}\right)&&\mbox{in distribution}.
\end{align}

For a start, the following lemma bounds the tails of  $\ETA_{x}^{(\ell)} $ for large enough $\ell$ and $x$ reasonably close to the root variable $\root$.

\begin{lemma}\label{lem_SmallOnTop}
  For every  $0<d< \dours(k)$ there exist $c=c(d,k)$ and a sequence $(\varepsilon_t)_t$ with $\lim_{t \to \infty}\varepsilon_t = 0$ such that for any $t>0$, $\ell > ct^c$ we have
  \begin{align}\label{eq_lem_SmallOnTop}
    \Pr \left[\max_{x \in \partial^{2t}\root}
    \;\left|\ETA_x^{(\ell)}\right| \le 2t^c\right]
    > 1 - \varepsilon_t
    \enspace.
  \end{align}
\end{lemma}

\noindent
The proof of \Lem~\ref{lem_SmallOnTop} rests on combinatorial arguments reminiscent of the analysis of \pulp.

A key feature of the definition~\eqref{eta} is that the random variables $\vec\eta_x^{(\ell)}$ exhibit a `reverse Markovian' behaviour.
This is because $\vec\eta_x^{(\ell)}$ depends only on $\TAU^+(x)$ and the part $\TT_x^{(\ell)}$ of the tree pending on $x$.
Furthermore, because the distribution of the random tree $\TT_x^{(\ell)}$ is symmetric with respect to sign flips, even the dependence on the value $\TAU^+(x)$ can be eliminated.
All we need to keep in mind is that the values $\TAU^+(y)$ for $y\in V(\TT_x^{(\ell)})$ are constructed from the value $\TAU^+(x)$ in accordance with the recurrence~\eqref{eq_extremalBoundDef}.
Thus, by flipping all signs in the tree $\TT_x^{(\ell)}$ if necessary, we could assume without loss that $\TAU^+(x)=1$ without changing the distribution of $\vec\eta_x^{(\ell)}$ with respect to the randomness of $\TT_x^{(\ell)}$.
As a consequence, it is possible to set up a recurrence that expresses the log-likelihood ratios $\vec\eta_x^{(\ell)}$ of variables $x$ at distance $q$ from $\root$ in terms of the $\vec\eta_y^{(\ell)}$ for $y$ at distance $q+2$ from $\root$.

Due to the recursive nature of the random tree $\TT$, it suffices to set up this recurrence for the root $\root$ of the tree.
In other words, to prove \eqref{eqstrangeconv} we just need a recurrence that expresses the distribution of the random variable $\vec\eta_{\root}^{(\ell+1)}$
in terms of the law of $\vec\eta_{\root}^{(\ell)}$ for $\ell\geq0$.
A bit of reflection (see \Cl~\ref{Lemma_Noela}), reveals that the corresponding distributional operator
\begin{align*}
  \LDEP_{d,k}&:\cP((-\infty, \infty])\to\cP((-\infty, \infty])\enspace,&\rho\mapsto\hat\rho=\LDEP_{d,k}(\rho)
\end{align*}
has the following shape.
For a distribution $\rho\in\cP((-\infty, \infty])$ let $(\vec\eta_{\rho,i,j})_{i,j\geq1}$ be a family of random variables with distribution $\rho$.
Moreover, let $(\vs_i)_{i\geq1}$ be a sequence of uniformly random $\pm1$-valued random variables and let $\vd\disteq\Po(d)$.
All of these random variables are mutually independent.
Additionally, for $q\geq0$ and $z_1,\ldots,z_q\in\RR\cup\{\pm\infty\}$ define
\begin{align}
  \Pfun(z_1,\ldots,z_q)&= \prod_{i=1}^{q} \frac{1 +\tanh(z_i/2)} {2} \enspace .\label{eq_ProdMacro}
\end{align}

Then $\hat\rho=\LDEP_{d,k}(\rho)$ is the distribution of the random variable
\begin{align}\label{eq:NudgedFixedPoint}
  -\sum_{i=1}^{\vd}&
  \vs_i
  \cdot
  \log\left(
  1-
  {\Pfun\left(\vs_i\cdot
    \bigl(
    {\ETA_{\rho, i, 1}}, \ldots, {\ETA_{\rho, i, k-1}}
    \bigr)
    \right)}
  \right)\enspace.
\end{align}

Ultimately we will derive \eqref{eqstrangeconv}, and thereby \Prop~\ref{prop:condMarg}, from \Lem~\ref{lem_SmallOnTop} and a contraction argument.
However, this is not quite as straightforward as one might be inclined to expect.
Indeed, at first glance, a natural approach to proving \eqref{eqstrangeconv} from \Lem~\ref{lem_SmallOnTop} seems to be to show that $\LDEP_{d,k}$ is a contraction, say, with respect to the $W_1$-metric.
This is indeed carried out in~\cite{2sat} for $k=2$, where it is shown that $\LDEP_{d,2}$ contracts for all $0<d<2$, i.e., right up to the random 2-SAT satisfiability threshold.
However, for $k\geq3$ we can only show that $\LDEP_{d,k}$ contracts for $d<2/(k-1)$, a value well below $\dours(k)$ and short of the correct asymptotic order~\eqref{eqpureasymptotic}.

\subsubsection{Pure and mixed literals}
To cover a larger range of $d$ we borrow from~\cite{MS} the idea of expressly taking into account pure literals.
To elaborate, while $\LDEP_{d,k}$ describes how the law of $\vec\eta_\root^{(\ell+1)}$ results from that of $\vec\eta_\root^{(\ell)}$, the operator fails to take into account that $\root$ itself as well as some of the grandchildren of $\root$ in $\TT$ may be pure literals.
However, the pure literal property has a marked effect on the log-likelihood ratios.
For if, say, $\root$ only appears positively, then a simple double counting argument shows that $\vec\eta_\root^{(\ell)}\geq0$ for all $\ell$.
By extension, pure literals among the grandchildren of $\root$ have a `dampening' effect and may thus improve the range of $d$ for which we can establish contraction.

For a variable node $x$ of $\TT$, let us denote by $\TT_x$ the subtree of $\TT$ rooted at $x$ and 
containing its progeny. Leveraging the above observation, we classify a variable $x$ of $\TT$ as $\rAllKid$, $\rPureP$, $\rPureM$, or $\rNoKid$, depending on whether $x$ appears 
both positively and negatively in $\TT_x$, only positively, only negatively, or whether $x$ has no children at all, respectively.
Furthermore, instead of just tracing the law of $\vec\eta_\root^{(\ell)}$ for $\ell\geq0$, we study the four separate conditional distributions 
given the type $\rAllKid,\rPureP,\rPureM$ or $\rNoKid$ of $\root$. Of course, the distribution of $\vec\eta_\root^{(\ell)}$ given type $\rNoKid$ 
(i.e., $\root$ has no children) is just the atom at zero for all $\ell$.

To describe the evolution of the other distributions we introduce the operator
\begin{align}\nonumber
  \LDELitS{d,k}:
  \cP  (-\infty,\infty] \times \cP (0, +\infty] &\times \cP(-\infty, 0] 
  \to
  \cP (-\infty,\infty] \times  \cP(0, +\infty] \times  \cP(-\infty, 0] 
  \enspace,\\
  \text{ with }
  (\rho_{\AllKid}, \rho_{\PureP}, \rho_{\PureM})& \mapsto
  (\hat\rho_{\AllKid}, \hat\rho_{\PureP}, \hat\rho_{\PureM})= \LDELitS{d,k}(\rho_{\AllKid}, \rho_{\PureP}, \rho_{\PureM})
  \label{eqLDELitS}
\end{align}
defined as follows.
Let $\vdpc, {\vdpc}^{\prime}, \vdmc, {\vdmc}^{\prime}$ be Poisson variables with parameter $d/2$, conditioned on being positive.
Moreover, let
$ \vec{r}_1=\left(\vec{r}_{\AllKid,1}, \vec{r}_{\PureP, 1}, \vec{r}_{\PureM,1}, \vec{r}_{\NoKid,1}\right), \; \vec{r}_2=\left(\vec{r}_{\AllKid,2}, \vec{r}_{\PureP,2}, \vec{r}_{\PureM,2}, \vec{r}_{\NoKid,2}\right), \ldots$
be multinomial variables with ${k-1}$ trials and probabilities
\begin{align}\label{eq_def_multipr}
  p_{\AllKid} = (1-e^{-d/2})^2, \enspace \quad
  p_{\PureP} = p_{\PureM}=e^{-d/2}(1-e^{-d/2}), \quad\enspace
  p_{\NoKid} = e^{-d}
  \enspace.
\end{align}
For $i,j \ge 1$ let $\ETA_{\AllKid,i,j}$, $\ETA_{\PureP,i,j}$, $\ETA_{\PureM,i,j}$ be random variables with law $\rho_{\AllKid}$, $\rho_{\PureP}$, $\rho_{\PureM}$, respectively.
All of the aforementioned random variables are mutually independent.
Further, for a sign $\varepsilon \in \{\pm 1\}$ and a vector $r=(r_{\AllKid}, r_{\PureP}, r_{\PureM}, r_{\NoKid})$ 
of non-negative integers with $r_{\AllKid} + r_{\PureP} + r_{\PureM} + r_{\NoKid} = {k-1}$ and $i\geq0$, $1\leq j\leq4$ we let
\begin{align}\label{eq:defOfXI}
  \LArg_{i,j}(\eps,r)=  1 -  \frac{1}{2^{{r}_{\NoKid}}} \cdot
  \Pfun
  \left(\varepsilon \bigl( {\ETA_{\AllKid,4i+j,1}},\ldots,{\ETA_{\AllKid,4i+j,r_{\AllKid}}}\bigr) \right)
  \Pfun
  \left(\varepsilon\bigl( {\ETA_{\PureP,4i+j,1}},\ldots,{\ETA_{\PureP,4i+j,r_{\PureP}}}\bigr) \right)
  \Pfun\left(\eps
  \bigl( {\ETA_{\PureM,4i+j,1}},\ldots,{\ETA_{\PureM,4i+j,r_{\PureM}}}\bigr)
  \right)
  .
\end{align}
The r.h.s.\ of~\eqref{eq:defOfXI} amounts to rewriting the argument of the logarithm in \eqref{eq:NudgedFixedPoint} when the number of variables of each type is distributed according to $r$.
Finally, let
\begin{align}\label{eq:defOfXIWithoutr}
  \LArg_{i,j}=\LArg_{i,j}\bc{(-1)^{j+1},\vec r_{4i+j}}
  \enspace.
\end{align}
Then the operator~\eqref{eqLDELitS} maps $\rho_{\AllKid}, \rho_{\PureP}, \rho_{\PureM}$ to the distributions  $\hat\rho_{\AllKid}, \hat\rho_{\PureP}, \hat\rho_{\PureM}$ of the random variables 
\begin{align}
  \hat\rho_{\AllKid} &\disteq -\sum_{i =1}^{\vdpc}  \log\,\LArg_{i,1}+\sum_{i =1}^{\vdmc}  \log\,\LArg_{i,2} \label{eq:AllKidRecGUnq} \enspace,
  &
  \hat\rho_{\PureP} &\disteq -\sum_{i =1}^{{\vdpc}'}  \log\, \LArg_{i,3} \enspace,				  
  &
  \hat\rho_{\PureM} &\disteq +\sum_{i =1}^{{\vdmc}'}  \log\, \LArg_{i,4} \enspace.
\end{align}

\subsubsection{Coupling and contraction}
While Montanari and Shah~\cite{MS} do not write their proof of the lower bound $\dMS(k)\leq\duniq(k)$ in the language of distributional recurrences, translating their argument to the current formalism evinces two key differences by comparison to the approach that we are going to take.
First, Montanari and Shah establish contraction with respect to messages \emph{from clauses to variables}, instead of messages from variables to clauses 
as considered here. 

While this change of perspective may seem innocuous at first, working with respect to variables provides us with greater control over how the change in log-likelihood ratios propagates. 
In particular, working with variable-to-clause messages and taking into account the four variable types $\rAllKid,\rPureP,\rPureM,\rNoKid$ allows us to optimise the metric with respect to which we establish contraction.
Hence, for $t>0$ we endow the space 
$\cP  (-\infty,\infty] \times \cP (0, +\infty] \times \cP(-\infty, 0]$ 
with the metric
\begin{align}
  \dist_{t}\left(
  \left(\rho_{\AllKid}, \rho_{\PureP}, \rho_{\PureM}\right), \left({\rho}_{\AllKid}', {\rho}_{\PureP}', {\rho}_{\PureM}'\right)
  \right)
  =
  \left(1-e^{-t/2}\right)
  \cdot
  W_1\left(\rho_{\AllKid}, {\rho}_{\AllKid}'\right)
  +
  e^{-t/2}\cdot
  W_1\left(\rho_{\PureP}, {\rho}_{\PureP}'\right)
  +
  e^{-t/2}\cdot
  W_1\left(\rho_{\PureM}, {\rho}_{\PureM}'\right)
  \enspace.
  \label{eq:DistDef}
\end{align}

The following proposition summarises the main step towards the proof of \Thm~\ref{thm_uniq}.

\begin{proposition}\label{prop_BPPureLit}
  For every $d< \dours(k)$, the operator $\LDELitS{d,k}$ is a contraction with respect to the metric $\dist_{d}$.
\end{proposition}

The second key difference between~\cite{MS} and the present approach will emerge in the proof of \Prop~\ref{prop_BPPureLit} itself.
As we are about to see, leveraging the four variable types enables us to carry out a sharper bound on the derivative of our operator $\LDELitS{d,k}$.
This comes in the form of a subtle combinatorial coupling between variable types among clauses with opposite signs.
To explain this, we recall that $\LDELitS{d,k}$ describes how the laws of the log-likelihood ratios 
$\rho_{\AllKid}, \rho_{\PureP}$, and $\rho_{\PureM}$, evolve given the corresponding laws of the variables in one generation below. 
Recall also that we are always considering the positive boundary condition, i.e., the one maximising the value of each log-likelihood ratio. 

Let us write $\rho = (\rho_{\AllKid}, \rho_{\PureP}, \rho_{\PureM})$, $\rho' = (\rho'_{\AllKid}, \rho'_{\PureP}, \rho'_{\PureM})$, and 
$\hat\rho, \hat\rho'$ for their corresponding images under the operator $\LDELitS{d,k}$. We wish to establish that $\dist_d(\hat\rho, \hat\rho') < c\cdot \dist_d(\rho, \rho')$, for
some constant $c=c(d,k) < 1$. 
We call a clause $a$ positive if it contains its parent variable as a direct literal; otherwise, 
we call $a$ negative.
The change between the output distributions $\hat\rho, \hat\rho'$ describing the log-likelihood law
of, say, variable $\root$, comes from two sources: the positive and the negative children of $\root$. Observe that there is no 
obvious symmetry between the two, as we have imposed the positive boundary condition, and therefore, the influence of positive clauses is typically more pronounced. 
In turn, the change caused by each clause can be further attributed to that of the $k-1$ grandchildren variables it features. To be more precise, let us consider
the contribution of a single positive clause $a$. Let us write
$\vr=(\vr_{\AllKid}, \vr_{\PureP}, \vr_{\PureM}, \vr_{\NoKid})$ for the type-distribution of the children variables of $a$, where $\vr$ follows the law described in 
\eqref{eq_def_multipr}. Consider also an arbitrary enumeration of the variables of each type $t \in \{\rAllKid, \rPureP, \PureM\}$,  and write 
$\diffr_i^{t}(z, \vr; +1)$ for the magnitude of the partial derivative of the message clause $a$ sends to $\root$, with respect to the message clause $a$ 
receives from its $i$-th variable of type $t$. Then, the expected contribution of clause $a$ to the distance $\dist(\hat\rho, \hat\rho')$ is bounded in terms of
$\diffr_i^{t}(z, \vr; +)$'s as follows
\begin{align}
  \Erw\brk{
    \sum_{i=1}^{\vec{r}_{\AllKid}}
    \left|\int_{{\ETA}_{\AllKid, i}}^{{\ETA}'_{\AllKid, i}}\!
    \diffr^{\AllKid}_{i}(w_i, \vr; +1)
    \dd w_i
    \right|
    +
    \sum_{j=1}^{\vec{r}_{\PureP}}
    \!\left|
    \int_{{\ETA}_{\PureP, j}}^{{\ETA}'_{\PureP, j}}\!
    \diffr^{\PureP}_{j}(y_j, \vr; +1)
    \dd y_j\right|
    +
    \sum_{\ell=1}^{\vec{r}_{\PureM}}
    \left|
    \int_{{\ETA}_{\PureM, \ell}}^{{\ETA}'_{\PureM, \ell}}
    \!
    \diffr^{\PureM}_{\ell}(z_\ell, \vr; +1)
    \dd z_\ell
    \right|
  }
  \enspace,
  \label{eqtypeContr}
\end{align}
where $\ETA_{t,i},\ETA'_{t,i}$ follow the law of $\rho_t, \rho'_t$, respectively. 
Expanding the expectation with respect to the type-distribution $\vr$, and writing  $P(r)=\pr[\vr=r]$, for the probability 
of a vector $r=(r_{\AllKid}, r_{\PureP}, r_{\PureM}, r_{\NoKid})$, we rewrite \eqref{eqtypeContr} as
\begin{align}
  \sum_{r}P(r)
  \bc{
    {r}_{\AllKid}
    \cdot
    \ex\left|\int_{{\ETA}_{\AllKid, 1}}^{{\ETA}'_{\AllKid, 1}}\!
    \diffr^{\AllKid}_{1}(z, r; +1)
    \dd z
    \right|
    +
    {r}_{\PureP}
    \cdot
    \ex
    \!\left|
    \int_{{\ETA}_{\PureP, 1}}^{{\ETA}'_{\PureP, 1}}\!
    \diffr^{\PureP}_{1}(z, r; +1)
    \dd z\right|
    +
    {{r}_{\PureM}}
    \cdot
    \ex
    \left|
    \int_{{\ETA}_{\PureM, 1}}^{{\ETA}'_{\PureM, 1}}
    \!
    \diffr^{\PureM}_{1}(z, r; +1)
    \dd z
    \right|
  }
  \enspace.
  \label{eqtypeContr2}
\end{align}
The expected contribution of a negative clause is given 
by an expression similar to \eqref{eqtypeContr2}, albeit in terms of $\diffr_1^{t}(z, \vr; -)$, 
i.e.,
the partial derivative of the message $a \to \root$, with respect to the message from a variable of type $t$ to clause $a$. 
Specifically, the expected contribution of a negative clause reads: 
\begin{align}
  \sum_{r'}P(r')
  \bc{
    {r}'_{\AllKid}
    \cdot
    \ex\left|\int_{{\ETA}_{\AllKid, 1}}^{{\ETA}'_{\AllKid, 1}}\!
    \diffr^{\AllKid}_{1}(z, r'; -1)
    \dd z
    \right|
    +
    {r}'_{\PureP}
    \cdot
    \ex
    \!\left|
    \int_{{\ETA}_{\PureP, 1}}^{{\ETA}'_{\PureP, 1}}\!
    \diffr^{\PureP}_{1}(z, r'; -1)
    \dd z\right|
    +
    {{r}'_{\PureM}}
    \cdot
    \ex
    \left|
    \int_{{\ETA}_{\PureM, 1}}^{{\ETA}'_{\PureM, 1}}
    \!
    \diffr^{\PureM}_{1}(z, r'; -1)
    \dd z
    \right|
  }
  \enspace.
  \label{eqtypeContr2N}
\end{align}

It is not hard to see that pure literals have a `dampening' effect on each partial 
derivative $\diffr^{t}_1(z, r; \pm )$. Consider a clause $a$ whose children variables are distributed among the different types according to $r$. Then 
each derivative in \eqref{eqtypeContr2}--\eqref{eqtypeContr2N}, can be bounded in terms of the number of pure literals featured in $a$, 
excluding the variable of type $t$ with respect to which the derivative is taken. Notice that the
operator $\LDELitS{d,k}$ effectively incorporates the positive boundary condition by imposing the sign
of each variable with respect to its parent clause, $a$, to be $+$ if $a$ is positive, and $-$ if $a$ is negative. 
With that in mind, we see that if $a$ is a positive
clause,  the total number of pure literals it contains is just $r_{\PureM} + r_{\NoKid}$. On the other hand, if $a$ is 
negative, then the total number of pure literals it contains is  $r_{\PureP} + r_{\NoKid}$.
Bounding \emph{separately} each derivative $\diffr^{t}_i(z, r; \pm 1 )$ in \eqref{eqtypeContr2}--\eqref{eqtypeContr2N}, 
and invoking the mean value theorem, yields an upper bound on for the contraction constant~$c$.

However, we can do better by partitioning the derivatives in \eqref{eqtypeContr2}--\eqref{eqtypeContr2N} into groups, and optimising them \emph{jointly}.  
Indeed, a careful examination of the expression \eqref{eq:defOfXI}, 
reveals that, for example, any sum of the form $\diffr^{\PureP}(z, (*, *, r_{\PureM},r_{\NoKid}); +1)+\diffr^{\PureP}(z, (*,r_{\PureM}+1, *,r_{\NoKid}); -1)$ 
can be explicitly maximised, and the resulting maximum is smaller than the sum of maxima of the parts. 
At first sight, this seems to be of little use, if any, as in order to implement such a coupling between terms \eqref{eqtypeContr2}--\eqref{eqtypeContr2N}, 
we should also match their coefficients, that is, the quantity $P(r)\cdot r_{\PureP}$, must remain invariant under the coupling.
Somewhat unexpectedly, it turns out that the coupling $r \mapsto r'$ with $r' = (r_{\AllKid}, r_{\PureM}+1, r_{\PureP}-1, r_{\NoKid})$, enjoys both features. 
Similar couplings strategies (depicted in Figure~\ref{fig_coupling}) facilitate the maximisation of $\rPureP, \rPureM$-terms. The full proof of 
\Prop~\ref{prop_BPPureLit} can be found in \Sec~\ref{sec_prop:condMarg}. 

We conclude the section explaining how \Thm~\ref{thm_uniq} follows from the above.

\begin{figure}
  \includegraphics[width=12.5cm]{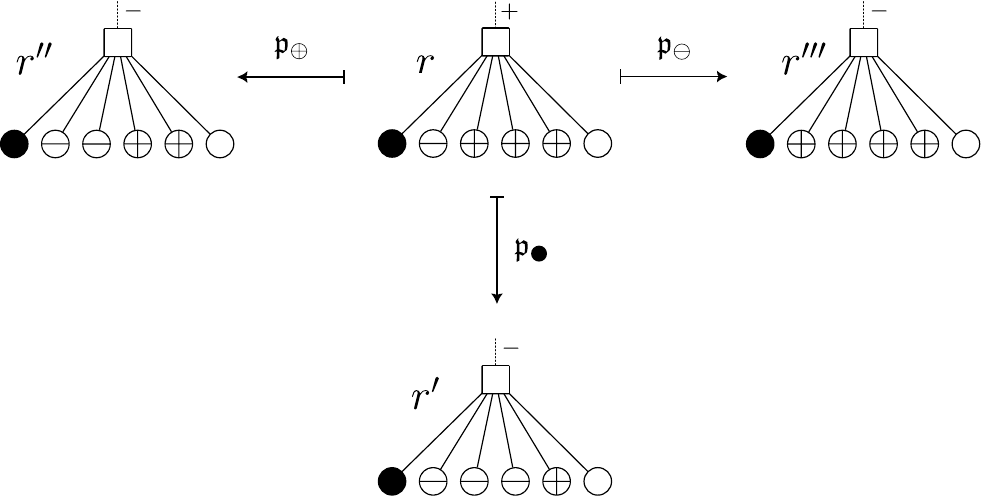}
  \caption{Example of a coupling between derivative terms in \eqref{eqtypeContr2}--\eqref{eqtypeContr2N}. 
    For vector $r$ and type $t\in \{\rAllKid, \rPureP, \rPureM\}$, we pair the term $\diffr^{t}(z, r; +1)$ in \eqref{eqtypeContr2} with
    the term $\diffr^{t}(z, \mathfrak{p}_t(r); -1)$ in \eqref{eqtypeContr2N}.}
  \label{fig_coupling}  
\end{figure}

\begin{proof}[Proof of \Thm~\ref{thm_uniq}]
  From the triangle inequality, and \Lem~\ref{lem:ExtremalBnd}, it is immediate to obtain \Thm~\ref{thm_uniq} from \Prop~\ref{prop:condMarg}.
\end{proof}

\subsection{Discussion}\label{sec_discussion}
The location of the random 2-SAT satisfiability threshold was pinpointed already in the 1990s~\cite{CR,Goerdt} essentially because the threshold coincides with the giant component phase transition of a directed random graph whose edges correspond to the clauses.
This argument also implies that both the pure literal algorithm and another efficient algorithm called unit clause propagation find satisfying assignments up to the satisfiability threshold \whp\
By contrast, in the case of random $k$-SAT with $k\geq3$ the satisfiability threshold is known only for~$k$ exceeding an undetermined (but large) constant $k_0$~\cite{DSS3}.
The proof is based on a sophisticated, physics-inspired second moment argument that significantly extends ideas from earlier work~\cite{nae,yuval,KostaSAT}.
Asymptotically in the limit of large $k$ the satisfiabiliy threshold reads
\begin{align}\label{eqdsat}
  \dsat(k)&=k\brk{2^k\log 2-\frac{1+\log 2}2}+\eps_k,&&\mbox{where }\lim_{k\to\infty}\eps_k=0.
\end{align}
Even though~\cite{nae,yuval,DSS3,KostaSAT} rely on the second moment method, they do not yield asymptotically tight estimates of the number of satisfying assignments for any regime of $d$.
This is because the second moment method is applied not to the number of satisfying assignments, but to another, exponentially smaller random variable.
The assumption that $k$ exceeds a large constant is used critically in~\cite{KostaSAT,DSS3} to ensure certain concentration and expansion properties.

For $3\leq k<k_0$ even the {\em existence} of a uniform satisfiability threshold remains an open problem, although a sharp threshold sequence that may vary with $n$ is known to exist~\cite{Friedgut}.
That said, an upper bound on the satisfiability threshold (sequence) that matches the so-called `1-step replica symmetry breaking' prediction from statistical physics can be verified using the interpolation method from mathematical physics~\cite{FranzLeone,MMZ,MPZ,PanchenkoTalagrand}.
However, the currently known lower bounds for small $k$ (say, $k=3,4,5$) fall short of this upper bound~\cite{yuval,Greg,KKL}.
For example, in the case $k=3$ the best current lower bound is $\dsat(3)\geq10.56$, while $\dsat(3)\approx12.801$ according to physics predictions~\cite{MMZ,MPZ}.

Thus, the satisfiability of random formulas continues to pose a substantial challenge for `small' $3\leq k<k_0$.
In light of this, a particularly satisfactory aspect of the present results is that they apply and are meaningful for all $k\geq3$.
In fact, comparing the asymptotic bounds~\eqref{eqpureasymptotic} and~\eqref{eqdsat}, we see that the Gibbs uniqueness threshold $\duniq(k)$ is much smaller than $\dsat(k)$ for large $k$	.
Thus, \Thm s~\ref{thm_main} and~\ref{thm_uniq} cover larger shares of the satisfiable regime of $d$ for smaller values of $k$; cf.~Table~\ref{tab_main}.

The best current lower bounds on the satisfiability thresholds for $k\geq4$ are non-constructive.
With respect to the algorithmic problem of finding a satisfying assignment of a random $k$-CNF the best current results for `small' $k$ are based on simple combinatorial algorithms, analysed via the method of differential equations~\cite{FriezeSuen,Greg,KKL}.
Asymptotically for large $k$ the best known efficient algorithm~\cite{better} succeeds up to
\begin{align}\label{eqdalg}
  \dalg(k)&=(1-\eps_k)2^k\log k&&\mbox{where }\lim_{k\to\infty}\eps_k=0,
\end{align}
about a factor of $\log(k)/k$ below~\eqref{eqdsat}.
There is evidence that certain types of algorithms do not succeed for much larger values of $d$, at least for enough large $k$~\cite{Barriers,Bresler,bpdec,Hetterich}.
Apart from the task of finding a satisfying assignment, an important line of work deals with the problem of counting and sampling satisfying assignments of random $k$-CNFs for large $k$~\cite{CGGGHPMM,Zongchen,HWY}.
The best current result~\cite{Zongchen} covers the regime $d\leq 2^k/k^c$ for an undetermined (large enough) constant $c>0$.
Since for large $k$ the bound $2^k/k^c$ significantly exceeds the pure literal threshold~\eqref{eqpureasymptotic}, it might be an interesting question whether ideas from~\cite{CGGGHPMM,Zongchen,HWY} can be used to verify the replica symmetric solution~\eqref{eqBFE} for $d$ beyond the Gibbs uniqueness threshold for large $k$.

Most of the prior work on the rigorous verification of the replica symmetric solution focuses on a soft version of random $k$-SAT, the so-called random $k$-SAT model at {\em inverse termperature} $\beta>0$~\cite{PanchenkoBook}.
The partition function $Z_\beta(\PHI)$ of this model, its the key quantity of interest, is defined as follows.
For a clause $a$ of the random formula $\PHI$ and a truth assignment $\sigma$ write $\sigma\models a$ if $\sigma$ satisfies clause $a$.
Then
\begin{align}\label{eqZbeta}
  Z_\beta(\PHI)&=\sum_{\sigma\in\PM^n}\exp\bc{-\beta\sum_{i=1}^{\vm}\vecone\{\sigma\not\models a_i\}}.
\end{align}
Thus, each assignment $\sigma$ contributes a summand equal to $\exp(-\beta)$ raised to the power of the number of clauses that $\sigma$ fails to satisfy.
In effect,
\begin{align}\label{eqZbetainfty}
  Z(\PHI)&=\lim_{\beta\to\infty}Z_\beta(\PHI).
\end{align}
A line of prior work~\cite{Biswas,Panchenko2,Talagrand} deals with the derivation of the `thermodynamic limit'
\begin{align}\label{eqthermo}
  \lim_{n\to\infty}\frac1n\ex\brk{\log Z_\beta(\PHI)}
\end{align}
for small $d$ and/or small $\beta$.
Specifically, these works verify that~\eqref{eqthermo} is given by the replica symmetric solution at inverse temperature $\beta$ from~\cite{MZ,MZ2} under the assumption
\begin{align}\label{eqBiswas}
  d(k-1)\min\cbc{1,6\beta\exp(4\beta)}&<1.
\end{align}
We observe that for large $\beta$ the bound \eqref{eqBiswas} holds only up to the giant component threshold $d=1/(k-1)$, where the replica symmetric solution trivially follows from the fact that Belief Propagation is exact on acyclic graphical models~\cite[Theorem 4.1]{MM}.
That said, a technique called the interpolation method shows that the replica symmetric solution yields an upper bound on~\eqref{eqthermo} for all $d,\beta>0$~\cite{FranzLeone,PanchenkoTalagrand}.
In particular, we will combine the interpolation method with a concentration argument in order to prove \Cor~\ref{cor_interpolation}.

According to physics predictions the `replica symmetric solution' from~\cite{MZ,MZ2} yields the correct value of both $\lim_{n\to\infty}n^{-1}\log Z_\beta(\PHI)$ for all $\beta>0$ and of $\lim_{n\to\infty}n^{-1}\log Z(\PHI)$ for all $d$ up to a threshold $\drsb(k)$ close to but strictly below the satisfiability threshold $\dsat(k)$ for all $k\geq3$~\cite{pnas}.
The threshold $\drsb(k)$ is known as the `1-step replica symmetry breaking phase transition' in physics jargon; its asymptotic value is predicted as
\begin{align}\label{eqdrsb}
  \drsb(k)&=k\brk{2^k\log 2-2\log 2}+\eps_k,&&\mbox{where }\lim_{k\to\infty}\eps_k=0.
\end{align}
Indeed, the interpolation method can be used to verify that the replica symmetric solution ceases to be correct for $\drsb(k)+\eps_k<d<\dsat(k)$ with $\eps_k\to0$.
Conversely, the replica symmetric solution is known to be correct for all $d$ and $\beta>0$ where a certain correlation decay condition is satisfied~\cite{COMR}, provided that $k$ is large enough.
Physics methods predict that this condition holds for all $\beta>0$ and all $d<\drsb(k)$~\cite{pnas}.

The aforementioned work of Montanari and Shah~\cite{MS} also deals with the soft variant of random $k$-SAT~\eqref{eqZbeta}, but allows for an inverse temperature $\beta=\beta(n)$ that tends to infinity slowly 
as $n\to\infty$.
Specifically, considering a small power $\beta=n^\delta$ enables Montanari and Shah to estimate the number of assignments that satisfy all but $o(n)$ clauses.
The proof combines an interpolation on $0\leq\beta\leq n^\delta$ with a contraction argument that improves over the previous contraction estimates from~\cite{Biswas,Panchenko2,Talagrand}.
Instead of the interpolation on $\beta$, in order to prove \Thm~\ref{thm_main} we use the Aizenman-Sims-Starr scheme.
Because we count actual satisfying assignments, this requires the careful combinatorial analysis of tail events, which is where the \pulp\ algorithm and its analysis come in.
Additionally, towards the proof of \Thm~\ref{thm_uniq} we devise an improved version of the contraction argument from~\cite{MS}.
Following Montanari and Shah, we also take advantage of the impact of pure literals on the Belief Propagation operator.
But we develop an improved coupling scheme that yields a better range of $d$ for which contraction occurs.
Additionally, once again because we deal with actual satisfying assignments, the proof of the Gibbs uniqueness property involves the analysis of the \pulp\ algorithm on a Galton-Watson tree in order to cope with unlikely events.

By comparison to the `soft' random $k$-SAT model~\eqref{eqZbeta}, few prior contributions deal with the actual number $Z(\PHI)$ of satisfying assignments.
A result of Abbe and Montanari~\cite{AM}  implies that a deterministic limit (in probability)
\begin{align}\label{eqAM}
  \lim_{n\to\infty}\frac1n\log Z(\PHI)
\end{align}
exists for Lesbegue-almost all $0<d<\dpure(k)$ for all $k\geq2$.
However, the proof, which is based on the interpolation method, does not reveal the value of~\eqref{eqAM}.
In fact, prior to the present work the limit \eqref{eqAM} was known only in two cases.
First, in the trivial regime $d<1/(k-1)$ below the giant component threshold.
Second, in the case $k=2$ for $0<d<\dsat(2)=2$~\cite{2sat}.
In both cases the limit \eqref{eqAM} coincides with the replica symmetric solution from~\cite{MZ}.
Beyond the convergence in probability, $\log Z(\PHI)$ is known to satisfy a central limit theorem in the case $k=2$~\cite{CCOMRRZZ}.

To compute the limit~\eqref{eqAM} in the case $k=2$ the contribution~\cite{2sat} employs the Aizenman-Sims-Starr scheme.
The couplings that we use towards the proofs of \Prop~\ref{lem_PHI''}--\ref{lem_PHI'''} generalise the argument from~\cite{2sat} to $k\geq3$.
The main technical novelty lies in the way that moderately unlikely events are treated.
Specifically, in the case $k=2$ the simple Unit Clause propagation algorithm, which essentially boils down to directed reachability, was sufficient to derive a tail bound similar to (and actually stronger than)~\eqref{eq_prop_bad}.
By contrast, since in the case $k\geq3$ the clauses ``branch out'', the analysis of tail events and, accordingly, the derivation of \eqref{eq_prop_bad} is far more delicate.
The core of this derivation is the detailed analysis of the \pulp\ algorithm right up to $\dpure(k)$.

Finally, by contrast to random $k$-SAT the validity of the replica symmetric solution is known for the optimal parameter range for several other random constraint satisfaction problems that enjoy certain symmetry properties.
Examples include random graph colouring or random $k$-NAESAT~\cite{cond,CKPZ}.
Due to the symmetry property
\footnote{A formal definition of `symmetry' could be that the uniform distribution on spins is a fixed point of the Belief Propagation operator on the respective Galton-Watson tree.}
the replica symmetric solution simply coincides with the first moment of the number of solutions.
In effect, in many symmetric problems it is even possible to precisely determine the limiting distribution of the number of solutions, which superconcentrates on the first moment~\cite{CKM}.
By contrast, in random $k$-SAT the first moment overshoots the typical number of satisfying assignments by an exponential factor~\cite{nae}, which is why random $k$-SAT is so much more delicate than symmetric problems.
That said, there is a regular variant of random $k$-SAT (where every variable appears an equal number of times positively and negatively) where symmetry and superconcentration are recovered~\cite{COW}.

\subsection{Organisation}\label{sec_org}
In the remaining sections we work our way through the proofs of \Thm s~\ref{thm_main} and \ref{thm_uniq}. Specifically, in
\Sec~\ref{sec_pulp_pr} we analyse the \pulp\ algorithm introduced in \Sec~\ref{sec_pulp}, proving \Lem s~\ref{lem_Ichi}-\ref{lem_EIchi},
which facilitates many of the subsequent results.

In \Sec~\ref{sec_prop_arnab} we establish \Prop~\ref{prop_arnab}, verifying that the quantities appearing in \Thm~\ref{thm_main} are well-defined.
The proof of \Cor~\ref{cor_interpolation} follows in  \Sec~\ref{sec_prop_interpolation}.
\Sec~\ref{sec_prop_aizenman} is devoted to the Aizenman-Sims-Starr scheme, and in particular the proof of \Prop~\ref{prop_aizenman}.
There we also complete the proof of \Thm~\ref{thm_main}.

Our final \Sec~\ref{sec_prop:condMarg}, deals with the remaining proofs toward establishing \Thm~\ref{thm_uniq}. We begin by proving \Lem~\ref{lem_SmallOnTop},
showing that
the log-likelihood ratios of the random Galton-Watson formula close to the root are bounded \whp. This enables us to
compare the output distribution of the non-random operator introduced in \Sec~\ref{sec_uniq_outline} with that of actual ratios on the random tree.
We then proceed with the proof of \Prop~\ref{prop_BPPureLit}, and conclude with the proof of \eqref{eqstrangeconv}, completing the proof of \Thm~\ref{thm_uniq}.

\section{Analysis of {\pulp}}\label{sec_pulp_pr}

\noindent
This section is concerned with the analysis of \pulp\ from \Sec~\ref{sec_pulp}.
In particular, we prove \Lem~\ref{lem_EIchi}.
But let us get the proof of \Lem~\ref{lem_Ichi} out of the way first.

\subsection{Proof of \Lem~\ref{lem_Ichi}}
Suppose that $\bar\cL \supseteq \cL$ satisfies {\bf PULP1}--{\bf PULP2}.
Let $U=\{|l|:l\in\bar\cL\}$  be the set of variables underlying the literals $\bar\cL$.
Moreover, let $\chi: U \to \PM$ be the truth assignments under which all literals of $l\in\bar\cL$ evaluate to `true'.
Due to {\bf PULP2}, the assignment $\chi$ is well defined.
Moreover, since $\cL\subseteq\bar\cL$, under $\chi$ all literals $l\in\cL$ evaluate to `true'.
Hence, for a satisfying assignment $\sigma\in S(\Phi)$ define an assignment $\sigma'$ by letting 
\begin{align*}
  \sigma'(x)=\vecone\{x\in U\}\chi(x)+\vecone\{x\not\in U\}\sigma(x) \enspace.
\end{align*}
Because $\bar\cL$ satisfies condition {\bf PULP1}, we have $\sigma'\in S(\Phi,\cL)$.
Finally, because for a satisfying assignment $\tau'\in S(\Phi,\cL)$ there are no more than $2^{|U|}=2^{|\bar\cL|}$ satisfying assignments $\tau\in S(\Phi)$ such that $\tau(x)=\tau'(x)$ for all $x \not\in U$, we obtain the desired bound $Z(\Phi)\leq 2^{|\bar\cL|}Z(\Phi,\cL)$.

\subsection{Turning a tree to \pulp}\label{sec_treepulp}
While the ultimate goal of this section is to study the \pulp\ algorithm on the random formula $\PHI'$ to prove \Lem~\ref{lem_EIchi}, a necessary preparation is to investigate the algorithm on the random Galton-Watson tree $\TT=\TT_{d,k}$.
Of course, since $\TT$ may be infinite we should formally confine ourselves to the finite trees $\TT^{(\ell)}$ truncated at the $2\ell$-th level from the root $\root$.
Hence, recalling~\eqref{eqheightdef}, we aim to estimate the height $\fn_\root(s,\TT^{(\ell)})$ for finite $\ell$.
That said, since these random variables are monotonically increasing in $\ell$, it makes sense to define
\begin{align}\label{eqhTT}
  \fn_\root(s,\TT)&=\lim_{\ell\to\infty}\fn_\root(s,\TT^{(\ell)})\in[0,\infty].
\end{align}
We point out that for $d<\dpure(k)$ the tails of $\fn_\root(s,\TT)$ decay at a doubly exponential rate.

\begin{lemma}\label{lem:Peeling_on_Trees}
  For any $d<\dpure(k)$ there exist $c_1=c_1(d,k),c_2=c_2(d,k)>0$ such that
  \begin{align*}
    \Pr \brk{ \fn_{\root}(\pm1,\TT) \geq h } &\leq c_1 \cdot \exp\bc{-\exp\bc{c_2 \cdot h}} && \text{ for every } h \ge 1 \enspace.
  \end{align*}
\end{lemma}
\begin{proof}
  By symmetry it suffices to consider $\fn_{\root}(1,\TT)$.
  Thus, let $p_{h, \ell} = \Pr \brk{ \fn_{\root} (1,\TT^{(\ell)})\geq h}$.
  All variables at distance $2\ell$ from $\root$ are leaves and therefore pure in the tree $\TT^{(\ell)}$.
  Consequently, pure literal elimination removes all clauses of $\TT^{(\ell)}$ within at most $\ell$ rounds.
  Hence, $p_{h,\ell}=0$ for $h>\ell$.
  Furthermore, we claim that
  \begin{align}\label{eqlem:Peeling_on_Trees_1}
    p_{h,\ell}&=\varphi_{d,k}(p_{h-1,\ell-1}),\qquad\mbox{where}\quad\varphi_{d,k}(z)=1-\exp\bc{-\frac d2z^{k-1}}&&(1\leq h\leq\ell).
  \end{align}
  Indeed, if $\fn_{\root} (1,\TT^{(\ell)})\geq h\geq1$ then  by~\eqref{eqheightdef} there exists a clause $a\in\partial_{\TT} r$ with $\sign(\root,a)=-1$ such that $\fn_a(\TT^{(\ell)}[\root\mapsto 1])\geq h-1$.
  In other words, pure literal elimination on the sub-tree $\TT^{(\ell)}[a]$ of $\TT^{(\ell)}$ rooted at clause $a$ and with variable $\root$ removed takes at least $h-1$ rounds to remove clause $a$.
  Consequently, pure literal elimination on $\TT^{(\ell)}[a]$ takes at least $h-1$ rounds to remove one of the variables $x\in\partial_\TT a\setminus\{\root\}$.
  In other words, the sub-tree $\TT^{(\ell)}[x]$ comprising $x$ and its successors satisfies
  \begin{align}\label{eqlem:Peeling_on_Trees_2}
    \fn_x(\sign(x,a),\TT^{(\ell)}[x])&\geq h-1&&\mbox{for every }x\in\partial_\TT a\setminus\{\root\}.
  \end{align}
  But since $\TT$ is a Galton-Watson tree, the sub-tree $\TT^{(\ell)}[x]$ has the same distribution as the random tree $\TT^{(\ell-1)}$.
  Hence, \eqref{eqlem:Peeling_on_Trees_2} implies that for every $a\in\partial_{\TT}\root$ with $\sign(x,a)=-1$,
  \begin{align}\label{eqlem:Peeling_on_Trees_3}
    \pr\brk{\fn_a(\TT^{(\ell)}[\root\mapsto 1])\geq h-1\mid\TT^{(1)}}&=p_{h-1,\ell-1}^{k-1}.
  \end{align}
  Finally, the construction of $\TT$ ensures that the number of $a\in\partial_{\TT}\root$ with $\sign(x,a)=-1$ has distribution $\Po(d/2)$.
  Therefore, \eqref{eqlem:Peeling_on_Trees_3} shows that
  \begin{align*}
    p_{h,\ell}&=\sum_{i=0}^{\infty}\pr\brk{\Po(d/2)=i}\bc{1-\bc{1-p_{h-1,\ell-1}^{k-1}}^i}=1-\exp\bc{-\frac d2p_{h-1,\ell-1}^{k-1}},
  \end{align*}
  which completes the proof of~\eqref{eqlem:Peeling_on_Trees_1}.
  
  Since the sequences $(p_{h,\ell})_{\ell}$ are non-decreasing, the limits $p_h=\lim_{\ell\to\infty}p_{h,\ell}$ exist.
  Moreover, \eqref{eqlem:Peeling_on_Trees_1} shows that
  \begin{align}\label{eqlem:Peeling_on_Trees_4}
    p_h&=\varphi_{d,k}(p_{h-1})&&(h\geq1).
  \end{align}
  Hence, recalling the definition~\eqref{eqpure} of $\dpure(k)$, we find
  \begin{align*}
    \left(\frac{p_{h+1}}{p_{h}}\right)^{k-1} = \left(\frac{\varphi_{d,k}(p_h)} {p_h} \right)^{k-1} = d\cdot \frac{\left( 1 - \exp\left(-dp_h^{k-1}/2\right)\right)^{k-1}} {d p^{k-1}_h} \le \frac{d}{\dpure}<1 \enspace.
  \end{align*}
  Consequently,
  \begin{align}\label{eqlem:Peeling_on_Trees_5}
    \lim_{h\to\infty}p_h&=0.
  \end{align}
  
  To complete the proof we expand $\varphi_{d,k}(z)$ around $z=0$:
  \begin{align}\label{eqlem:Peeling_on_Trees_6}
    \varphi_{d,k}(z) &= \frac{d}{2}z^{k-1} +O(z^{2k-2})&&\mbox{ as }z\to 0.
  \end{align}
  Thus, 
  the function $\varphi_{d,k}(z)$ is well approximated by a $(k-1)$-th power. 
  Since $k\geq3$, combining \eqref{eqlem:Peeling_on_Trees_4}--\eqref{eqlem:Peeling_on_Trees_6}, we conclude that for sufficiently large $h$ we have $p_h\leq(d/2+1)p_{h-1}^{k-1}$.
  Consequently, $p_h\leq c_1 \cdot \exp\bc{-\exp\bc{c_2 \cdot h}}$ for suitable $c_1=c_1(d,k)$ and $c_2=c_2(d,k)$.
\end{proof}

We remind ourselves that 
$\overline{\{\pm1\cdot\root\}}_{\TT^{(\ell)}}$ signifies the output of \pulp\ run on the formula $\TT^{(\ell)}$ with initial literal set $\{\pm1\cdot\root\}$.
We extend the definition of the closure to the (possibly infinite) tree $\TT$ by letting
\begin{align*}
  \overline{\{\pm1\cdot\root\}}_{\TT}&=\bigcap_{\ell_0\geq1}\bigcup_{\ell\geq\ell_0}\overline{\{\pm1\cdot\root\}}_{\TT^{(\ell)}}.
\end{align*}
This definition ensures that if the height $\fn_{\root}(\pm1,\TT)$ from~\eqref{eqhTT} is finite, then
\begin{align*}
  \overline{\{\pm1\cdot\root\}}_{\TT}=\overline{\{\pm1\cdot\root\}}_{\TT^{(\ell)}}\qquad\mbox{ for all }\ell\geq\fn_{\root}(\pm1,\TT).
\end{align*}
In order to estimate the size of this set, we combine \Lem~\ref{lem:Peeling_on_Trees} with a crude bound on the total number of variable nodes of the Galton-Watson tree $\TT^{(\ell)}$.
Recall that $V(\TT^{(\ell)})$ signifies the set of variable nodes of $\TT^{(\ell)}$.

\begin{lemma}\label{lem_crude}
  Let $d>0$.
  For any $\ell\geq1$ and any $t>100(1+d(k-1))^2$ we have
  \begin{align*}
    \pr\brk{|V(\TT^{(\ell)})|>t^\ell}&\leq \ell\exp(-t^{\ell/2}/4).
  \end{align*}
\end{lemma}
\begin{proof}
  Let $\vN_\ell=|V(\TT^{(\ell)})|$ for brevity, set $g=10(1+d(k-1))$ and notice that $t>g^2$.
  The construction of the Galton-Watson tree $\TT$ ensures that $\vN_0=1$ and that for $\ell\geq1$ given $\vN_{\ell-1}$ we have $\vN_\ell\disteq(k-1)\cdot\Po(d\vN_{\ell-1}).$
  Therefore, Bennett's inequality shows that
  \begin{align}\label{eqcrude1}
    \pr\brk{\vN_h>g^{h-\ell}t^\ell\mid\vN_{h-1}\leq g^{h-1-\ell}t^\ell}&\leq
    \exp\bc{-\frac{t^\ell}{4g^{\ell-h}}}\leq\exp(-t^{\ell/2}/4)&&(1\leq h\leq\ell).
  \end{align}
  Furthermore, if $\vN_\ell>t^\ell$ then there exists $1\leq h\leq\ell$ such that $\vN_h>g^{h-\ell}t^\ell$ while $\vN_{h-1}\leq g^{h-1-\ell}t^\ell$.
  Hence, combining~\eqref{eqcrude1} with the union bound completes the proof.
\end{proof}

\begin{corollary}\label{lem:SMSizeGivenHeight}
  For any $d<\dpure(k)$ there exists $c_3=c_3(d,k)>0$ such that
  \begin{align*}
    \pr[|\overline{\{\pm1\cdot\root\}}_{\TT}|>t]&\leq c_3\exp(-t^{1/c_3})&&\mbox{ for all }t>0.
  \end{align*}
\end{corollary}
\begin{proof}
  By symmetry it suffices to consider $\vN=|\overline{\{\root\}}_{\TT}|$.
  Since \Lem~\ref{lem:Peeling_on_Trees} shows that $\pr[\fn_{\root}(1,\TT)<\infty]=1$, we may assume from now on that indeed $\fn_{\root}(1,\TT)<\infty$.
  Moreover, picking $c_3=c_3(d,k)>0$ large enough, we may assume that $t>t_0$ for a large $t_0=t_0(d,k)$.
  Let $\vN_\ell=|V(\TT^{(\ell)})|$, $p_h=\pr\brk{\fn_\root(1,\TT)=h}$ and $g=10(1+d(k-1))$.
  It is an immediate consequence of the way that \pulp\ proceeds that for all $l\in\overline{\{\root\}}_{\TT}$ we have $|l|\in V(\TT^{(\fn_\root(1,\TT))})$.
  Hence, $\vN\leq\vN_{\fn_\root(1,\TT)}$.
  Therefore, by the law of total probability,
  \begin{align}\label{eqlem:SMSizeGivenHeight_1}
    \pr\brk{\vN>t}&\leq\sum_{h\geq0}S_h,&&\mbox{where }S_h=\pr\brk{\fn_\root(1,\TT)=h}\pr\brk{\vN_h>t\mid\fn_\root(1,\TT)=h}.
  \end{align}
  Depending on the value of $t$ in relation to $h$, we use either \Lem~\ref{lem:Peeling_on_Trees} or \Lem~\ref{lem_crude} to bound $S_h$.
  \begin{description}
    \item[Case 1: $t_0<t\leq g^{2h}$] \Lem~\ref{lem:Peeling_on_Trees} shows that for certain $c_1,c_2>0$ we have
    \begin{align}\label{eqlem:SMSizeGivenHeight_2}
      S_h&\leq \pr\brk{\fn_\root(1,\TT)=h} \leq c_1\exp(-\exp(c_2h))\leq c_12^{-h}\exp(-t^{1/c_3}),
    \end{align}
    provided $c_3$ is chosen large enough.
    \item[Case 2: $t>g^{2h}$] we apply \Lem~\ref{lem_crude} to obtain
    \begin{align}\label{eqlem:SMSizeGivenHeight_3}
      S_h&\leq \pr\brk{\vN_h>t} \leq h\exp(-\sqrt t/4)\leq h2^{-h}\exp(-t^{1/3}),
    \end{align}
    provided that $t>t_0$ is sufficiently large.
  \end{description}
  Combining the bounds~\eqref{eqlem:SMSizeGivenHeight_1}--\eqref{eqlem:SMSizeGivenHeight_3} completes the proof.
\end{proof}

\begin{corollary}\label{cor_c3}
  For any $d<\dpure(k)$ we have $\ex[|\overline{\{\pm1\cdot\root\}}_{\TT}|^2]<\infty.$
\end{corollary}
\begin{proof}
  This is an immediate consequence of \Cor~\ref{lem:SMSizeGivenHeight}.
\end{proof}

\subsection{Proof of \Lem~\ref{lem_EIchi}}\label{sec_lem_EIchi}
Because the distribution of $\PHI'$ is invariant under variable permutations and inversions, we may assume the initial set $\cL$ of literals passed to \pulp\ is just $\cL=\{x_1,\ldots,x_L\}$ for an integer $L=\tilde O(1)$.
For an integer $\ell\geq1$ let $\vec\phi_{\ell,L}'$ be the sub-formula of $\PHI'$ comprising all clauses and variables at distance at most $2\ell$ from $\cL$.
We recall that this formula has a bipartite graph representation $G(\vec\phi_{\ell,L}')$ with variable nodes $V(\vec\phi_{\ell,L}')$, clause nodes $F(\vec\phi_{\ell,L}')$ and edges $E(\vec\phi_{\ell,L}')$.
The {\em excess} of $\vec\phi_{\ell,L}'$ is defined as
\begin{align*}
  \vX_{\ell,L}&=|E(\vec\phi_{\ell,L}')|-|V(\vec\phi_{\ell,L}')|-|F(\vec\phi_{\ell,L}')|.
\end{align*}
Thus, $\vX_{\ell,L}=-L$  iff $G(\vec\phi'_{\ell,L})$ consists of $L$ acyclic components.

\begin{lemma}\label{lem_excess}
  Let $d>0$, $c>0$ and assume that $L\leq\log^cn$ and $\ell\leq c\log\log n$.
  Then
  \begin{align}\label{eqlem_excess1}
    \pr\brk{\vX_{\ell,L}>-L}&= \tilde O(n^{-1}),&\pr\brk{\vX_{\ell,L}>1-L}&= \tilde O(n^{-2}).
  \end{align}
  Furthermore, there exists $c_4=c_4(c,d,k)>0$ such that
  \begin{align}\label{eqlem_excess2}
    \pr\brk{|V(\vec\phi'_{\ell,L})|+|F(\vec\phi'_{\ell,L})|>\log^{c_4}n}&= O(n^{-2}).
  \end{align}
\end{lemma}
\begin{proof}
  We study breadth first search (`BFS') on the graph $G(\PHI')$ from the start vertices $\cL$ by means of a routine deferred decisions argument.
  Throughout the execution of BFS each variable node is in one of three possible states: unexplored, active, or finished.
  
  Towards the proof of~\eqref{eqlem_excess2} we study a `parallel' version of BFS.
  More precisely, let $\cA_0=\cL$ be the set of initially active variables, let $\cU_0=\{x_1,\ldots,x_n\}\setminus\cL$ comprise the initially unexplored variables and let $\cF_0=\emptyset$.
  Further, for $t\geq0$ define $\cA_{t+1},\cU_{t+1},\cF_{t+1}$ as follows.
  If $\cA_t=\emptyset$ then the process has stopped and we let $\cA_{t+1}=\cA_t=\emptyset,\cU_{t+1}=\cU_t,\cF_{t+1}=\cF_t$.
  Otherwise let $\cA_{t+1}$ be the set of all variable nodes $y\in\cU_{t}$ such that there exist an active variable node $x\in\cA_t$ and a clause $a$ that contains $x$ and $y$; in symbols, $x,y\in\partial_{\PHI'}a$.
  Further, let $\cF_{t+1}=\cF_t\cup\cA_t$ and $\cU_{t+1}=\cU_t\setminus\cA_{t+1}$.
  The BFS exploration occurs `in parallel' in the sense that all active vertices activate their previously unexplored second neighbours simultaneously.
  
  Let $\fF_t$ be the $\sigma$-algebra generated by the first $t$ rounds of parallel exploring.
  Then the distribution of $|\cA_{t+1}|$ given $\fF_t$ is stochastically dominated by a random variable with distribution $(k-1)\Po(d|\cA_t|)$.
  This is because by the construction of the formula $\PHI'$ the total number of clauses containing a given variable node has distribution $\Po(d(1-(k-1)/n))$.
  Hence, for any $u>0$ we have
  \begin{align}\label{eqlem_excess3}
    \pr\brk{|\cA_{t+1}|>u\mid\fF_t}&\leq\pr\brk{(k-1)\Po(d|\cA_t|)>u}.
  \end{align}
  
  To complete the proof of \eqref{eqlem_excess2} we mimic the argument from the proof of \Lem~\ref{lem_crude}.
  Thus, let $u=\log^{c_4-3}n$ for a large enough $c_4=c_4(c,d,k)$ and set $g=10(1+d(k-1))$.
  Since $\ell\leq c\log\log n$, the bound \eqref{eqlem_excess3} and Bennett's inequality show that
  \begin{align*}
    \pr\brk{|\cA_{t+1}|>g^{t+1-\ell}u\mid|\cA_t|\leq g^{t-\ell}u}&\leq\exp\bc{-\frac{u}{4g^{\ell-t+1}}}\leq\exp(-\sqrt u/4)=O(n^{-3})&&(0\leq t<\ell).
  \end{align*}
  Hence, taking a union bound on $0\leq t<\ell$ and observing that $|V(\vec\phi'_{\ell,L})|\subset\cA_0\cup\cdots\cup\cA_\ell$, we obtain
  \begin{align}\label{eqlem_excess4}
    \pr\brk{|V(\vec\phi'_{\ell,L})|>u\log n}&=\tilde O(n^{-3}).
  \end{align}
  Finally, another application of Bennett's inequality demonstrates that with probability $1-O(n^{-2})$ no variable of $\PHI'$ appears in more than $\log n$ clauses.
  Thus, $|F(\vec\phi'_{\ell,L})|\leq|V(\vec\phi'_{\ell,L})|\log n$.
  Hence, \eqref{eqlem_excess4} implies~\eqref{eqlem_excess2}.
  
  We are left to establish~\eqref{eqlem_excess1}.
  The way we set up the BFS process implies that there are only two ways in which excess edges can come about.
  First, there may be clauses $a$ with $\partial_{\PHI'}a\subset\cA_t\cup\cA_{t+1}$ such that $|\partial_{\PHI'}a\cap\cA_t|\geq2$.
  Given that $|\cA_t\cup\cA_{t+1}|\leq\log^{c_4}n$, the number of such $a$ with $|\partial_{\PHI'}a\cap\cA_t|=2$ has distribution $\Po(\tilde O(1/n))$, and the number of $a$ with $|\partial_{\PHI'}a\cap\cA_t|>2$ has distribution $\Po(\tilde O(1/n^2))$.
  The second possibility is that for a variable $x\in\cA_{t+1}$ there exist clauses $a,b\in\partial_{\PHI'}x$ with $\partial_{\PHI'}a,\partial_{\PHI'}b\subset\cA_t\cup\cA_{t+1}$.
  Once again the number of such clauses has distribution $\Po(\tilde O(1/n))$ given $|\cA_t\cup\cA_{t+1}|\leq\log^{c_4}n$.
  Furthermore, excess inducing clauses occur independently at different rounds $t$ of the BFS process.
  Thus, \eqref{eqlem_excess1} follows from~\eqref{eqlem_excess2}.
\end{proof}

We proceed to derive bounds on $|\bar\cL|=|\bar\cL_{\PHI'}|$ depending on the value of the excess.
To deal with the case of excess $-L$, let $\Lambda=\Theta(\log\log n)$ and let $(\TT[i])_{i\geq1}$ be a sequence of independent copies of the random tree $\TT$.
In the case that the excess $\vX_{\Lambda,L}$ equals $-L$, the bound on $|\bar\cL|$ follows from the fact that the Galton-Watson tree $\TT$ captures the local structure of the graph $G(\PHI')$ in combination with the bound from \Cor~\ref{lem:SMSizeGivenHeight}.
More precisely, the following is true.

\begin{lemma}\label{lem_noexcess}
  For any $0<d<\duniq(k)$ and $c>0$ there exists $\zeta=\zeta(c,d,k)>0$ such that with $\Lambda=\lceil\zeta\log\log n\rceil$ 
  uniformly for all $1\leq L\leq\log^cn$ and all $u>0$ we have
  \begin{align*}
    \pr\brk{\vecone\{\vX_{\Lambda,L}=-L\}|\bar\cL|>u}&\leq\pr\brk{\sum_{i=1}^L|\overline{\{\root\}}_{\TT[i]}|>u}+O(n^{-2}).
  \end{align*}
\end{lemma}
\begin{proof}
  We begin by coupling the random formula $\vec\phi'_{\ell,1}$ with the Galton-Watson tree $\TT^{(\ell)}[1]$ for $0\leq\ell\leq\Lambda$.
  The coupling operates in accordance with the iterations of the BFS process from the proof of
  \Lem~\ref{lem_excess}.
  Under the coupling some of the variable and clause nodes of $\vec\phi'_{\ell,1}$ and of the tree $\TT^{(\ell)}[1]$ are identical, but both $\TT^{(\ell)}[1]$ and $\vec\phi'_{\ell,1}$ may contain additional clauses or variables.
  These additional clauses/variables result from excess edges of $G(\vec\phi_{\ell,1}')$, i.e., edges that close cycles or merge different components in the course of the BFS process.
  
  For $\ell=0$ we just identify the start variable $x_1$ with the root $\root$ of the Galton-Watson tree $\TT[1]$.
  Going from $\ell$ to $\ell+1$, we remember the sets $\cA_\ell,\cA_{\ell+1}$ from the proof of \Lem~\ref{lem_excess}.
  For each variable $x\in\cA_\ell$ let $\cC_x$ be the set of clauses $a\in\partial_{\PHI'}x$ such that $|\partial_{\PHI'}a\cap\cA_{\ell+1}|=k-1$ and also 
  such that none of the variables $y\in\cA_{\ell+1}\cap\partial_{\PHI'}a$ appear in another clause $b\neq a$ with $\partial_{\PHI'}b\subset\cA_\ell\cup\cA_{\ell+1}$.
  In other words, $\cC_x$ contains all clauses $a\in\partial_{\PHI'}x$ that do not induce excess edges.
  Let $\vd_x=|\cC_x|$ be the number of such clauses.
  As we pointed out in the proof of \Lem~\ref{lem_excess}, $\vd_x$ is stochastically dominated by a $\Po(d)$ variable.
  Hence, there is a random variable $\vd_x'$ such that $\vd_x+\vd_x'\disteq\Po(d)$.
  
  For any variable $x\in\cA_\ell$ that is also a variable node of $\TT^{(\ell)}[1]$ we add all clauses $a\in\cC_x$ and the $k-1$ variables $y\in\partial_{\PHI'}a\cap\cA_{\ell+1}$ to $\TT^{(\ell+1)}[1]$.
  Additionally, $\TT^{(\ell+1)}[1]$ contains $\vd_x'$ independent random clauses that contain $x$ and $k-1$ new variable nodes without a counterpart in $\vec\phi'_{\ell+1,1}$.
  Finally, to complete $\TT^{(\ell+1)}[1]$ every variable $y$ of $\TT^{(\ell)}[1]$ at distance precisely $2\ell$ from $r$ such that $y\not\in V(\vec\phi_{\ell,1}')$ independently begets $\Po(d)$ offspring clause nodes, each containing $k-1$ new variable nodes that do not belong to $V(\vec\phi_{\ell+1,1}')$.
  
  The coupling ensures that $\vec\phi_{\Lambda,1}'$ is a sub-formula of $\TT^{(\Lambda)}[1]$ unless $\vX_{\Lambda,1}>-L$.
  The extension of this coupling to $\cL=\{x_1,\ldots,x_L\}$ is straightforward.
  We simply perform BFS exploration from the start variables $x_1,\ldots,x_L$ one after the other.
  Given that $\vX_{\Lambda,L}=-L$, we thus couple the sub-formula of $\PHI'$ explored from each $x_i$ with $\TT^{(\Lambda)}[i]$ for $1\leq i\leq L$ such that $\vec\phi'_{\Lambda,L}$ is contained in the union of $\TT^{(\Lambda)}[1],\ldots,\TT^{(\Lambda)}[L]$.
  Finally, we obtain independent copies $\TT[1],\ldots,\TT[L]$ of the (possibly infinite) tree $\TT$ by continuing the Galton-Watson processes $\TT^{(\Lambda)}[i]$ independently for depths $\ell>\Lambda$.
  
  The remaining task is to compare $|\bar\cL|$ with $\sum_{i=1}^L|\overline{\{\root\}}_{\TT[i]}|$.
  If $\vX_{\Lambda,L}=-L$ and if $\fn_\root(1,\TT[i])<\Lambda$ for all $1\leq i\leq L$, then the coupling ensures that all clauses 
  and variables of $\vec\phi_{\Lambda,L}'$ are contained in the disjoint union of the trees $\TT[1],\ldots,\TT[L]$, and thus 
  $|\bar\cL|\leq\sum_{i=1}^L|\overline{\{\root\}}_{\TT[i]}|$.
  Therefore, for any $u>0$ we have
  \begin{align}\label{eq_lem_noexcess1}
    \pr\brk{\vecone\cbc{\vX_{\Lambda,L}=-L,\,\max_{1\leq i\leq L}\fn_\root(1,\TT[i])<\Lambda}|\bar\cL|>u}&\leq
    \pr\brk{\sum_{i=1}^L|\overline{\{\root\}}_{\TT[i]}|>u}.
  \end{align}
  Furthermore, since $\Lambda\geq\zeta\log\log n$ for a large $\zeta>0$, \Lem~\ref{lem:Peeling_on_Trees} ensures that
  \begin{align}\label{eq_lem_noexcess2}
    \pr\brk{\max_{1\leq i\leq L}\fn_\root(1,\TT[i])\geq\Lambda}&\leq O(n^{-2}).
  \end{align}
  Combining~\eqref{eq_lem_noexcess1} and~\eqref{eq_lem_noexcess2} completes the proof.
\end{proof}

For later reference we make a note of the following immediate consequence of the coupling from the proof of \Lem~\ref{lem_noexcess}.
For two rooted Boolean formulas $\phi,\phi'$ we write $\phi\ism\phi'$ if there is an isomorphism of $\phi$ and $\phi'$ that preserves the root variable.
We consider the random formula $\vec\phi'_{\ell,1}$ rooted at $x_1$.

\begin{corollary}\label{cor_lcwk}
  For every $\ell \ge0$ and any fixed tree $T$ we have $\left|\Pr\brk{\TT^{(\ell)}\ism T}- \Pr\brk{\vec\phi'_{\ell,1}\ism T}\right| = o(1).$
\end{corollary}

{\em From here on, we set $\Lambda=\lceil c_5\log\log n\rceil$ for a large enough $c_5=c_5(d,k)>0$.}
We obtain the following bound on the second moment of $|\bar\cL|$ on the event that the excess equals $-L$.

\begin{corollary}\label{lem:coup}
  For any $0<d<\duniq(k)$ and any $1\leq L\leq\log^2n$ we have $\ex[\vecone\{\vX_{\Lambda,L}=-L\}\cdot|\bar\cL|^2]=O(1)$.
\end{corollary}
\begin{proof}
  Since $|\bar\cL|\leq 2n$ deterministically, this is an immediate consequence of \Cor~\ref{lem:SMSizeGivenHeight} and \Lem~\ref{lem_noexcess}.
\end{proof}

As a next step we deal with the case that the excess equals $1-L$.
More precisely, with $c_6=c_6(d,k)\gg c_5$ a large enough constant let $\Lambda^+=\lceil c_6\log\log n\rceil$.
We are going to bound $|\bar\cL|$ on the event that $\vX_{\Lambda,L}=\vX_{\Lambda^+,L}=1-L$.
The proof combines the bound on the probability of this event from \Lem~\ref{lem_excess} with a crude bound on $|\bar\cL|$.
To elaborate, since \Lem~\ref{lem_excess} shows that the event $\vX_{\Lambda,L}=\vX_{\Lambda^+,L}=1-L$ has probability $\tilde O(n^{-1})$, we can essentially get away with simply bounding $|\bar\cL|$ by the total number of variables within a $2\Lambda^+$ radius around the start variables $\cL$.
Indeed, as \Lem~\ref{lem_excess} shows, this number of variables is very likely polylogarithmic in $n$.
Working out the details, we obtain the following.

\begin{lemma}\label{lem_unicyc}
  Let $0<d<\duniq(k)$ and let $1\leq L\leq\log^2n$.
  Then $\ex\brk{\vecone\{\vX_{\Lambda,L}=\vX_{\Lambda^+,L}=1-L\}|\bar\cL|^{3/2}}=o(1).$
\end{lemma}
\begin{proof}
  Let $\cV^+=V(\vec\phi'_{\Lambda,L})\setminus V(\vec\phi'_{\Lambda-1,L})$ and obtain $\vec\psi^-$ from $\PHI'$ by deleting all variables from $V(\vec\phi'_{\Lambda-1,L})$ and all clauses from $F(\vec\phi'_{\Lambda,L})$. 
  Further, let $\Lambda^-=\Lambda^+-\Lambda$ and let $\vec\psi^+$ be the sub-formula of $\vec\psi^-$ comprising all clauses and variables 
  of $\vec\psi^-$ with distance at most $2\Lambda^-$ from $\cV^+$ (see Figure \ref{fig_form} below).
\begin{figure}
  \begin{tikzpicture}[
  >=Latex,
  dot/.style={circle,fill=black,inner sep=1.8pt},
  leaf/.style={draw,rectangle,minimum size=4.2pt},
  ghost/.style={draw=black!40},
  faint/.style={draw=black!35},
  arr/.style={dotted,->,line width=.5pt}
]

\def\L{5}                 
\def\N{6}                 
\def\BN{9}                 
\def\R{5.9}               
\def\angA{200}            
\def\angB{340}            
\def\topY{3.9}            
\def\ovalW{6.8}           
\def\ovalH{1.2}           

\def\W{6.5}     
\def\H{0.6}     

\foreach \k in {1,...,\L}{
  \pgfmathsetmacro{\x}{-0.5*\ovalW + (\k-0.5)*\ovalW/\L}
  \node[circle,fill=gray] (t\k) at (\x,\topY) {};
}
\node[above=0.1pt of t1] { $1$};
\node[above=0.1pt of t2] { $2$};
\node[above=0.1pt of t\L] { $L$};
\node at ($(t2)!0.50!(t\L)$) 
{$\cdots$};

\draw[rounded corners=15pt,dashed]
  ($(-0.5*\ovalW,\topY+ 0.6)$) rectangle
  ($(0.5*\ovalW,\topY- 0.5)$);

\foreach \p in {t1}{
  \draw[ghost, thick] (\p) -- ++(-0.35,-0.85) node[fill=black,minimum size=8pt] {};
  \draw[ghost, thick] (\p) -- ++(0.00,-0.9) node[fill=black,minimum size=8pt] {};
  \draw[ghost, thick] (\p) -- ++(0.35,-0.9) node[fill=black,minimum size=8pt] {};
}

\foreach \p in {t\L}{
  \draw[ghost, thick] (\p) -- ++(-0.35,-0.9) node[fill=black,minimum size=8pt] {};
  \draw[ghost, thick] (\p) -- ++(0.00,-0.9) node[fill=black,minimum size=8pt] {};
  \draw[ghost, thick] (\p) -- ++(0.35,-0.9) node[fill=black,minimum size=8pt] {};
  \draw[ghost, thick] (\p) -- ++(0.7,-0.85) node[fill=black,minimum size=8pt] {};
}
\foreach \p in {t2}{
  \draw[ghost, thick] (\p) -- ++(-0.2,-0.9) node[fill=black,minimum size=8pt] {};
  \draw[ghost, thick] (\p) -- ++(0.2,-0.9) node[fill=black,minimum size=8pt] {};
}

\draw[faint, dotted, line width=1pt,domain=-\W+0.5:\W,samples=60,smooth,variable=\x]
  plot ({\x},{\H*(1.2*\x/\W)^2});

\foreach \i in {1,...,\BN}{
  \pgfmathsetmacro{\x}{-\W+0.5 + 1.*(\i-1)*2*(\W-0.5)/(\BN-1)}
  \pgfmathsetmacro{\y}{1+\H*(1.2*\x/\W)^2}
  \node[fill=black, minimum size=8pt] (b\i) at (\x,\y) {};
}

\foreach \i in {1,...,\N}{
  \pgfmathsetmacro{\x}{-\W+1 + 1.*(\i-1)*2*(\W-1)/(\N-1)}
  \pgfmathsetmacro{\y}{\H*(1.2*\x/\W)^2}
  \node[fill=black, circle] (b\i) at (\x,\y) {};
}

\node[above right=8pt, fill=white] at ($(b\N)+(-5.8,-1.05)$) {$ \mathcal{V}^{+}$};

\foreach \i in {1,6}{
  \foreach \dx in {-0.35,0,0.35}{
    \draw[thick] (b\i) -- ++(\dx,-0.85) node[ fill=black,minimum size=8pt] {};
  }
}

\foreach \i in {2,4}{
  \foreach \dx in {-0.25,0.25}{
    \draw[thick] (b\i) -- ++(\dx,-0.85) node[ fill=black,minimum size=8pt] {};
  }
}

\foreach \i in {5}{
  \foreach \dx in {-0.7,-0.35,0,0.35, 0.7}{
    \draw[thick] (b\i) -- ++(\dx,-0.85) node[ fill=black,minimum size=8pt] {};
  }
}
\draw[faint, dotted, line width=1pt,domain=-\W+0.5:\W,samples=60,smooth,variable=\x]
  plot ({\x},{-4+\H*(1.2*\x/\W)^2});


\draw[<->, line width=0.05mm] ($(t2)+(0.72,-0.5)$) -- ($(t2)+(0.72,-3.85)$) node[midway, fill=white] {$2\Lambda$};

\draw[<->,line width=0.05mm] ($(t2)+(3.25,-0.5)$) -- ($(t2)+(3.25,-7.82)$) node[midway, fill=white, yshift=-40pt] {$2\Lambda^+$};

\draw[<->, line width=0.05mm] ($(t2)+(-0.72,-3.83)$) -- ($(t2)+(-0.72,-7.78)$) node[midway, fill=white] {$2\Lambda^-$};

%
%
\draw[decorate,decoration={mirror, brace,amplitude=6pt}] (-7.22,4.2) -- (-7.22,-3)
  node[midway,xshift=-18pt,rotate=0] {$\pmb{\phi}'_{\Lambda^{+},L}$};
\draw[decorate,decoration={brace, mirror ,amplitude=6pt}] (-6.3,4.2) -- (-6.3,0.8)
  node[midway,xshift=-16pt,rotate=0] {$\pmb{\phi}'_{\Lambda,L}$};

\draw[decorate,decoration={brace,amplitude=6pt}] (7,0.9) -- (7,-3)
  node[midway,xshift=14pt,rotate=0] {$\pmb{\psi}^{+}$};
\draw[decorate,decoration={brace, amplitude=6pt}] (7.7,1.18) -- (7.7,-6)
  node[midway,xshift=14pt,rotate=0] {$\pmb{\psi}^{-}$};

\node at (-2.5,2.3) {$\vdots$};
\node at (0.0,2.1) {$\vdots$};
\node at (2.5,2.3) {$\vdots$};

\node at (-5,-1.7) {$\vdots$};
\node at (-2.9,-2) {$\vdots$};
\node at (0.0,-2.2) {$\vdots$};
\node at (2.5,-2) {$\vdots$};
\node at (5.0,-1.7) {$\vdots$};


\node at (-6,-4.7) {$\vdots$};
\node at (-3.5,-5) {$\vdots$};
\node at (-1.0,-5.2) {$\vdots$};
\node at (1.0,-5.2) {$\vdots$};
\node at (3.5,-5) {$\vdots$};
\node at (6,-4.7) {$\vdots$};
\end{tikzpicture}
  \caption{A sketch depicting the subformulas $\vec\psi^+, \vec\psi^-, \vec\phi'_{\Lambda,L}$, and $\vec\phi'_{\Lambda^+,L}$ of $\PHI'$ constructed above.}
  \label{fig_form}
\end{figure}
  If $\vX_{\Lambda,L}=\vX_{\Lambda^+,L}=1-L$ then
  \begin{align}\label{eqlem_unicyc_1}
    |V(\vec\phi'_{\Lambda,L})|+|F(\vec\phi'_{\Lambda,L})|-|E(\vec\phi'_{\Lambda,L})|&=L-1,\\
    |V(\vec\psi^+)|+|F(\vec\psi^+)|-|E(\vec\psi^+)|&=|\cV^+|.\label{eqlem_unicyc_2}
  \end{align}
  Moreover, \Lem~\ref{lem_excess} shows that for suitable $c_5'=c_5'(d,k,c_5),c_6'=c_6'(d,k,c_6)>0$ we have
  \begin{align}\label{eqlem_unicyc_3}
    \pr\brk{|V(\vec\phi'_{\Lambda,L})|+|F(\vec\phi'_{\Lambda,L})|>\log^{c_5'}n}&= O(n^{-2}),\\
    \label{eqlem_unicyc_4}
    \pr\brk{|V(\vec\phi'_{\Lambda^+,L})|+|F(\vec\phi'_{\Lambda^+,L})|>\log^{c_6'}n}&= O(n^{-2}).
  \end{align}
  
  Let $\hat\cL\subset\bar\cL$ be the set of literals $l$ that were added to the output set $\bar\cL$ by Step~7 of \pulp\ by way of clauses $a\in F(\vec\phi'_{\Lambda,L})$, i.e., at distance less than $2\Lambda$ from the initial set $\cL$.
  Let $\cV^-=\{|l|:l\in\hat\cL\}\cap\cV^+$ be the set of all variables at distance $2\Lambda$ from $\cL$ in $\PHI'$ that underlie a literal from $\hat\cL$.
  If $\vX_{\Lambda,L}=1-L$, the variables and clauses at distance at most $2\Lambda$ from $\cL$ do not cause \pulp\ to run into a contradiction, because each clause contains $k\geq3$ literals.
  Therefore, there does not exist a variable $x$ such that both $x$ and $\neg x$ belong to $\hat\cL$.
  Hence, because the signs of $\PHI'$ are uniformly random and \pulp\ proceeds in a BFS order, we may assume without loss of generality that $\hat\cL$ contains positive literals only.
  Thus,
  \begin{align}\label{eqlem_unicyc_10}
    \cV^-=\hat\cL\cap\cV^+\subset\cV^+.
  \end{align}
  
  We now apply \Lem~\ref{lem_noexcess} to the random formula $\vec\psi^-$.
  Specifically, let $\bar\cL^+$ be the output of \pulp\ on the formula $\vec\psi^-$ with the start set $\cL^+$ comprised 
  by the positive literals of $\cV^+$.
  Further, let $\fE$ be the event that 
  $|V(\vec\phi'_{\Lambda,L})|+|F(\vec\phi'_{\Lambda,L})|\leq\log^{c_5'}n$. Let
  $$\vX^+=|E(\vec\psi^+)|-|V(\vec\psi^+)|-|F(\vec\psi^+)|$$
  be the excess of $\vec\psi^+$.
  Since $0<d<\duniq(k)$, given $\fE$ the formula $\vec\psi^-$ has the same distribution as a random $k$-CNF with $n^-=n-O(\log^cn)$ variables and $\vm^-\disteq\Po(d^-n^-/k)$ random clauses, with $0<d^-=d+o(1)<\duniq(k)$.
  Hence, assuming that $c_6=c_6(d,k)> c_5'$ is sufficiently large, \Lem~\ref{lem_noexcess} shows that
  \begin{align}\label{eqlem_unicyc_5}
    \pr\brk{\vecone\fE\cap\{\vX^+=-|\cV^+|\}\cdot|\bar\cL^+|>u}&\leq\pr\brk{\sum_{1\leq i\leq\log^{c_5'}n}|\overline{\{\root\}}_{\TT[i]}|>u}+O(n^{-2})&&(u>0).
  \end{align}
  Combining~\eqref{eqlem_unicyc_5} with \Cor~\ref{lem:SMSizeGivenHeight}, we conclude that for a large enough $c_7=c_7(d,k) > c_6'$,
  \begin{align}\label{eqlem_unicyc_6}
    \pr\brk{\vecone\fE\cap\{\vX^+=-|\cV^+|\}\cdot|\bar\cL^+|>\log^{c_7}n}&=O(n^{-2}).
  \end{align}
  In light of above, we see that
  \begin{align*}
    \ex\brk{\vecone\{\vX_{\Lambda,L}=\vX_{\Lambda^+,L}=1-L\}|\bar\cL|^{3/2}}
      &\le\ex\brk{\vecone\fE\cap\{\vX_{\Lambda,L}=\vX_{\Lambda^+,L}=1-L\}|\bar\cL|^{3/2}} + (2n)^{3/2}(1-\Pr\brk{\fE}) ,
        &&\mbox[ \text{since } |\bar\cL|\le 2n]\\
        &\le\ex\brk{\vecone\fE\cap\{\vX_{\Lambda,L}=\vX_{\Lambda^+,L}=1-L\}(|\bar\cL^+| + \log^{c'_5}n)^{3/2}} + o(1) ,
        &&\mbox[ \text{from \eqref{eqlem_unicyc_3}} ]\\
      &\le\ex\brk{\vecone\fE\cap\{\vX^+=-|\cV^+|\}(|\bar\cL^+| + \log^{c'_5}n)^{3/2}} + o(1) ,
        &&\mbox[ \text{from \eqref{eqlem_unicyc_1}, \eqref{eqlem_unicyc_2}} ]\\
        &\le\pr\brk{\vecone\fE\cap\{\vX^+=-|\cV^+|\}\cdot|\bar\cL^+|>\log^{c_7}n} (2n)^{3/2}  \\
        &\quad + \pr\brk{\{\vX_{\Lambda,L}=1-L\}} (2\log^{c_7}n)^{3/2} +o(1)
      &&\mbox[ \text{total probability} ]  \\
      &= o(1)      &&\mbox[ \text{from }\eqref{eqlem_unicyc_6}, \text{\Lem}~\ref{lem_excess} ]  
  \end{align*}
  completing the proof.
\end{proof}

\begin{proof}[Proof of \Lem~\ref{lem_EIchi}]
  The assertion is immediate from \Cor~\ref{lem:coup}, \Lem~\ref{lem_unicyc}, \Lem~\ref{lem_excess} and the deterministic bound $|\bar\cL|\leq 2n$.
\end{proof}

\section{Proof of \Prop~\ref{prop_arnab}}\label{sec_prop_arnab}

\noindent
Let $\pi^{(\ell)}_{d,k}=\BP_{d,k}^{\ell}(\delta_{1/2})$ be the distribution obtained after $\ell$ iterations of $\BP_{d,k}(\cdot)$, with the convention $\pi^{(0)}_{d,k}=\delta_{1/2}$.
We recall $(\MU_{\pi,i,j})_{i,j\geq1}$ signify independent random variables with distribution $\pi$.

\begin{fact}\label{prop_identical}
For all $\ell \geq 0$ the random variables $\MU_{\pi^{(\ell)}_{d,k},1,1}$ and $1-\MU_{\pi^{(\ell)}_{d,k},1,1}$ are identically distributed.
\end{fact}

\begin{proof}
This is an immediate consequence of the fact that the random variables $\vd^+,\vd^-$ from the definition~\eqref{eqBPop}--\eqref{eqhat} are identically distributed.
\end{proof}

While the following is a direct consequence of the fact that Belief Propagation is `exact on trees' (see~\cite[\Chap~14]{MM} for precise statements), we carry out a detailed proof for the sake of completeness.
Following the conventions from \Sec~\ref{sec_results_num_sol}, we continue to denote by $\TAU^{(\ell)}$ a random satisfying assignment of the $k$-CNF $\TT^{(\ell)}=\TT^{(\ell)}_{d,k}$.

\begin{fact}\label{fac_BPexact}
For all $\ell\ge 0$, $d>0$ we have $\pr\brk{\TAU^{(\ell)}(\root)=1\mid\TT}\disteq\pi^{(\ell)}_{d,k}\enspace.$
\end{fact}
\begin{proof}
We proceed by induction on $\ell$.
As $\pi^{(0)}_{d,k}=\delta_{1/2}$, for $\ell=0$ there is nothing to show.
To go from $\ell-1$ to $\ell\geq1$, for a clause $a\in\partial_\TT \root$ and a variable $y\in\partial_\TT a\setminus\cbc\root$ let $\TT_{y\to a}$ be the component of the forest $\TT-a$ obtained by removing clause $a$ that contains variable $y$.
We consider $y$ the root of $\TT_{y\to a}$.
Further, obtain $\TT_{y\to a}^{(\ell-1)}$ from $\TT_{y\to a}$ by deleting all clauses and variables at a distance greater than $2(\ell-1)$ from $y$.
Additionally, for $s\in\PM$ let
\begin{align}\label{eq_fac_BPexact_0}
\vZ^{(\ell)}(s)&=\abs{\cbc{\sigma\in S(\TT^{(\ell)}):\sigma(\root)=s}},&
\vZ_{y\to a}^{(\ell-1)}(s)&=\abs{\cbc{\sigma\in S(\TT_{y\to a}^{(\ell-1)}):\sigma(y)=s}}.
\end{align}
In words, $\vZ^{(\ell)}(s)$ is the number of satisfying assignments of $\TT^{(\ell)}$ that set the root $\root$ to $s$, and $\vZ_{y\to a}^{(\ell-1)}(s)$ is the corresponding quantity for the sub-tree $\TT^{(\ell-1)}_{y\to a}$.

Clearly, setting $\root$ to $s\in\PM$ immediately satisfies all clauses $a\in\partial_\TT^s\root$.
By contrast, once $\root$ is assigned the value $+1$ each clause $a\in\partial_\TT^{-s}\root$ needs to be satisfied by setting some other variable $y\in\partial_\TT a\setminus\cbc\root$ to the value $\sign(y,a)$.
Hence,
\begin{align}\label{eq_fac_BPexact_1}
\vZ^{(\ell)}(s)&=\brk{\prod_{a\in\partial^+_{\TT}\root}\prod_{y\in\partial_{\TT}a\setminus\cbc\root}\sum_{t\in\PM}\vZ^{(\ell-1)}_{y\to a}\bc t}\cdot\brk{\prod_{a\in\partial^-_\TT\root} \bc{\prod_{y\in\partial_{\TT}a\setminus\cbc\root}\sum_{t\in\PM}\vZ^{(\ell-1)}_{y\to a}\bc t-\prod_{y\in\partial_{\TT}a\setminus\cbc\root}\vZ^{(\ell-1)}_{y\to a}\bc{-\sign(y,a)}}}.
\end{align}
Furthermore, the definition of the Galton-Watson tree $\TT$ ensures that the sub-trees $\TT_{y\to a}^{(\ell-1)}$ are independent copies of $\TT^{(\ell-1)}$.
Hence, by induction we have
\begin{align}\label{eq_fac_BPexact_2}
\frac{\vZ_{y\to a}^{(\ell-1)}(1)}{\sum_{s\in\PM}\vZ_{y\to a}^{(\ell-1)}(s)}&\disteq\pi_{d,k}^{(\ell-1)}&&\mbox{for all }a\in\partial_{\TT}\root,\,y\in\partial_{\TT}a\setminus\cbc\root,
\end{align}
and the random variables $\vZ_{y\to a}^{(\ell-1)}(1)/\sum_{s\in\PM}\vZ_{y\to a}^{(\ell-1)}(s)$ are mutually independent.
Combining~\eqref{eq_fac_BPexact_1}--\eqref{eq_fac_BPexact_2} with Fact~\ref{prop_identical}, we finally obtain
\begin{align*}
\pr\brk{\TAU^{(\ell)}(\root)=1\mid\TT}&=\frac{\vZ^{(\ell)}(1)}{\sum_{s\in\PM}\vZ^{(\ell)}(s)}\disteq\frac{\prod_{i=1}^{\vd^-}\brk{1-\prod_{j=1}^{k-1}\MU_{\pi_{d,k}^{(\ell-1)},2i-1,j}}}{\prod_{i=1}^{\vd^-}\brk{1-\prod_{j=1}^{k-1}\MU_{\pi_{d,k}^{(\ell-1)},2i-1,j}}+\prod_{i=1}^{\vd^+}\brk{1-\prod_{j=1}^{k-1}\MU_{\pi_{d,k}^{(\ell-1)},2i,j}}}\disteq\pi_{d,k}^{(\ell)},
\end{align*}
thereby completing the induction.
\end{proof}

Combining the combinatorial interpretation of the distributions $\pi_{d,k}^{(\ell)}$ with the Gibbs uniqueness property, we proceed to show that the sequence $(\pi_{d,k}^{(\ell)})_\ell$ converges in the weak topology.
To this end, it suffices to show that the sequence is Cauchy with respect to the Wasserstein $W_1$ metric.

\begin{lemma}\label{lem_weak_pi}
If $d < \duniq(k)$ then $(\pi^{(\ell)}_{d,k})_{\ell\geq0}$ is a $W_1$-Cauchy sequence.
\end{lemma}
\begin{proof}
If $d<\duniq(k)$ then the random tree $\TT=\TT_{d,k}$ enjoys the Gibbs uniqueness property; hence, ~\eqref{eqTreeUniq} is satisfied.
Consequently, given $0<\eps<1$ we can choose $\ell_0=\ell_0(d,k,\eps)>0$ large enough so that the event
\begin{align*}
\fU_{\eps,\ell}=\cbc{\max_{\tau\in S(\TT^{(\ell)})}\abs{\pr\brk{\TAU^{(\ell)}(\root)=1\mid\TT}-\pr\brk{\TAU^{(\ell)}(\root)=1\mid\TT,\,\forall x\in\partial^{2\ell}\root:\TAU^{(\ell)}(x)=\tau(x)} }>\eps}
\end{align*}
has probability
\begin{align}\label{eq_lem_weak_pi_1}
\pr\brk{\fU_{\eps,\ell}}&<\eps \enspace,&&\mbox{for all }\ell\geq\ell_0.
\end{align}

Now suppose that $\ell_0\leq\ell<\ell'$.
Let $\TAU^{(\ell)},\TAU^{(\ell')}$ be independent uniformly random satisfying assignments of $\TT^{(\ell)}$ and $\TT^{(\ell')}$, respectively.
We claim that
\begin{align}\label{eq_lem_weak_pi_2}
\pr\brk{\abs{\pr\brk{\TAU^{(\ell)}(\root)=1\mid\TT}-\pr\brk{\TAU^{(\ell')}(\root)=1\mid\TT}}>\eps}&<\eps.
\end{align}
To see this, let $\TAU^{(\ell,\ell')}=(\TAU^{(\ell')}(x))_{x\in\partial^{2\ell}_{\TT}\root}$ comprise the truth values that $\TAU^{(\ell')}$ assigns to the variables at distance exactly $2\ell$ from $\root$.
Then
\begin{align*}
\pr\brk{\TAU^{(\ell')}(\root)=1\mid\TT}&=
\ex\brk{\pr\brk{\TAU^{(\ell')}(\root)=1\mid\TT,\TAU^{(\ell,\ell')}}\mid\TT}\\&=
\ex\brk{\pr\brk{\TAU^{(\ell)}(\root)=1\mid\TT,\TAU^{(\ell,\ell')},\,\forall x\in\partial^{2\ell}\root:\TAU^{(\ell)}(x)=\TAU^{(\ell,\ell')}_x} \mid\TT}.
\end{align*}
Hence, for every $T \in \fU_{\eps,\ell}$ we have
\begin{align}\label{eq_lem_weak_pi_3}
\abs{\pr\brk{\TAU^{(\ell')}(\root)=1\mid\TT = T}-\pr\brk{\TAU^{(\ell)}(\root)=1\mid\TT = T}}&\leq\eps.
\end{align}
Thus, \eqref{eq_lem_weak_pi_2} follows from~\eqref{eq_lem_weak_pi_1} and~\eqref{eq_lem_weak_pi_3}.

Finally, since Fact~\ref{fac_BPexact} demonstrates that $\pr\brk{\TAU^{(\ell)}(\root)=1\mid\TT = T}\disteq\pi_{d,k}^{(\ell)}$ and 
$\pr\brk{\TAU^{(\ell')}(\root)=1\mid\TT = T}\disteq\pi_{d,k}^{(\ell')}$, \eqref{eq_lem_weak_pi_2} shows that
\begin{align*}
W_1\bc{\pi_{d,k}^{(\ell)},\pi_{d,k}^{(\ell')}}&<2\eps&&\mbox{	for all $\ell_0\leq\ell<\ell'$.}
\end{align*}
Hence, the sequence $(\pi_{d,k}^{(\ell)})_\ell$ is Cauchy.
\end{proof}

We are left to bound the lower tail of the limiting distribution $\pi_{d,k}=\lim_{\ell \to \infty}\pi^{(\ell)}_{d,k}$.

\begin{lemma}\label{lem_secmu}
If $d<\duniq(k)$ then $\ex\log^2\MU_{\pi_{d,k},1,1} < \infty$ .
\end{lemma}
\begin{proof}
We are going to bound $\ex\log^2\MU_{\pi_{d,k}^{(\ell)},1,1}$ and subsequently invoke the monotone convergence theorem to complete the proof.
First, we note that for all $\ell \ge 0$ we have
\begin{align*}
  \ex\log^2\MU_{\pi^{(\ell)}_{d,k},1,1} &= \ex\log^2\frac{Z(\TT^{(\ell)}, \{\root\})}{Z(\TT^{(\ell)})}&&\mbox{[by \Fact~\ref{fac_BPexact}]}\\
                                        &\le \Erw \brk{|\overline{\{\root\}}_{\TT}|^2}&&\mbox{[by \Lem~\ref{lem_Ichi}]}\\
&<\infty&&\mbox{[by \Cor~\ref{cor_c3}]}.
\end{align*}
Since $\pi_{d,k}$ is the weak limit of $(\pi_{d,k}^{(\ell)})_\ell$, we conclude that for any $N\in\NN$,
\begin{align}\label{eq_lem_secmu_1}
\ex\brk{ N \wedge \log^2\MU_{\pi_{d,k},1,1} } = \lim_{\ell \to \infty} \ex\brk{N \wedge \log^2\MU_{\pi^{(\ell)}_{d,k},1,1} } \le\Erw \brk{|\overline{\{\root\}}_{\TT}|^2}<\infty.
\end{align}
Finally,  applying the monotone convergence theorem to the limit $N\to\infty$, we see that the uniform bound~\eqref{eq_lem_secmu_1} implies the assertion.
\end{proof}

\begin{proof}[Proof of \Prop~\ref{prop_arnab}]
In light of Fact~\ref{prop_identical} and \Lem s~\ref{lem_weak_pi} and \ref{lem_secmu}, it only remains to show that
\begin{align*}
\ex\left|\log\bc{\prod_{i=1}^{\vd^-}\MU_{\pi_{d,k},2i}+\prod_{i=1}^{\vd^+}\MU_{\pi_{d,k},2i-1}}\right| < \infty
&&\text{ and }&&
\ex\left|\log\bc{1-\prod_{j=1}^k\MU_{\pi_{d,k},1,j}}\right| < \infty
\enspace.
\end{align*}
Recall the definition of $\MU_{\pi_{d,k},i}$ from \eqref{eqhat}. Using \Fact~\ref{prop_identical} and \Lem~\ref{lem_secmu}, we obtain
\begin{align*}
\ex\left|\log\bc{\prod_{i=1}^{\vd^-}\MU_{\pi_{d,k},2i}+\prod_{i=1}^{\vd^+}\MU_{\pi_{d,k},2i-1}}\right|
\le \log(2) +   \ex\left|\log{\prod_{i=1}^{\vd^-}\MU_{\pi_{d,k},2i}}\right| 
&\le  \log(2) +  \frac{d}{2} \ex\left|\log \MU_{\pi_{d,k},1}\right|\\
&\le  \log(2) +  \frac{d}{2} \sqrt{ \ex\left|\log^2 {\MU_{\pi_{d,k},1,1}}\right|} < \infty
\enspace,
\end{align*}
yielding the first inequality. Similarly, invoking \Fact~\ref{prop_identical} and \Lem~\ref{lem_secmu} for the second l.h.s. above gives
\begin{align*}
\ex\left|\log\bc{1-\prod_{j=1}^k\MU_{\pi_{d,k},1,j}}\right| 
\le \ex\left|\log\bc{1-\MU_{\pi_{d,k},1,1}}\right|
= \ex\left|\log\bc{\MU_{\pi_{d,k},1,1}}\right|
\le \sqrt{\ex\left|\log^2{\MU_{\pi_{d,k},1,1}}\right|}< \infty
\enspace,
\end{align*}
thereby completing the proof.
\end{proof}

\section{Proof of \Cor~\ref{cor_interpolation}}\label{sec_prop_interpolation}

\noindent
In order to turn the estimate of the expectation of $\log 1\vee Z(\PHI)$ provided by \Prop~\ref{prop_aizenman} into a `with high probability' statement, we harness a `soft' version of the $k$-SAT problem where violated clauses are discouraged but not strictly forbidden.
To be precise, for a $k$-CNF $\Phi$ and a real $\beta>0$ define
\begin{align}\label{eqZbeta_exp}
Z_\beta(\Phi)&=\sum_{\sigma\in\PM^{V(\Phi)}}\prod_{a\in F(\Phi)}\exp(-\beta\vecone\{\sigma\not\models a\}).
\end{align}
Thus, each satisfying assignment contributes one to the sum on the r.h.s.\ of~\eqref{eqZbeta_exp},
while the contribution of assignments that violate a number $M$ of clauses equals $\exp(-\beta M)$.
The value $Z_\beta(\PHI)$, called the {\em partition function of the random $k$-SAT model at inverse temperature $\beta$}, has received a considerable
amount of attention in the mathematical physics literature (see, e.g.,~\cite{PanchenkoBook}).
Crucially, by means of an interpolation argument~\cite{FranzLeone,Guerra} it is possible to prove the following.

\begin{theorem}[{\cite[\Thm~1]{PanchenkoTalagrand}}]\label{thm_interp}
For any $k\geq 3$, any $\beta>0$ and any probability measure $\pi$ on $[0,1]$ we have
\begin{align}\label{eqthm_interp}
\frac1n\ex\brk{\log Z_\beta(\PHI)}&\leq
\ex\brk{\log\bc{\prod_{i=1}^{\vd^-}\MU_{\beta, \pi,2i}+\prod_{i=1}^{\vd^+}\MU_{\beta, \pi,2i-1}}
-\frac{d(k-1)}{k}\log\bc{1-\bc{1- e^{-\beta}}\prod_{j=1}^k\MU_{\pi,1,j}}},&&\mbox{where}\\
\MU_{\beta,\pi,i}&=1-(1-\exp(-\beta))\prod_{j=1}^{k-1}\MU_{\pi,i,j}&&\mbox{ for $i\geq1$}.\nonumber
\end{align}
\end{theorem}

\noindent
We emphasise that the bound~\eqref{eqthm_interp} holds for any $n\geq k$ {\em without} an error term.
We also notice that by the monotone convergence theorem for the measure $\pi=\pi_{d,k}$ from \Thm~\ref{thm_main} we have
\begin{align}\nonumber
\lim_{\beta\to\infty}&
\ex\brk{\log\bc{\prod_{i=1}^{\vd^-}\MU_{\beta, \pi_{d,k},2i}+\prod_{i=1}^{\vd^+}\MU_{\beta, \pi_{d,k},2i-1}}
-\frac{d(k-1)}{k}\log\bc{1-\bc{1- e^{-\beta}}\prod_{j=1}^k\MU_{\pi_{d,k},1,j}}}\\
&= \ex\brk{\log\bc{\prod_{i=1}^{\vd^-}\MU_{ \pi_{d,k},2i}+\prod_{i=1}^{\vd^+}\MU_{ \pi_{d,k},2i-1}}
-\frac{d(k-1)}{k}\log\bc{1-\prod_{j=1}^k\MU_{\pi_{d,k},1,j}}} = \fB_{d,k}(\pi_{d,k})
\enspace.
\label{eq_dom}
\end{align}

The reason why we proceed by way of the `soft' model with $\beta<\infty$ is that for this model a routine application of Azuma-Hoeffding
implies the following concentration bound.

\begin{lemma}\label{lem_AH}
For any fixed $\beta>0$ we have $\pr\brk{\abs{\log Z_\beta(\PHI)-\ex\log Z_\beta(\PHI)}>\sqrt n\log n}=o(1/n).$
\end{lemma}
\begin{proof}
The clauses of the random formula $\PHI$ are drawn independently, and adding or removing a single clause can alter the value of $\log Z_\beta(\nix)$
by no more than $\pm\beta$.
\end{proof}

\begin{proof}[Proof of \Cor~\ref{cor_interpolation}]
We proceed with a proof by contradiction. In particular, towards a contradiction, assume there exists an $\eps > 0$ such that 
for infinitely many $n \ge 1$ we have
\begin{align}\label{eq_fin_contra}
  \Pr\brk{\frac1n\log Z(\PHI) > \fB_{d,k}(\pi_{d,k}) + \eps }
    > \eps \enspace.
\end{align}
Moreover, by
\eqref{eq_dom} we can find a $\beta_0 >0$ such that for every $\beta \ge \beta_0$ we have 
\begin{align}\label{eq_beta_to_Bethe}
\left|\ex\brk{\log\bc{\prod_{i=1}^{\vd^-}\MU_{\beta, \pi_{d,k},2i}+\prod_{i=1}^{\vd^+}\MU_{\beta, \pi_{d,k},2i-1}}
-\frac{d(k-1)}{k}\log\bc{1-\bc{1- e^{-\beta}}\prod_{j=1}^k\MU_{\pi_{d,k},1,j}}} -   \fB_{d,k}(\pi_{d,k}) \right| < \eps/3 \enspace.
\end{align}
Invoking \Lem~\ref{lem_AH} for $\beta = \beta_0$ and sufficiently large $n$ gives
\begin{align}\label{eq_lem_inv}
  \Pr \brk{\frac1n\log Z_{\beta_{0}}(\PHI) >\frac1n \ex{\log Z_{\beta_0}(\PHI)} + \eps/3} \le \eps/3
  \enspace.
\end{align}
The definition~\eqref{eqZbeta_exp} of the partition function ensures that $Z_\beta(\PHI)\geq Z(\PHI)$ for all $\beta>0$.
Therefore, combining \eqref{eq_fin_contra}--\eqref{eq_lem_inv}, and \Thm~\ref{thm_interp} we see that for
large enough $n$ the following holds with probability at least $1 - \frac{2}{3} \eps$: 
\begin{align*}
  \frac1n \log Z(\PHI) \le \frac1n \log Z_{\beta_0}(\PHI)
  \le \frac1n \ex{\log Z_{\beta_0}(\PHI)} + \frac{\eps}{3} \le \fB_{d,k}(\pi_{d,k}) + \frac{2}{3} \eps
  \enspace,
\end{align*}
contradicting our assumption, and thus completing the proof.
\end{proof}

\section{Proof of \Prop~\ref{prop_aizenman}}\label{sec_prop_aizenman}

\noindent
In this section we prove \Prop s~\ref{lem_PHI''} and \ref{lem_PHI'''}, which in light of \Fact~\ref{fact_coupling},
imply \Prop~\ref{prop_aizenman}.
Both proofs follow a similar structure and make use of \Prop s~\ref{prop_bad} and~\ref{prop_benign}, which we therefore prove first.

\subsection{Proof of \Prop~\ref{prop_bad}}\label{sec_prop_bad}

We show that both terms of \eqref{eq_prop_bad} have finite expectation.
Let us begin with the first one.

\begin{lemma}\label{lem_bad''}
If $d<\duniq(k)$ then $\ex\brk{\left|\log\frac{Z(\PHI'')\vee 1}{Z(\PHI')\vee 1}\right|^{3/2}}=O(1)$.
\end{lemma}
\begin{proof}
Since $\PHI''$ is obtained from $\PHI'$ by adding clauses, we have
\begin{align}\label{eq_lem_PHI''_1_1}
0\leq Z(\PHI'')\leq Z(\PHI')\leq 2^n.
\end{align}
Hence,
\begin{align}\label{eq_lem_bad''_1}
\log\frac{Z(\PHI'')\vee 1}{Z(\PHI')\vee 1}=0&&\mbox{ if }Z(\PHI')=0.
\end{align}
Therefore, we may assume from now on that $\PHI'$ is satisfiable.

The number $\vDelta''\disteq\Po(d(k-1)/k)$ of new clauses is a Poisson variable with bounded mean.
Therefore, Bennett's inequality shows that $\pr\brk{\vDelta''>\log n}=O(n^{-2})$.
Since~\eqref{eq_lem_PHI''_1_1} shows that $|\log((Z(\PHI'')\vee 1)/(Z(\PHI')\vee 1))|^{3/2}\leq n^{3/2}$, we conclude that
\begin{align}\label{eq_lem_PHI''_1_4}
\ex\brk{\vecone\cbc{\vDelta''>\log n}\cdot \left|\log\frac{Z(\PHI'')\vee 1}{Z(\PHI')\vee 1}\right|^{3/2}}&=o(1).
\end{align}

Further, let $c_1,\ldots,c_{\vDelta''}$ be the new clauses added by {\bf CPL2}.
Let $\vx_{1,1},\ldots,\vx_{1,k},\ldots,\vx_{\vDelta'',1},\ldots,\vx_{\vDelta'',k}$ be their constituent variables and let $\cX=\{\vx_{1,1},\ldots,\vx_{1,k},\ldots,\vx_{\vDelta'',1},\ldots,\vx_{\vDelta'',k}\}$.
Since the clauses $c_1,\ldots,c_{\vDelta''}$ are chosen uniformly and independently, a routine balls-into-bins consideration shows that
\begin{align}\label{eq_lem_PHI''_1_2}
\pr\brk{|\cX|\leq k(\vDelta''-1)\mid\vDelta''\leq\log n}&=\tilde O(n^{-2}).
\end{align}

Now, consider the `good' event
\begin{align*}
\fG&=\cbc{Z(\PHI')>0,\,\vDelta''\leq\log n,\,|\cX|> k(\vDelta''-1)}.
\end{align*}
Combining~\eqref{eq_lem_PHI''_1_1}--\eqref{eq_lem_PHI''_1_2}, we see that
\begin{align}\label{eq_lem_PHI''_1_3}
\ex\brk{(1-\vecone\fG)\cdot\left|\log\frac{Z(\PHI'')\vee 1}{Z(\PHI')\vee 1}\right|^{3/2}}&=o(1).
\end{align}
Hence, we are left to bound $\ex[\vecone\fG\cdot|\log((Z(\PHI'')\vee 1)/(Z(\PHI')\vee1))|^{3/2}]$.
If $\fG$ occurs and thus $|\cX|>k(\vDelta''-1)$, then there exists a set of literals $\cL\subset\cbc{\vx_{1,1},\neg\vx_{1,1},\ldots,\vx_{1,k},\neg\vx_{1,k},\ldots,\vx_{\vDelta'',1},\neg\vx_{\vDelta'',1},\ldots,\vx_{\vDelta'',k},\neg \vx_{\vDelta'',k}}$ such that
\begin{itemize}
\item every clause $c_i$ contains a literal from $\cL$ ($1\leq i\leq\vDelta''$), and
\item there does not exist $x\in\cX$ such that $x\in\cL$ and $\neg x\in\cL$.
\end{itemize}
Moreover, on $\fG$ we have $|\cL|\leq|\cX|\leq k\log n$.
Let $\bar\cL=\bar\cL_{\PHI'}$ be the output of \pulp\ on $(\PHI',\cL)$.
Then \Lem~\ref{lem_Ichi} shows that
\begin{align}\label{eq_lem_PHI''_1_3a}
\ex\brk{\vecone\fG\cdot\left|\log\frac{Z(\PHI'')\vee 1}{Z(\PHI')\vee 1}\right|^{3/2}}&\leq\ex\brk{\vecone\fG\cdot\abs{\bar\cL}^{3/2}}.
\end{align}
Furthermore, since by {\bf CPL2} the new clauses $c_1,\ldots,c_{\vDelta''}$ are chosen independently of the formula $\PHI'$, \Lem~\ref{lem_EIchi} implies that there exists $C=C(d,k)>0$ such that
\begin{align}\label{eq_lem_PHI''_1_3b}
\ex\brk{\vecone\fG\cdot\abs{\bar\cL}^{3/2}\mid\vDelta''}&\leq C\cdot\bc{\vDelta''}^{3/2}.
\end{align}
Combining~\eqref{eq_lem_PHI''_1_3a}--\eqref{eq_lem_PHI''_1_3b} and recalling that $\vDelta''\disteq\Po(d(k-1)/k)$, we obtain
\begin{align}\label{eq_lem_PHI''_1_3c}
\ex\brk{\vecone\fG\cdot\left|\log\frac{Z(\PHI'')\vee 1}{Z(\PHI')\vee 1}\right|^{3/2}}&=O(1).
\end{align}
Finally, the assertion follows from~\eqref{eq_lem_PHI''_1_3} and \eqref{eq_lem_PHI''_1_3c}.
\end{proof}

We move on to the second term of \eqref{eq_prop_bad}.

\begin{lemma}\label{lem_bad'''}
If $d<\duniq(k)$ then $\ex\brk{\left|\log\frac{Z(\PHI''')\vee 1}{Z(\PHI')\vee 1}\right|^{3/2}}=O(1)$.
\end{lemma}
\begin{proof}
We proceed similarly as in the proof of \Lem~\ref{lem_bad''}.
The construction in {\bf CPL3} ensures that $\PHI'''$ contains one additional variable $x_{n+1}$ and $\vDelta'''\disteq\Po(d)$ new clauses $b_1,\ldots,b_{\vDelta'''}$ that each contain $x_{n+1}$ and $k-1$ other variables.
Let $\vx_{1,1},\ldots,\vx_{1,k-1},\ldots,\vx_{\vDelta''',1},\ldots,\vx_{\vDelta''',k-1}\in\{x_1,\ldots,x_n\}$ be the variables among $x_1,\ldots,x_n$ that appear in $b_1,\ldots,b_{\vDelta'''}$ and let $\cX=\{\vx_{1,1},\ldots,\vx_{\vDelta''',k-1}\}$.
Then
\begin{align}\label{eq_lem_PHI'''_1_1}
0\leq Z(\PHI'')\leq 2Z(\PHI')\leq 2^{n+1}.
\end{align}
Hence, if $\PHI'$ is unsatisfiable, then so is $\PHI'''$ and thus
\begin{align}\label{eq_lem_PHI'''_1_1a}
\log\frac{Z(\PHI''')\vee 1}{Z(\PHI')\vee1}&=0&&\mbox{ if }Z(\PHI')=0.
\end{align}
Furthermore, since $\vDelta'''\disteq\Po(d)$, Bennett's inequality shows that $\pr\brk{\vDelta'''>\log n}=O(n^{-2})$.
Therefore, \eqref{eq_lem_PHI'''_1_1} shows that
\begin{align}\label{eq_lem_PHI'''_1_2}
\ex\brk{\vecone\cbc{\vDelta'''>\log n}\cdot \left|\log\frac{Z(\PHI''')\vee 1}{Z(\PHI')\vee1}\right|^{3/2}}&=o(1).
\end{align}
Moreover, since the $k-1$ variables among $x_1,\ldots,x_n$ that appear in the clauses $b_1,\ldots,b_{\Delta'''}$ are chosen uniformly and independently, a simple balls-into-bins argument shows that
\begin{align}\label{eq_lem_PHI'''_1_2_a}
\pr\brk{|\cX|\leq(k-1)(\vDelta''-1)\mid\vDelta''\leq\log n}&=\tilde O(n^{-2}).
\end{align}

Hence, consider the event
\begin{align*}
\fG&=\cbc{Z(\PHI')>0,\,\vDelta'''\leq\log n,\,|\cX|>(k-1)(\vDelta''-1)}.
\end{align*}
Combining~\eqref{eq_lem_PHI'''_1_1}--\eqref{eq_lem_PHI'''_1_2_a}, we obtain
\begin{align}\label{eq_lem_PHI'''_1_2_b}
\ex\brk{(1-\vecone\fG)\cdot \left|\log\frac{Z(\PHI''')\vee 1}{Z(\PHI')\vee1}\right|^{3/2}}&=o(1).
\end{align}
Furthermore, if the event $\fG$ occurs, then there exists a set $\cL\subseteq\{x,\neg x:x\in\cX\}$ of literals such that each clause $b_i$, $1\leq i\leq\vDelta'''$, contains a literal $l\in\cL$ and such that $\{x,\neg x\}\not\subseteq\cL$ for all $x\in\cX$.
Hence, with $\bar\cL=\bar\cL_{\PHI'}$ the output of \pulp\ on $(\PHI',\cL)$, \Lem~\ref{lem_Ichi} shows that
\begin{align}\label{eq_lem_PHI'''_1_2_c}
\ex\brk{\vecone\fG\cdot \left|\log\frac{Z(\PHI''')\vee 1}{Z(\PHI')\vee1}\right|^{3/2}}&\leq\ex\brk{\vecone\fG\cdot|\bar\cL|^{3/2}}.
\end{align}
Furthermore, since the clauses $b_1,\ldots,b_{\vDelta'''}$ are drawn independently of $\PHI'''$, \Lem~\ref{lem_EIchi} shows that there exists $C=C(d,k)>0$ such that
\begin{align}\label{eq_lem_PHI'''_1_2_d}
\ex\brk{\vecone\fG\cdot|\bar\cL|^{3/2}\mid\vDelta'''}&\leq C\cdot(\vDelta''')^{3/2}.
\end{align}
Finally, since $\vDelta'''\disteq\Po(d)$, the assertion follows from~\eqref{eq_lem_PHI'''_1_2_b}, \eqref{eq_lem_PHI'''_1_2_c} and~\eqref{eq_lem_PHI'''_1_2_d}.
\end{proof}

\begin{proof}[Proof of \Prop~\ref{prop_bad}]
The proposition follows immediately from \Lem s~\ref{lem_bad''}--\ref{lem_bad'''}.
\end{proof}

\subsection{Proof of \Prop~\ref{prop_benign}}\label{sec_prop_benign}

Let $\pi_{d,k}^{(\ell)}=\BP_{d,k}^\ell(\delta_{1/2})$ be the result of an $\ell$-fold application of the operator $\BP_{d,k}$ from~\eqref{eqBPop} to the point mass at $1/2$.
Also recall from~\eqref{eqempirical} that $\vec\pi_n'$ denotes the empirical distribution 
of the marginals $(\pr[\SIGMA_{\PHI'}(x_i)=1\mid\PHI'])_{1\leq i\leq n}$.

\begin{lemma}\label{lem_benign}
Suppose that $d<\duniq(k)$.
For any $\eps>0$ there exists $\ell_0=\ell_0(d,k,\eps)>0$ such that for all $\ell\geq\ell_0$ we have $$\ex[W_1(\vec\pi_n',\pi_{d,k}^{(\ell)})\mid Z(\PHI')>0]<\eps+o(1).$$
\end{lemma}
\begin{proof}
Assume that $\ell\geq\ell_0$ for a large enough $\ell_0=\ell_0(d,k,\eps)>0$.
Since $d<\duniq(k)$ and since $\TT=\TT_{d,k}$ is a Galton-Watson tree in which every variable node has $\Po(d)$ clause nodes as offspring and the offspring of every 
clause node consists of $k-1$ variable nodes, there exists a set $\cT_\ell$ of trees, with $|\cT_\ell|=O(1)$, such that the following hold:
\begin{description}
\item[T0] for every $T\in\cT_\ell$ we have $\pr\brk{\TT^{(\ell)}=T}>0$.
\item[T1] $\pr\brk{\TT^{(\ell)}\in\cT_\ell}>1-\eps$.
\item[T2] given $\TT^{(\ell)}\in\cT_\ell$ we have
\begin{align*}
\max_{\tau\in S(\TT^{(\ell)})}\abs{\pr\brk{\TAU^{(\ell)}(\root)=1\mid\TT^{(\ell)}} -\pr\brk{\TAU^{(\ell)}(\root)=1\mid\TT^{(\ell)},\,\forall x\in\partial^{2\ell}\root:\TAU^{(\ell)}(x)=\tau(x)} }&<\eps.
\end{align*}
\end{description}

For a variable node $x_i$ of $\PHI'$ obtain $\vec\phi'_{\ell}(x_i)$ from $\PHI'$ by deleting all variables and clauses at distance greater than $2\ell$ from $x_i$.
We consider $x_i$ being the root of $\vec\phi'_{\ell}(x_i)$.
Moreover, for a tree $T\in\cT_\ell$ let $\cV_T$ be the set of variable nodes $x_i$, $1\leq i\leq n$, such that $\vec\phi'_{\ell}(x_i)\ism T$; thus, there is an isomorphism of the CNFs $T$ and $\vec\phi'_\ell(x_i)$ that maps the root $\root$ of $T$ to $x_i$.
Consider the event
\begin{align}\label{eq_lem_benign_100}
\fT_\ell&=\cbc{\sum_{T\in\cT_\ell}\abs{\pr\brk{\TT^{(\ell)}\ism T}-|\cV_T|/n}<\eps}.
\end{align}
Then \Cor~\ref{cor_lcwk} implies that
\begin{align}\label{eq_lem_benign_1}
\pr\brk{\fT_\ell}&=1-o(1)&&\mbox{for every }\ell\geq0.
\end{align}

We now claim that
\begin{align}\label{eq_lem_benign_4}
\abs{\pr\brk{\TAU^{(\ell)}(\root)=1\mid\TT^{(\ell)}=T}-\pr\brk{\SIGMA_{\PHI'}(x_i)=1\mid \PHI'}}<\eps&&\mbox{ for all }T\in\cT_\ell,\,x_i\in\cV_T.
\end{align}
To see this, let $S_\ell(\PHI',x_i)$ be the set of all assignments $\sigma\in\PM^{\partial^{2\ell}x_i}$ of the variables at distance $2\ell$ from $x_i$ in $\PHI'$ such that there exists a satisfying assignment $\sigma'\in S(\PHI')$ with $\sigma'(y)=\sigma(y)$ for all $y\in\partial^{2\ell}x_i$.
Then the law of total probability shows that
\begin{align}\label{eq_lem_benign_2}
\pr\brk{\SIGMA_{\PHI'}(x_i)=1\mid \PHI'}&=\sum_{\sigma\in S_\ell(\PHI',x_i)}\pr\brk{\SIGMA_{\PHI'}(x_i)=1\mid \PHI',\,\forall y\in \partial^{2\ell}x_i : \SIGMA_{\PHI'}(y)=\sigma(y)}\pr\brk{\forall y\in \partial^{2\ell}x_i : \SIGMA_{\PHI'}(y)=\sigma(y)\mid\PHI'}.
\end{align}
Further, since for $T\in\cT_\ell$ and $x_i\in\cV_T$ we have $\vec\phi_\ell'(x_i)\ism T$, condition {\bf T2} implies that
\begin{align}\label{eq_lem_benign_3}
\abs{\pr\brk{\SIGMA_{\PHI'}(x_i)=1\mid \PHI',\,\forall y\in \partial^{2\ell}x_i : \SIGMA_{\PHI'}(y)=\sigma(y)}-\pr\brk{\TAU^{(\ell)}(\root)=1\mid \TT^{(\ell)}=T}}&<\eps.
\end{align}
Combining~\eqref{eq_lem_benign_2} and~\eqref{eq_lem_benign_3}, we obtain~\eqref{eq_lem_benign_4}.

To complete the proof, we recall from Fact~\ref{fac_BPexact} that $\pi_{d,k}^{(\ell)}$ is precisely the distribution of $\pr\brk{\TAU^{(\ell)}(\root)=1\mid\TT^{(\ell)}}$.
Therefore, coupling the formulas $\TT^{(\ell)}, \PHI'$ on the event $\fT_\ell$ we have

\begin{align*}
W_1(\vec\pi_n',\pi_{d,k}^{(\ell)})
&\leq\pr\brk{\TT^{(\ell)}\not\in\cT_\ell}+\frac1n\sum_{T\in\cT_\ell}\sum_{x\in\cV_T}\abs{\pr\brk{\TAU^{(\ell)}(\root)=1\mid\TT^{(\ell)}=T}-\pr\brk{\SIGMA_{\PHI'}(x)=1\mid \PHI'}}+\eps&&\mbox{[by~\eqref{eq_lem_benign_100}]}\\
&\leq3\eps&&\mbox{[by {\bf T1} and \eqref{eq_lem_benign_4}].}
\end{align*}
Combining this bound with~\eqref{eq_lem_benign_1} completes the proof.
\end{proof}

\begin{proof}[Proof of \Prop~\ref{prop_benign}]
The first assertion follows from \Prop~\ref{prop_arnab}, \Lem~\ref{lem_benign} and
the fact that, since $0<d < \duniq(k) <\dsat(k)$, we have that $\pr\brk{Z(\PHI')>0}=1-o(1)$. 

The second follows a 
routine argument, which we present below for the case $\ell = 2$ and it is standard to extend to any finite $\ell$ (see \cite[Proposition 2.5]{COP}). 
Let $t= \Theta(\log\log n)$ and recall the definitions of $\vec\phi'_{t}(x_i)$, $\cT_{t}$ and $S_t(\PHI',x_i)$ from the proof of \Lem~\ref{lem_benign}. 
Consider the event 
$\mathfrak{D} = \{\vec\phi'_{t}(x_1), \vec\phi'_{t}(x_2) \text{ are disjoint tree formulas}\}$.

From \Lem~\ref{lem_excess}, we have that $\Pr\brk{\mathfrak{D}} = 1 - o(1)$. On the event 
$\mathfrak{D}$, \Lem~\ref{lem_benign} implies that for every $\sigma_1, \sigma_2 \in \{\pm 1\}$, and $\tau_1 \in S_t(\PHI',x_1),
\tau_2 \in S_t(\PHI',x_2)$ we have
\begin{align}\label{eq_fromabove}
  |\pr\brk{\SIGMA_{\PHI'}(x_i)=\sigma_i\mid \PHI',\,\forall y\in \partial^{2t}x_i : \SIGMA_{\PHI'}(y)=\tau_i(y)} - \pr\brk{\SIGMA_{\PHI'}(x_i)=\sigma_i\mid \PHI'}| = o(1)\enspace,
  &&\text{ for } i=1,2\enspace.
\end{align}
Therefore, from the law of total probability and the triangle inequality we see that for every $\sigma_1,\sigma_2 \in \{\pm 1\}$
\begin{align*}
  &\abs{\pr\brk{\SIGMA_{\PHI'}(x_1)=\sigma_1, \SIGMA_{\PHI'}(x_2)=\sigma_2\mid\PHI'}-\pr\brk{\SIGMA_{\PHI'}(x_1)=\sigma_1\mid\PHI'}\pr\brk{\SIGMA_{\PHI'}(x_2)=\sigma_2\mid\PHI'}} \\
  &\le \Big|\abs{\pr\brk{\SIGMA_{\PHI'}(x_1)=\sigma_1\mid \PHI', \SIGMA_{\PHI'}(x_2)= \sigma_2 } - \ex_{\tau_1, \tau_2}\brk{\pr\brk{\SIGMA_{\PHI'}(x_1)=\sigma_1, \SIGMA_{\PHI'}(x_2)=\sigma_2\mid\PHI',\tau_1, \tau_2}}}\\
  &- \abs{\ex_{\tau_1, \tau_2}\brk{\pr\brk{\SIGMA_{\PHI'}(x_1)=\sigma_1\mid \PHI',\tau_1}\pr\brk{\SIGMA_{\PHI'}(x_2)=\sigma_2\mid\PHI', \tau_2}}-\pr\brk{\SIGMA_{\PHI'}(x_2)=\sigma_2\mid\PHI'}\pr\brk{\SIGMA_{\PHI'}(x_2)=\sigma_2\mid\PHI'}} \Big|  \\
  &= \Big|\abs{\ex_{\tau_1, \tau_2}\brk{\pr\brk{\SIGMA_{\PHI'}(x_1)=\sigma_1\mid \PHI',\tau_1}\pr\brk{\SIGMA_{\PHI'}(x_2)=\sigma_2\mid\PHI', \tau_2}}-\pr\brk{\SIGMA_{\PHI'}(x_2)=\sigma_2\mid\PHI'}\pr\brk{\SIGMA_{\PHI'}(x_2)=\sigma_2\mid\PHI'}} \Big| \\
  &\le\ex_{\tau_1}\abs{\pr\brk{\SIGMA_{\PHI'}(x_1)=\sigma_1 \mid \PHI',\tau_1} - \pr\brk{\SIGMA_{\PHI'}(x_1)=\sigma_1 \mid \PHI'}} + \ex_{\tau_2}\abs{\pr\brk{\SIGMA_{\PHI'}(x_2)=\sigma_2\mid \PHI',\tau_2} - \pr\brk{\SIGMA_{\PHI'}(x_2)=\sigma_2\mid \PHI'}} \\
  &= o(1). &&\!\!\!\!\!\!\!\!\!\!\mbox{[by \eqref{eq_fromabove}]}
\end{align*}
Summing over the four sign combinations of $\sigma_1, \sigma_2$ gives the desired result.
\end{proof}

\subsection{Proof of Proposition~\ref{lem_PHI''}}\label{sec_lem_PHI''}

As in the proof of \Lem~\ref{lem_bad''} let $c_1,\ldots,c_{\vDelta''}$ be the new clauses added by {\bf CPL2} and let $\vx_{1,1},\ldots,\vx_{1,k},\ldots,\vx_{\vDelta'',1},\ldots,\vx_{\vDelta'',k}$ be their constituent variables.
Let $\cX=\{\vx_{1,1},\ldots,\vx_{1,k},\ldots,\vx_{\vDelta'',1},\ldots,\vx_{\vDelta'',k}\}$.
For $\eps > 0$ and $z \in \RR$ define $\lambda_\eps(z) = \log(z\vee\eps)$.
Finally, let $(\vs_i)_{i\geq0}$ be a sequence of uniformly random $\pm1$-valued random variables, mutually independent and independent of all other randomness.

\begin{lemma}\label{lem_PHI''_2}
Assume that $d<\duniq(k)$.
There exists $B=B(d,k)>0$ such that for all $0<\eps<1$ we have
\begin{align*}
\limsup_{n\to\infty}\ex\brk{\bc{\sum_{i=1}^{\vDelta''}\lambda_\eps\bc{1-\prod_{j=1}^k
\pr\brk{\SIGMA(\vx_{i,j})\neq\sign(\vx_{i,j},c_i)\mid\PHI'}}}^2\mid Z(\PHI')>0}\leq B.
\end{align*}
\end{lemma}
\begin{proof}
Given $Z(\PHI')>0$ we have
\begin{align}\label{eq_lem_PHI''_2_2}
0&\geq\lambda_\eps\bc{1-\prod_{j=1}^k
\pr\brk{\SIGMA(\vx_{1,j})\neq\sign(\vx_{1,j}, c_1)\mid\PHI'}
}\geq
\lambda_\eps\bc{1-\pr\brk{\SIGMA(\vx_{1,1})\neq\sign(\vx_{1,1}, c_1)\mid\PHI'}}.
\end{align}
Recalling that $\vDelta''\disteq\Po(d(k-1)/k)$, we combine~\eqref{eq_lem_PHI''_2_2} with Cauchy-Schwarz to obtain $B'=B'(d,k)>0$ such that
\begin{align}\nonumber
\ex&\brk{\bc{\sum_{i=1}^{\vDelta''}\lambda_\eps\bc{1-\prod_{j=1}^k\pr\brk{\SIGMA(\vx_{i,j})\neq\sign(\vx_{i,j}, c_i)\mid\PHI'}}}^2\mid Z(\PHI')>0}\\
&\leq B'\cdot \ex\brk{\lambda_\eps\bc{1-\pr\brk{\SIGMA(\vx_{1,1})\neq\sign(\vx_{1,1}, c_1)\mid\PHI'}}^2\mid Z(\PHI')>0}\label{eq_lem_PHI''_2_1}.
\end{align}
Further, since the function $\lambda_\eps$ is bounded and continuous for every $\eps>0$ and since $\sign(\vx_{1,1}, c_1)$ is chosen independently of $\PHI'$, \Prop~\ref{prop_benign} shows that for any $\eps>0$,
\begin{align}\nonumber
\ex\brk{\lambda_\eps\bc{1-\pr\brk{\SIGMA(\vx_{1,1})\neq\sign(\vx_{1,1}, c_1)\mid\PHI'}}^2\mid Z(\PHI')>0}&= \ex\brk{\lambda_\eps\bc{\MU_{\pi_{d,k},1,1}}^2}+o(1) \\
&\leq\ex\brk{\log\bc{\MU_{\pi_{d,k},1,1}}^2}+o(1).\label{eq_lem_PHI''_2_3}
\end{align}
Since \Prop~\ref{prop_arnab} shows that $\ex\brk{\log\bc{\MU_{\pi_{d,k},1,1}}^2}=O(1)$, the assertion follows from \eqref{eq_lem_PHI''_2_1} and \eqref{eq_lem_PHI''_2_3}.
\end{proof}

\begin{lemma}\label{lem_PHI''_3}
Assume that $d<\duniq(k)$.
For any $\delta>0$ there exists $\eps_0>0$ such that for all $\eps_0>\eps>0$ we have
\begin{align*}
\limsup_{n\to\infty}\abs{\ex\brk{\log\frac{Z(\PHI'')\vee 1}{Z(\PHI')\vee 1}}-\frac{d(k-1)}k\ex\brk{\lambda_\eps\bc{1-\prod_{j=1}^k\pr\brk{\SIGMA(x_j)=\vs_j\mid\PHI'}}\mid Z(\PHI')>0}}<\delta.
\end{align*}
\end{lemma}
\begin{proof}
We choose small enough $\xi=\xi(d,k,\delta)>\zeta(\xi)>\eta=\eta(\zeta)>\eps_0=\eps_0(\eta)>0$, let $0<\eps<\eps_0$ and assume that $n\ge n_0(\eps)$ is large enough.
Also let $\gamma=\gamma(n)=o(1)$ be a sequence that tends to zero sufficiently slowly.
Additionally, let $\fE$ be the event that all of the following conditions occur.
\begin{description}
\item[E1] $Z(\PHI')>0.$
\item[E2] $\vDelta''\leq\zeta^{-1}.$
\item[E3] $|\cX|=k\vDelta''.$
\item[E4] $\max_{x\in\cX,s\in\PM}\pr[\SIGMA(x)=s\mid\PHI']\leq1-\eta.$
\item[E5] $\sum_{\tau\in\PM^{\cX}}\abs{\pr[\forall x\in\cX:\SIGMA(x)=\tau(x)\mid\PHI']-\prod_{x\in\cX}\pr[\SIGMA(x)=\tau(x)\mid\PHI']}<\gamma.$
\end{description}

We claim that
\begin{align}\label{eqlem_PHI''_3_1}
\pr\brk{\fE}&\geq1-2\xi+o(1).
\end{align}
Indeed, since $0<d < \duniq(k) <\dsat(k)$, we have that $\pr\brk{Z(\PHI')>0}=1-o(1)$.
Moreover, since $\vDelta^{\prime \prime}\disteq\Po(d(k-1)/k)$, Markov's inequality shows that $\pr\brk{\vDelta''>\zeta^{-1}}\leq\zeta d <\xi$.
Further, since the new clauses $c_1,\ldots,c_{\vDelta''}$ are chosen independently, we have $\pr\brk{|\cX|=k\vDelta''\mid\vDelta''\leq\zeta^{-1}}=1-O(1/n)$.

Moreover, per \Prop~\ref{prop_benign} we see that the joint distribution on the assignments over
$\cX$ must be approximately the product measure. The tails of the limiting distribution of the latter are 
controlled by \eqref{eq_prop_arnab_bound}. Therefore, for small enough $\eta$ we should have
\begin{align*}
\pr\brk{\max_{x\in\cX,s\in\PM}\pr[\SIGMA(x)=s\mid\PHI']\leq1-\eta\mid\vDelta''\leq\zeta^{-1},\,Z(\PHI')>0}\geq1-\xi
\enspace.
\end{align*}
Similarly, \Prop~\ref{prop_benign} shows together with Markov's inequality that
\begin{align*}
  \pr\brk{\mbox{{\bf E5} occurs}\mid\vDelta'' \le \zeta^{-1},\,Z(\PHI')>0}=1-o(1)
\enspace ,
\end{align*}
provided that $\gamma\to\infty$ sufficiently slowly.
Thus, we obtain \eqref{eqlem_PHI''_3_1}.

Furthermore, \eqref{eqlem_PHI''_3_1} implies together with
\Prop~\ref{prop_bad}
and H\"older's inequality that
\begin{align}\label{eqlem_PHI''_3_2}
\ex\abs{(1-\vecone\fE)\cdot\log\frac{Z(\PHI'')}{Z(\PHI')}}&\leq\delta/3+o(1),
\end{align}
provided that $\xi=\xi(d,k,\delta)>0$ is small enough.
Analogously, \eqref{eqlem_PHI''_3_1}, \Lem~\ref{lem_PHI''_2} and Cauchy-Schwarz yield
\begin{align}\label{eqlem_PHI''_3_3}
\ex\abs{(1-\vecone\fE)\lambda_\eps\bc{1-\prod_{j=1}^k\pr\brk{\SIGMA(x_j)=\vs_j\mid\PHI'}}}&\leq\delta/3+o(1).
\end{align}

Thus, we confine ourselves to the event $\fE$, on which we have $Z(\PHI'),Z(\PHI'')>0$ due to {\bf E1}, {\bf E3}, {\bf E4} and {\bf E5}.
Hence,
\begin{align}\nonumber
\log\frac{Z(\PHI'')\vee1}{Z(\PHI')\vee1}&=\log\frac{Z(\PHI'')}{Z(\PHI')}=
\log\sum_{\tau\in\PM^{\cX}}\vecone\cbc{\tau\models c_1,\ldots,c_{\vDelta''}}\pr[\forall x\in\cX:\SIGMA(x)=\tau(x)\mid\PHI']\\
&=\log\sum_{\tau\in\PM^{\cX}}\vecone\cbc{\tau\models c_1,\ldots,c_{\vDelta''}}\prod_{x\in\cX}\pr[\SIGMA(x)=\tau(x)\mid\PHI']+o(1)&&\mbox{[by {\bf E4, E5}]}\nonumber\\
&=\sum_{i=1}^{\vDelta''}\log\brk{1-\prod_{j=1}^k\pr\brk{\SIGMA(\vx_{i,j})\neq\sign(\vx_{i,j},c_i)\mid\PHI'}} +o(1)	&&\mbox{[by {\bf E3}].}
\label{eqlem_PHI''_3_4}
\end{align}
Further, {\bf E4} ensures that for any $1\leq i\leq\Delta''$,
\begin{align}
\abs{\log\brk{1-\prod_{j=1}^k\pr\brk{\SIGMA(\vx_{i,j})\neq\sign(\vx_{i,j},c_i)\mid\PHI'}}-\lambda_\eps\brk{1-\prod_{j=1}^k\pr\brk{\SIGMA(\vx_{i,j})\neq\sign(\vx_{i,j}, c_i)\mid\PHI'}}}<\xi.
\label{eqlem_PHI''_3_5}
\end{align}
Thus, combining \eqref{eqlem_PHI''_3_4} and~\eqref{eqlem_PHI''_3_5}, we obtain
\begin{align}
\ex\abs{\vecone\fE\bc{\log\frac{Z(\PHI'')\vee1}{Z(\PHI')\vee1}-\sum_{i=1}^{\vDelta''}\lambda_\eps\bc{1-\prod_{j=1}^k\pr\brk{\SIGMA(\vx_{i,j})\neq\sign(\vx_{i,j}, c_i)\mid\PHI'}}}}&<\delta/3+o(1).
\label{eqlem_PHI''_3_6}
\end{align}
Further, combining \eqref{eqlem_PHI''_3_2} and \eqref{eqlem_PHI''_3_6} with \Lem~\ref{lem_PHI''_2}, we obtain
\begin{align}\label{eqlem_PHI''_3_7}
\abs{\ex\brk{\log\frac{Z(\PHI'')\vee1}{Z(\PHI')\vee1}}-\ex\brk{\sum_{i=1}^{\vDelta''}\lambda_\eps\bc{1-\prod_{j=1}^k\pr\brk{\SIGMA(\vx_{i,j})\neq\sign(\vx_{i,j},c_i)\mid\PHI'}}\mid Z(\PHI')>0}}&<\delta+o(1).
\end{align}
Finally, since the clauses $c_1,\ldots,c_{\vDelta''}$ are drawn uniformly and independently and since the distribution of $\PHI'$ is invariant under permutation of the variable nodes, we find
\begin{align}\nonumber
\ex&\brk{\sum_{i=1}^{\vDelta''}\lambda_\eps\bc{1-\prod_{j=1}^k\pr\brk{\SIGMA(\vx_{i,j})\neq\sign(\vx_{i,j},c_i)\mid\PHI'}\mid Z(\PHI')>0}}\\
&=\frac{d(k-1)}{k}\ex\brk{\lambda_\eps\bc{1-\prod_{j=1}^k\pr\brk{\SIGMA(x_j)=\vs_j}\mid\PHI'}\mid Z(\PHI')>0}.\label{eqlem_PHI''_3_8}
\end{align}
Combining~\eqref{eqlem_PHI''_3_7} and~\eqref{eqlem_PHI''_3_8} completes the proof.
\end{proof}

\begin{proof}[Proof of \Prop~\ref{lem_PHI''}]
\Prop~\ref{prop_benign} shows together with \Lem~\ref{lem_PHI''_3} that
\begin{align}\label{eq_lem_PHI''_A}
\ex\brk{\log\frac{Z(\PHI'')\vee 1}{Z(\PHI')\vee 1}}&=\frac{d(k-1)}k\ex\brk{\lambda_\eps\bc{1-\prod_{j=1}^k\MU_{\pi_{d,k},1,j}}}+o_\eps(1),
\end{align}
with $o_\eps(1)$ hiding a term that vanishes in the limit $\eps\to0$.
Furthermore, in light of~\eqref{eq_prop_arnab_bound} the monotone convergence theorem yields
\begin{align}\label{eq_lem_PHI''_B}
\ex\brk{\log\bc{1-\prod_{j=1}^k\MU_{\pi_{d,k},1,j}}}&=\lim_{\eps\to0}\ex\brk{\lambda_\eps\bc{1-\prod_{j=1}^k\MU_{\pi_{d,k},1,j}}}.
\end{align}
The assertion follows from~\eqref{eq_lem_PHI''_A} and~\eqref{eq_lem_PHI''_B}.
\end{proof}

\subsection{Proof of \Prop~\ref{lem_PHI'''}}\label{sec_lem_PHI'''}
We adapt the steps from \Sec~\ref{sec_lem_PHI''} to the coupling of $\PHI'$, $\PHI'''$.
Recall that the latter is obtained by adding to $\PHI'$ a single variable $x_{n+1}$ along with $\vDelta'''$ clauses $b_1,\ldots,b_{\vDelta'''}$ that each contain $x_{n+1}$ and $k-1$ other variables.
Thus, let $\vx_{1,1},\ldots,\vx_{1,k-1},\ldots,\vx_{\vDelta''',1},\ldots,\vx_{\vDelta''',k-1}\in\{x_1,\ldots,x_n\}$ be the variables other than $x_{n+1}$ that appear in $b_1,\ldots,b_{\vDelta'''}$ and let $\cX=\{\vx_{1,1},\ldots,\vx_{\vDelta''',k-1}\}$ be the set comprising all these variables.

\begin{lemma}\label{lem_PHI'''_2}
Assume that $0<d<\duniq(k)$.
There exists $B=B(d,k)>0$ such that for all $0<\eps<1$ we have
\begin{align*}
  \limsup_{n\to\infty}\ex\brk{ \lambda_\eps\bc{\sum_{s\in\PM}\prod_{i=1}^{\vDelta'''}\bc{1-\vecone\{\sign(x_{n+1},b_i)\neq s\}\prod_{j=1}^{k-1}\pr[\SIGMA(\vx_{i,j})\neq\sign(\vx_{i,j}, b_i)\mid\PHI']}}^2 \mid Z(\PHI')>0}&\leq B.
\end{align*}
\end{lemma}
\begin{proof}
Given that $\PHI'$ is satisfiable, and noticing that $\lambda_\eps$ is increasing, and $\eps \in (0,1)$, we see that
\begin{align}
0\wedge\lambda_\eps&\bc{\sum_{s\in\PM}\prod_{i=1}^{\vDelta'''}\bc{1-\vecone\{\sign(x_{n+1},b_i)\neq s\}\prod_{j=1}^{k-1}\pr[\SIGMA(\vx_{i,j})\neq\sign(\vx_{i,j}, b_i)\mid\PHI']}}
                    \nonumber\\
                   &\ge \lambda_\eps\bc{\sum_{s\in\PM}\prod_{i=1}^{\vDelta'''}1-\prod_{j=1}^{k-1}\pr[\SIGMA(\vx_{i,j})\neq\sign(\vx_{i,j}, b_i)\mid\PHI']}
                   \nonumber\\
                   &\ge \lambda_\eps\bc{\prod_{i=1}^{\vDelta'''}1-\pr[\SIGMA(\vx_{i,1})\neq\sign(\vx_{i,1}, b_i)\mid\PHI']}
                   \nonumber \\
                   &= \lambda_\eps\bc{\prod_{i=1}^{\vDelta'''}\pr[\SIGMA(\vx_{i,1})=\sign(\vx_{i,1}, b_i)\mid\PHI']}\geq\sum_{i=1}^{\vDelta'''}\lambda_\eps(\pr[\SIGMA(\vx_{i,1})=\sign(\vx_{i,1}, b_i)\mid\PHI']).
                   \label{eq_le0leps}
\end{align}
We also notice that  
\begin{align}\label{eq_lepsle1}
  0 \vee \lambda_\eps&\bc{\sum_{s\in\PM}\prod_{i=1}^{\vDelta'''}\bc{1-\vecone\{\sign(x_{n+1},b_i)\neq s\}\prod_{j=1}^{k-1}\pr[\SIGMA(\vx_{i,j})\neq\sign(\vx_{i,j}, b_i)\mid\PHI']}} <1 \enspace.
\end{align}
In light of the above, we now bound 
\begin{align}
  \ex&\brk{ \lambda_\eps\bc{\sum_{s\in\PM}\prod_{i=1}^{\vDelta'''}\bc{1-\vecone\{\sign(x_{n+1},b_i)\neq s\}\prod_{j=1}^{k-1}\pr[\SIGMA(\vx_{i,j})\neq\sign(\vx_{i,j},b_i)\mid\PHI']}}^2 \mid Z(\PHI')>0}\nonumber\\
   &\leq\ex\brk{ 1+ \bc{\sum_{i=1}^{\vDelta'''}\lambda_\eps(\pr[\SIGMA(\vx_{i,1})=\sign(\vx_{i,1}, b_i)\mid\PHI']}^2\mid Z(\PHI')>0}&&\mbox{[from \eqref{eq_le0leps},\eqref{eq_lepsle1}]}\nonumber\\
   &\leq d(d+1)\ex\brk{1+\bc{\lambda_\eps(\pr[\SIGMA(\vx_{1,1})=\sign(\vx_{1,1}, b_1)\mid\PHI'])}^2\mid Z(\PHI')>0} &&\mbox{[$\vDelta'''\disteq\Po(d)$]}\nonumber\\
&\leq d(d+1)\bc{1+\ex\brk{\lambda_\eps(\pr[\SIGMA(\vx_{1,1})=\sign(\vx_{1,1}, b_1)\mid\PHI'])^2\mid Z(\PHI')>0}}.
\label{eq_lem_PHI'''_2_2}
\end{align}
Further, \Prop~\ref{prop_benign} implies that for any $\eps>0$,
\begin{align}
\ex\brk{\lambda_\eps(\pr[\SIGMA(\vx_{1,1})=\sign(\vx_{1,1}, b_1)\mid\PHI'])^2\mid Z(\PHI')>0}&= \ex\brk{\lambda_\eps(\MU_{\pi_{d,k},1,1})^2}+o(1)\leq \ex\brk{\log^2\MU_{\pi_{d,k},1,1}}+o(1).
\label{eq_lem_PHI'''_2_3}
\end{align}
Finally, the assertion follows from \eqref{eq_lem_PHI'''_2_2} and~\eqref{eq_lem_PHI'''_2_3}.
\end{proof}

\begin{lemma}\label{lem_PHI'''_3}
Assume that $0<d<\duniq(k)$.
For any $\delta>0$ there exists $\eps_0>0$ such that for all $\eps_0>\eps>0$ we have
\begin{align*}
  \limsup_{n\to\infty}&\Bigg|\ex\brk{\log\frac{Z(\PHI''')\vee 1}{Z(\PHI')\vee 1}}\\&-\ex\brk{\lambda_\eps\bc{\sum_{s\in\PM}\bc{\prod_{i=1}^{\vDelta'''}1-\vecone\{\sign(x_{n+1},b_i)\neq s\}\prod_{j=1}^{k-1}\pr[\SIGMA(\vx_{i,j})\neq\sign(\vx_{i,j}, b_i)\mid\PHI']}}^2 \mid Z(\PHI')>0}\Bigg|<\delta.
\end{align*}
\end{lemma}
\begin{proof}
Choose small enough $\xi=\xi(d,k,\delta)>\zeta(\xi)>\eta=\eta(\zeta)>\eps_0=\eps_0(\eta)>0$, let $0<\eps<\eps_0$, suppose that $n>n_0(\eps)$ is sufficiently large and let $0<\gamma=\gamma(n)=o(1)$ be a sequence that converges to zero slowly.
Let $\fE$ be the event that the following conditions occur.
\begin{description}
\item[E1] $Z(\PHI')>0$.
\item[E2] $\vDelta'''\leq\zeta^{-1}$.
\item[E3] $|\cX|=(k-1)\vDelta'''$.
\item[E4] $\max_{x\in\cX,s\in\PM}\pr[\SIGMA(x)=s\mid\PHI']\leq1-\eta$.
\item[E5] $\sum_{\tau\in\PM^{\cX}}\abs{\pr[\forall x\in\cX:\SIGMA(x)=\tau(x)\mid\PHI']-\prod_{x\in\cX}\pr[\SIGMA(x)=\tau(x)\mid\PHI']}<\gamma$.
\end{description}
As in the proof of \Lem~\ref{lem_PHI''_3} we find that
\begin{align}\label{eqlem_PHI'''_3_1}
\pr\brk{\fE}&\geq1-2\xi+o(1).
\end{align}

Let
\begin{align*}
  \vL_\eps&= \lambda_\eps\bc{\sum_{s\in\PM}\prod_{i=1}^{\vDelta'''}\bc{1-\vecone\{\sign(x_{n+1}, b_i)\neq s\}  \prod_{j=1}^{k-1}\pr[\SIGMA(\vx_{i,j})\neq\sign(\vx_{i,j}, b_i)\mid\PHI']}}
\end{align*}
for brevity.
Combining \Prop~\ref{prop_bad}, \Lem~\ref{lem_PHI'''_2} and \eqref{eqlem_PHI'''_3_1} and using H\"older's inequality, we obtain
\begin{align}\label{eqlem_PHI'''_3_2}
\ex\abs{(1-\vecone\fE)\log\frac{Z(\PHI'')}{Z(\PHI')}}&\leq\delta/3+o(1),&
\ex\abs{(1-\vecone\fE)\vL_\eps\mid Z(\PHI')>0}&\leq\delta/3+o(1).
\end{align}

Hence, we are left to compare $\ex\abs{\vecone\fE\cdot\log\frac{Z(\PHI'')}{Z(\PHI')}}$ and $\ex\abs{\vecone\fE\cdot\vL_\eps\mid Z(\PHI')>0}$.
On the event $\fE$ we have $Z(\PHI'),Z(\PHI''')>0$.
Consequently,
\begin{align}\nonumber
\log\frac{Z(\PHI''')\vee1}{Z(\PHI')\vee1}&=\log\frac{Z(\PHI''')}{Z(\PHI')}=
\log\sum_{\tau\in\PM^{\cX\cup\{x_{n+1}\}}}\vecone\cbc{\tau\models b_1,\ldots,b_{\vDelta'''}}\pr[\forall x\in\cX:\SIGMA(x)=\tau(x)\mid\PHI']\\
&=\log\sum_{\tau\in\PM^{\cX\cup\{x_{n+1}\}}}\vecone\cbc{\tau\models b_1,\ldots,b_{\vDelta'''}}\prod_{x\in\cX}\pr[\SIGMA(x)=\tau(x)\mid\PHI']+o(1)&&\mbox{[by {\bf E4, E5}]}\nonumber\\
&=\log\brk{\sum_{s\in\PM}\prod_{i=1}^{\vDelta'''}\bc{1-\vecone\cbc{\sign(x_{n+1}, b_i)\neq s}\prod_{j=1}^{k-1}\pr\brk{\SIGMA(\vx_{i,j})\neq\sign(\vx_{i,j},b_i)\mid\PHI'}}} +o(1)	&&\mbox{[by {\bf E3}].}
\label{eqlem_PHI'''_3_4}
\end{align}
Now, {\bf E4} guarantees that
\begin{align}
  \log\brk{\sum_{s\in\PM}\prod_{i=1}^{\vDelta'''}\bc{1-\vecone\cbc{\sign(x_{n+1},b_i)\neq s}\prod_{j=1}^{k-1}\pr\brk{\SIGMA(\vx_{i,j})\neq\sign(\vx_{i,j}, b_i)\mid\PHI'}}}&=\vL_\eps.
\label{eqlem_PHI'''_3_5}
\end{align}
Therefore, we combine \eqref{eqlem_PHI'''_3_4} and~\eqref{eqlem_PHI'''_3_5} to obtain
\begin{align}
\ex\abs{\vecone\fE\bc{\log\frac{Z(\PHI''')\vee1}{Z(\PHI')\vee1}-\vL_\eps}}&<\delta/3+o(1).
\label{eqlem_PHI'''_3_6}
\end{align}
Finally, the assertion follows from \eqref{eqlem_PHI'''_3_2} and \eqref{eqlem_PHI'''_3_6}.
\end{proof}

\begin{proof}[Proof of \Prop~\ref{lem_PHI'''}]
Following similar steps as in the proof of \Prop~\ref{lem_PHI''}, we see that the assertion follows from \Lem~\ref{lem_PHI'''_3}, \Prop~\ref{prop_arnab}, \Prop~\ref{prop_benign}, and the dominated convergence theorem.
\end{proof}

\begin{proof}[Proof of \Prop~\ref{prop_aizenman}]
Immediate from Fact~\ref{fact_coupling}, \Prop~\ref{lem_PHI''} and \Prop~\ref{lem_PHI'''}.
\end{proof}

\section
{Proof of \Prop~\ref{prop_BPPureLit} }\label{sec_prop:condMarg}

\subsection{Proof of \Lem~\ref{lem:ExtremalBnd}}
The proof is by induction on the height of the tree.
The following claim summarises the main step of the induction.

\begin{claim}\label{Lemma_Noela}
For all $\ell\geq0$, all variables $x$ of ${\TT}^{(\ell)}$ and all
satisfying assignments $\tau\in S({\TT}^{(\ell)})$ we have
\begin{align}\label{eqLemma_Noela1}
\frac{Z({\TT}_x^{(\ell)},\tau,\TAU^+(x))}{Z(\TT_x^{(\ell)},\tau)}
&\leq\frac{Z({\TT}_x^{(\ell)},\TAU^+,\TAU^+(x))}{Z(\TT_x^{(\ell)},\TAU^+)}.
\end{align}
\end{claim}

\begin{proof}
For boundary variables $x\in\partial^{2\ell} \root$ there is nothing to show
because the r.h.s.\ of \eqref{eqLemma_Noela1} equals one.
Hence, consider a variable $x\in\partial^{2q}\root$ for some $q<\ell$.
If $Z({\TT}_x^{(\ell)},\tau,\TAU^+(x))=0$, then \eqref{eqLemma_Noela1} is
trivially satisfied. Hence, assume that
$Z({\TT}_x^{(\ell)},\tau,\TAU^+(x))>0$.

Let $a_1^+,\ldots,a_g^+$ be the children (clauses) of $x$ with
$\sign(x,a_i^+)=\TAU^+(x)$. Also let
$y_{11},\ldots,y_{1(k-1)}, \ldots, y_{g1}, \ldots, y_{g(k-1)}$ be
the children (variables) of $a_1^+,\ldots,a_g^+$.
Similarly, let $a_1^-,\ldots,a_h^-$ be the children of $x$ with
$\sign(x,a_i^-)=-\TAU^+(x)$ and let
$z_{11},\ldots,z_{1(k-1)}, \ldots, z_{h1}, \ldots, z_{h(k-1)}$ be their
children.
We claim that for all $\tau\in S({\TT}^{(\ell)})$,
\begin{align}
Z({\TT}_x^{(\ell)},\tau,\TAU^+(x))
&=
\left(
\prod_{i=1}^g
\prod_{t=1}^{k-1}
Z({\TT}_{y_{it}}^{(\ell)},\tau)
\right)
\cdot
\prod_{j=1}^h
\left(
\prod_{t=1}^{k-1}
Z({\TT}_{z_{jt}}^{(\ell)},\tau)
-
\prod_{t=1}^{k-1}
Z({\TT}_{z_{jt}}^{(\ell)},\tau,-\TAU^+(z_{jt}))
\right)
\label{eq:SplusMax}
\enspace,
\\
Z({\TT}_x^{(\ell)},\tau,-\TAU^+(x))
&=
\prod_{i=1}^g
\left(
\prod_{t=1}^{k-1}
Z({\TT}_{y_{it}}^{(\ell)},\tau)
-
\prod_{t=1}^{k-1}
Z({\TT}_{y_{it}}^{(\ell)},\tau, \TAU^+(y_{it}))
\right)
\cdot
\left(
\prod_{j=1}^h
\prod_{t=1}^{k-1}
Z({\TT}_{z_{jt}}^{(\ell)},\tau)
\right)
\label{eq:SminusMax}
\enspace.
\end{align}
For setting $x$ to $\TAU^+(x)$ satisfies $a_1^+,\ldots,a_g^+$; hence,
arbitrary satisfying assignments of the sub-trees ${\TT}_{y_{it}}^{(\ell)}$
can be combined, which explains the first product in \eqref{eq:SplusMax}.
By contrast, upon assigning $x$ the value $\TAU^+(x)$ we need to ensure
that each of the clauses $a_1^-,\ldots,a_g^-$ are satisfied by at least
one variable other than $x$. This explains the second factor of
\eqref{eq:SplusMax}.
A similar argument yields \eqref{eq:SminusMax}.
Dividing \eqref{eq:SminusMax} by \eqref{eq:SplusMax} and invoking the
induction hypothesis (for $q+1$), we obtain
\begin{align*}
\frac{Z({\TT}_x^{(\ell)},\tau,-\TAU^+(x))}{Z(\TT_x^{(\ell)},\tau,\TAU^+(x))}
&=
\prod_{i=1}^g
\left(1-\prod_{t=1}^{k-1}
\frac{Z({\TT}_{y_{it}}^{(\ell)},\tau,\TAU^+(y_{it}))}
{Z({\TT}_{y_{it}}^{(\ell)},\tau)}
\right)
\cdot
\prod_{j=1}^h
\left(
1-
\prod_{t=1}^{k-1}
\frac{Z({\TT}_{z_{jt}}^{(\ell)},\tau, -\TAU^+(z_i))}
{Z({\TT}_{z_{jt}}^{(\ell)},\tau)}
\right)^{-1}
\\
&\ge
\prod_{i=1}^g
\left(1-\prod_{t=1}^{k-1}
\frac{Z({\TT}_{y_{it}}^{(\ell)},\TAU^+,\TAU^+(y_{it}))}
{Z({\TT}_{y_{it}}^{(\ell)},\TAU^+)}
\right)
\cdot
\prod_{j=1}^h
\left(
1-
\prod_{t=1}^{k-1}
\frac{Z({\TT}_{z_{jt}}^{(\ell)},\TAU^+, -\TAU^+(z_i))}
{Z({\TT}_{z_{jt}}^{(\ell)},\TAU^+)}
\right)^{-1}
=\frac{Z({\TT}_x^{(\ell)},\TAU^+,-\TAU^+(x))}
{Z({\TT}_x^{(\ell)},\TAU^+,\TAU^+(x))} \enspace ,
\end{align*}
completing the induction.
\end{proof}

\begin{proof}[Proof of \Lem~\ref{lem:ExtremalBnd}]
Applying Claim~\ref{Lemma_Noela} to $x=\root$ completes the proof of \Lem~\ref{lem:ExtremalBnd}.
\end{proof}

\subsection{Proof of Lemma \ref{lem_SmallOnTop}}\label{sec:lem:SmallOnTop}
We employ the {\pulp} algorithm introduced in \Sec~\ref{sec_pulp} and its analysis on the random tree from \Sec~\ref{sec_treepulp}.
Recall that given an initial set of literals $\cL$, {\pulp} returns a superset $\bar\cL$ with the property that the partial assignment obtained from setting all literals of  $\bar\cL$ to true, leaves no clause with only unsatisfying literals.
Let us write $\bar\cL = \bar\cL_{x,s}$ for the set returned by {\pulp} algorithm, initialized with the literal set $\cL =\{s \cdot x\}$.

\begin{claim}\label{cor_from_Ichi}
  Let $0\leq t<\ell$ and assume that $x\in \partial^{2t}_\TT\root$, $s \in \PM$, satisfy $|\bar\cL_{x,s}| < \ell - t$.
Then for all $\tau\in S(\TT^{(\ell)})$
\begin{align}
Z({\TT}_x^{(\ell)},\tau)\leq 2^{|\bar\cL_{x,s}|}\cdot Z(\TT_x^{(\ell)},\tau,s)
\enspace.
\end{align}
\end{claim}
\begin{proof}
Notice that under our assumption on the size of $\bar\cL_{x,s}$, the assignment $\tau$ does not clash with the one imposed by {\pulp}.
The assertion therefore follows immediately from 
the same argument as in the proof of \Lem~\ref{lem_Ichi}.
\end{proof}

\begin{claim}\label{cl_size_t_GWtree}
We have $\lim_{t \to \infty} \Pr\brk{ |\partial_{\TT}^{2t} \root |>(200d\cdot (k-1))^{t}} = 0$.
\end{claim}
\begin{proof}
This is an immediate consequence of \Lem~\ref{lem_crude}.
\end{proof}

\begin{proof}[Proof of \Lem~\ref{lem_SmallOnTop}]
Assume that $\ell>ct^c$ for a large enough $c=c(d,k)>0$ and that $t>t_0=t_0(d,k)$ is sufficiently large.
Then \Cor~\ref{lem:SMSizeGivenHeight} shows that
\begin{align}\label{eq_tail_Nct}
\Pr\brk{|\bar\cL_{\root,\pm 1}| \ge t^c} \le \exp(-t^2) \enspace.
\end{align}
Combining \Cl~\ref{cl_size_t_GWtree} with \eqref{eq_tail_Nct} and using the union bound, we obtain a sequence $\varepsilon_t \to 0$ such that
\begin{align}\label{eq_UnOnN}
\Pr \brk{ \forall x \in \partial^{2t}_\TT\root: |\bar\cL_{x,\pm 1}| < t^c} \ge 1-\eps_t
\enspace.
\end{align}
If $x \in \partial_\TT^{2t}\root$ satisfies $|\bar\cL_{x,\pm 1}|  < t^c$ and $\ell > ct^c$, then Claim~\ref{cor_from_Ichi} yields that for all $x\in \partial^{2t}\root$
\begin{align}\label{eq_Inter_diffN}
\left|\ETA^{(\ell)}_x\right| \le \log \frac{Z(\TT^{(\ell)}_x, \SIGMA^+)}{Z(\TT^{(\ell)}_x, \SIGMA^+, +1)} +
\log \frac{Z(\TT^{(\ell)}_x, \SIGMA^+)}{Z(\TT^{(\ell)}_x, \SIGMA^+, -1)}
\le |\bar\cL_{x, +1}| + |\bar\cL_{x, -1}|\le 2t^c \enspace.
\end{align}
The result now follows from \eqref{eq_UnOnN} and \eqref{eq_Inter_diffN}.
\end{proof}

\subsection{Proof of \Prop~\ref{prop_BPPureLit}}\label{sec_prop_BPPureLit}

We focus on the operator $\LDELitS{d,k}$ introduced in \Sec~\ref{sec_uniq_outline}.
Let $\rho = \left(\rho_{\AllKid}, \rho_{\PureP}, \rho_{\PureM} \right)$, and
${\rho'} = \left({\rho}'_{\AllKid}, {\rho}'_{\PureP}, {\rho}'_{\PureM}\right)$
be two arbitrary triplets in
$\cP (-\infty,\infty] \times \cP (0, +\infty] \times \cP(-\infty, 0] $, and write
$\hat\rho = \left(\hat\rho_{\AllKid}\hat\rho_{\PureP}, \hat\rho_{\PureM} \right)$ and 
$\hat\rho' = \left(\hat\rho'_{\AllKid}\hat\rho'_{\PureP}, \hat\rho'_{\PureM} \right)$ for the 
images $\LDELitS{d,k}(\rho)$ and $\LDELitS{d,k}(\rho')$, respectively.
We wish to bound $\dist_{d}\left(\hat\rho, \hat\rho'\right)$ in terms of $\dist_{d}\left(\rho, \rho'\right)$.

To this end, we begin with bounding the $W_1$-distance separately for each of the coordinates $(\hat\rho_{\PureP}, \hat\rho'_{\PureP}),
(\hat\rho_{\PureM}, \hat\rho'_{\PureM})$ and $(\hat\rho_{\AllKid}, \hat\rho'_{\AllKid})$. Observe that  
it is sufficient to consider only $W_1(\hat\rho_{\PureP}, \hat\rho'_{\PureP})$ and $W_1(\hat\rho_{\PureM}, \hat\rho'_{\PureM})$, 
as the triangle inequality implies that $W_1(\hat\rho_{\AllKid}, \hat\rho'_{\AllKid}) \le 
W_1(\hat\rho_{\PureP}, \hat\rho'_{\PureP})+ W_1(\hat\rho_{\PureM}, \hat\rho'_{\PureM})$.

To spell out our bounds, we need to introduce some additional notation. Recall that for 
$i,j \ge 1$ the random variables $\ETA_{\AllKid,i,j}$, $\ETA_{\PureP,i,j}$, $\ETA_{\PureM,i,j}$ follow the law of 
$\rho_{\AllKid}$, $\rho_{\PureP}$, $\rho_{\PureM}$, respectively. Similarly, let $\ETA'_{\AllKid,i,j}$, $\ETA'_{\PureP,i,j}$, $\ETA'_{\PureM,i,j}$ 
be random variables with law $\rho'_{\AllKid}, \rho'_{\PureP}$ and $\rho'_{\PureM}$,
respectively.
We denote with $\ETA^{\wedge}_{\AllKid,i,j}$ the random variable $\ETA_{\AllKid,i,j} \wedge \ETA'_{\AllKid,i,j}$,  and with $\ETA^{\vee}_{\AllKid,i,j}$ 
the random variable $\ETA_{\AllKid,i,j} \vee \ETA'_{\AllKid,i,j}$. Similarly, we write  
$\ETA^{\wedge}_{\PureP,i,j} = \ETA_{\PureP,i,j}\wedge \ETA'_{\PureP,i,j}$
and $\ETA^{\vee}_{\PureP,i,j} = \ETA_{\PureP,i,j} \vee \ETA'_{\PureP,i,j}$, and also write 
$\ETA^{\wedge}_{\PureM,i,j} = \ETA_{\PureM,i,j} \wedge \ETA'_{\PureM,i,j}$ and 
$\ETA^{\vee}_{\PureM,i,j} = \ETA_{\PureM,i,j}, \ETA'_{\PureM,i,j}$.

Moreover, for a sign $\varepsilon \in \{\pm 1\}$ and a vector $r=(r_{\AllKid}, r_{\PureP}, r_{\PureM}, r_{\NoKid})$ of non-negative integers with 
$r_{\AllKid} + r_{\PureP} + r_{\PureM} + r_{\NoKid} = {k-1}$ and $1 \le i \le r_{\AllKid}$, $1 \le j \le r_{\PureP}$, $1 \le \ell \le r_{\PureM}$, we let
\begin{align*}
\diffr^{\AllKid}_{i}&(z, r; \varepsilon)\\
&= 
\left|
\frac{\partial}
{\partial z}
\log{
{\left(
1 -  \frac{1}{2^{\vec{r}_{\NoKid}}}
\Pfun\left(
\varepsilon  (\ETA_{\AllKid,1,1}, \ldots,\ETA_{\AllKid,1,i-1}, z, \ETA'_{\AllKid,1,i+1}, \ldots
\ETA'_{\AllKid,1,r_{\AllKid}})
\right)
\Pfun\left(
\varepsilon(\ETA'_{\PureP,1,1}, \ldots, \ETA'_{\PureP,1,r_{\PureP}})
\right)
\Pfun\left(
\varepsilon(\ETA'_{\PureM,1,1}, \ldots, \ETA'_{\PureM,1,r_{\PureM}})
\right)
\right)
}
}\right|
.
\end{align*}
Analogously, we define
\begin{align*}
\diffr^{\PureP}_{j}&(z, r; \varepsilon)\\
&=
\left|
\frac{\partial}
{\partial z}
\log{
{\left(
1 -  \frac{1}{2^{\vec{r}_{\NoKid}}} 
\Pfun\left(
\varepsilon  (\ETA_{\AllKid,1,1}, \ldots
\ETA_{\AllKid,1,r_{\AllKid}})
\right)
\Pfun\left(
\varepsilon  (\ETA_{\PureP,1,1}, \ldots, \ETA_{\PureP,1,j-1}, z, \ETA'_{\PureP,1,j+1},
\ldots, \ETA'_{\PureP,1,r_{\PureP}})
\right)
\Pfun\left(
\varepsilon  (\ETA'_{\PureM,1,1}, \ldots
\ETA'_{\PureM,1,r_{\PureM}})
\right)
\right)
}
}\right|,
\end{align*}
\begin{align*}
\diffr^{\PureM}_{\ell}&(z, r; \varepsilon) \\
&=
\left|
\frac{\partial}
{\partial z}
\log{
{\left(
1 -  \frac{1}{2^{\vec{r}_{\NoKid}}} 
\Pfun\left(
\varepsilon  (\ETA_{\AllKid,1,1}, \ldots
\ETA_{\AllKid,1,r_{\AllKid}})
\right)
\Pfun\left(
\varepsilon(\ETA_{\PureP,1,1}, \ldots, \ETA_{\PureP,1,r_{\PureP}})
\right)
\Pfun\left(
\varepsilon  (\ETA_{\PureM,1,1}, \ldots, \ETA_{\PureM,1,\ell-1}, z, \ETA'_{\PureM,1,\ell+1},
\ldots, \ETA'_{\PureP,1,r_{\PureM}})
\right)
\right)
}
}\right|.
\end{align*}

With the above notation in place, we are now ready to bound $W_1(\hat\rho_{\PureP}, \hat{\rho}'_{\PureP})$. 
For each of the pairs of distributions $(\rho_{\AllKid},\rho'_{\AllKid})$, 
$(\rho_{\PureP},\rho'_{\PureP})$, and 
$(\rho_{\PureM},\rho'_{\PureM})$, fix an arbitrary coupling among its coordinates. 

\begin{lemma}\label{lem_BoundW1P}
$W_1(\hat\rho_{\PureP}, \hat{\rho}'_{\PureP})$ is upper bounded by
\begin{align}\label{eq:FinalBoundForW1P}
\frac{d/2}{1-e^{-\frac{d}{2}}}
\cdot
\Erw\Biggl[
\sum_{i=1}^{\vec{r}_{\AllKid,1}}
\int_{{\ETA}^\wedge_{\AllKid, 1, i}}^{{\ETA}^\vee_{\AllKid, 1, i}}\!
\diffr^{\AllKid}_{i}(w_i, \vr_1; +1)
\dd w_i
+
\sum_{j=1}^{\vec{r}_{\PureP,1}}
\!
\int_{{\ETA}^\wedge_{\PureP, 1, j}}^{{\ETA}^\vee_{\PureP, 1, j}}\!
\diffr^{\PureP}_{j}(y_j, \vr_1; +1)
\dd y_j
+
\sum_{\ell=1}^{\vec{r}_{\PureM,1}}
\int_{{\ETA}^\wedge_{\PureM, 1, \ell}}^{{\ETA}^\vee_{\PureM, 1, \ell}}
\!
\diffr^{\PureM}_{\ell}(z_\ell, \vr_1; +1)
\dd z_\ell
\Biggr]
\enspace.
\end{align}
\end{lemma}

\begin{proof}
Let us write $\LArg_{i,j}'(\eps,r)$ for the expression 
in the r.h.s. of \eqref{eq:defOfXI} where distribution $\rho'$ is used instead of $\rho$, i..,
\begin{align}\label{eq:defOfprimeXI}
\LArg_{i,j}'(\eps,r)=  1 -  \frac{1}{2^{{r}_{\NoKid}}} 
\Pfun
\left(\eps \bigl( {\ETA'_{\AllKid,4i+j,1}},\ldots,{\ETA'_{\AllKid,4i+j,r_{\AllKid}}}\bigr) \right)
\Pfun
\left(\eps\bigl( {\ETA'_{\PureP,4i+j,1}},\ldots,{\ETA'_{\PureP,4i+j,r_{\PureP}}}\bigr) \right)
\Pfun\left(\eps
\bigl( {\ETA'_{\PureM,4i+j,1}},\ldots,{\ETA'_{\PureM,4i+j,r_{\PureM}}}\bigr)
\right)
.
\end{align}

By identically coupling the number of clauses and the types of the children variables of each clause in $\hat\rho_{\PureP}, \hat\rho'_{\PureP}$, 
we see that by the definition of the $W_1$ norm, 
\begin{align*}
W_1(\hat\rho_{\PureP}, \hat{\rho}'_{\PureP})
&\leq
\Erw{
        \brk
        {\left| -\sum_{i=1}^{\vdpc}
\log{\LArg_{i,3}(+1, \vr_{4i+3})
\over
\LArg'_{i,3}(+1, \vr_{4i+3})}
\right|}
}
\enspace.
\end{align*}
Applying Wald's lemma, we further obtain 
\begin{align}\label{eq:WaldsOnlyP}
W_1(\hat\rho_{\PureP}, \hat{\rho}'_{\PureP})
\le	
\frac{d/2}{1- e^{-d/2}}
\cdot
\Erw{
\brk
{\left| 
\log{\LArg_{1,3}(+1, \vr_{7})
\over
\LArg'_{1,3}(+1, \vr_{7})}
\right|}
}
=
\frac{d/2}{1- e^{-d/2}}
\cdot
\Erw{
\brk
{\left| 
\log{\LArg_{0,1}(+1, \vr_{1})
\over
\LArg'_{0,1}(+1, \vr_{1})}
\right|}
}
\enspace.
\end{align}
Let us now focus on the expectation in the r.h.s. of
\eqref{eq:WaldsOnlyP}. Recalling the definition of $\; \LArg$ in \eqref{eq:defOfXI}, and the definition of $\; \LArg'$ in \eqref{eq:defOfprimeXI},
we expand  
\begin{align}\label{eq:RecallXI}
\log{\LArg_{0,1}(+1, \vr_{1})
\over
\LArg'_{0,1}(+1, \vr_{1})}
=
\log{
\frac{
1 -  {2^{-{r}_{\NoKid}}} \cdot
\Pfun
\left( {\ETA_{\AllKid,1,1}},\ldots,{\ETA_{\AllKid,1,\vr_{\AllKid, 1}}} \right)
\cdot
\Pfun
\left( {\ETA_{\PureP,1,1}},\ldots,{\ETA_{\PureP,1,\vr_{\PureP,1}}} \right)
\cdot
\Pfun\left(
{\ETA_{\PureM,1,1}},\ldots,{\ETA_{\PureM,1,\vr_{\PureM,1}}}
\right)
}
{
1 -  {2^{-{r}_{\NoKid}}} \cdot
\Pfun
\left( {\ETA'_{\AllKid,1,1}},\ldots,{\ETA'_{\AllKid,1,\vr_{\AllKid,1}}} \right)
\cdot
\Pfun
\left( {\ETA'_{\PureP,1,1}},\ldots,{\ETA'_{\PureP,1,\vr_{\PureP,1}}} \right)
\cdot
\Pfun\left( {\ETA'_{\PureM,1,1}},\ldots,{\ETA'_{\PureM,1,\vr_{\PureM,1}}}  \right)
}
}.
\end{align}
Telescoping over the arguments of the functions $\Pfun$ in the r.h.s of
\eqref{eq:RecallXI}, invoking the fundamental theorem of calculus for
each term, and applying the triangle inequality we further obtain
\begin{align*}
\left|\log{\LArg_{0,1}(+1, \vr_{1})
\over
\LArg'_{0,1}(+1, \vr_{1})}
\right|
\le
\sum_{i=1}^{\vec{r}_{\AllKid,1}}
\left|
\int_{{\ETA}'_{\AllKid,1,i}}^{\ETA_{\AllKid,1,i}}
\diffr^{\AllKid}_{i}(w_i, \vec{r}_1; +1)
\dd w_i
\right|
+
\sum_{j=1}^{\vec{r}_{\PureP,1}}
\left|
\int_{{\ETA}'_{\PureP,1,j}}^{\ETA_{\PureP,1,j}}
\diffr^{\PureP}_{j}(y_j, \vec{r}_1; +{1})
\dd y_j
\right|
+
\sum_{\ell=1}^{\vec{r}_{\PureM,1}}
\left|
\int_{{\ETA}'_{\PureM,1,\ell}}^{\ETA_{\PureM,1,\ell}}
\diffr^{\PureM}_{\ell}(z_\ell, \vec{r}_1; +{1})
\dd z_\ell
\right|.
\end{align*}
Plugging the above into \eqref{eq:WaldsOnlyP} gives the result.
\end{proof}

Following the same steps as above, but replacing `$+1$' with `$-1$', yields the corresponding bound for $W_1(\hat\rho_{\PureM},\hat{\rho}'_{\PureM})$.

\begin{lemma}\label{lem_BoundW1M}
$W_1(\hat\rho_{\PureM},\hat{\rho}'_{\PureM})$ is upper bounded by
\begin{align}\label{eq:FinalBoundForW1M}
\frac{d/2}{1-e^{-\frac{d}{2}}}
\cdot
\Erw\Biggl[
\sum_{i=1}^{\vec{r}_{\AllKid,1}}
\int_{{\ETA}^\wedge_{\AllKid, 1, i}}^{{\ETA}^\vee_{\AllKid, 1, i}}\!
\diffr^{\AllKid}_{i}(w_i, \vr_1; -1)
\dd w_i
+
\sum_{j=1}^{\vec{r}_{\PureP,1}}
\!
\int_{{\ETA}^\wedge_{\PureP, 1, j}}^{{\ETA}^\vee_{\PureP, 1, j}}\!
\diffr^{\PureP}_{j}(y_j, \vr_1; -1)
\dd y_j
+
\sum_{\ell=1}^{\vec{r}_{\PureM,1}}
\int_{{\ETA}^\wedge_{\PureM, 1, \ell}}^{{\ETA}^\vee_{\PureM, 1, \ell}}
\!
\diffr^{\PureM}_{\ell}(z_\ell, \vr_1; -1)
\dd z_\ell
\Biggr]
\enspace.
\end{align}
\end{lemma}

Exploiting the signs of the variables with types $\rPureP$ and $\rPureM$, we obtain the following bounds for each of the $\diffr$-functions.
For $\lambda \in (0,1]$, we define the real function $\fun_\lambda : [0,1] \to \RR$ as
\begin{align}
\fun_\lambda \left(w\right)	=  \frac{\lambda\cdot w}{1-\lambda\cdot w} \cdot(1-w)
\enspace.
\label{eq:DefofPsi}
\end{align}
It is easy to check that $\fun_{\lambda'}(w) \le \fun_{\lambda}(w)$, for every $\lambda' \le \lambda$.
\begin{claim}\label{cl_Dbound}
For every
$r = \left({r}_{\AllKid},
{r}_{\PureP}, {r}_{\PureM}, {r}_{\NoKid}\right)$, and
$i\in [r_{\AllKid}]$ we have that
\begin{align}\label{eq:BoundOfDiffAllKid}
\diffr^{{\AllKid}}_{i}
\left(w_i, r ;+ 1
\right)
\le
\fun_{2^{-r_{\NoKid}-r_{\PureM}}}
\left(\frac{1 + \tanh({w_i}/{2})}{2}\right)
\enspace
&& \text{ and } &&
\diffr^{{\AllKid}}_{i}
\left(w_i, r ;- {1}
\right)
&\le
\fun_{2^{-r_{\NoKid}-r_{\PureP}}}
\left(\frac{1 - \tanh({w_i}/{2})}{2}\right)
\enspace.
\end{align}
Similarly, we also have that for $j \in [r_{\PureP}]$,
\begin{align}\label{eq:BoundOfDiffPureP}
\diffr^{{\PureP}}_{j}
\left(y_j, r ;+ {1}
\right)
\le
\fun_{2^{-r_{\NoKid}-r_{\PureM}}}
\left(\frac{1 + \tanh({y_j}/{2})}{2}\right)
\enspace
&& \text{ and } &&
\diffr^{{\PureP}}_{j}
\left(y_j, r ;- {1}
\right)
&\le
\fun_{2^{-r_{\NoKid}-(r_{\PureP}-1)}}
\left(\frac{1 - \tanh({y_j}/{2})}{2}\right)
\enspace,
\end{align}
and for $\ell \in [r_{\PureM}]$
\begin{align}\label{eq:BoundOfDiffPureM}
\diffr^{{\PureM}}_{\ell}
\left(z_\ell, r ;+ {1}
\right)
\le
\fun_{2^{-r_{\NoKid}-(r_{\PureM}-1)}}
\left(\frac{1 + \tanh({z_\ell}/{2})}{2}\right)
\enspace
&& \text{ and } &&
\diffr^{{\PureM}}_{\ell}
\left(z_\ell, r ;- {1}
\right)
&\le
\fun_{2^{-r_{\NoKid}-r_{\PureP}}}
\left(\frac{1 - \tanh({z_\ell}/{2})}{2}\right)
\enspace.
\end{align}
\end{claim}

\begin{proof}
We only prove the first inequality of \eqref{eq:BoundOfDiffAllKid} as the rest of them follow in a similar manner.  
A straightforward calculation shows that for $\rvec{z} \in \RR^q, \varepsilon \in \PM$, and $i \in [q]$ we have
\begin{align}\label{eq_derOfG}
\frac{\partial}{\partial z_i}
\Pfun \left(\varepsilon\cdot\rvec{z}\right)
&=
{\varepsilon}
\cdot
\frac{1-\tanh(\varepsilon\cdot{z_i}/2)}{2}
\cdot
\Pfun \left(\varepsilon \cdot \rvec{z}\right)
\enspace.
\end{align}
Writing 
$K =  {2^{-{r}_{\NoKid}}}
\Pfun\left(
\ETA_{\AllKid,1,1}, \ldots,\ETA_{\AllKid,1,i-1}, \ETA'_{\AllKid,1,i+1}, \ldots
\ETA'_{\AllKid,1,r_{\AllKid}}
\right)
\Pfun\left(
\ETA'_{\PureP,1,1}, \ldots, \ETA'_{\PureP,1,r_{\PureP}}
\right)
\Pfun\left(
\ETA'_{\PureM,1,1}, \ldots, \ETA'_{\PureM,1,r_{\PureM}}
\right)$, 
applying the chain rule, and using \eqref{eq_derOfG}, we see that
\begin{align}
\diffr^{{\AllKid}}_{i}\left(w_i, r ;+ 1\right) 
=
\fun_K \bc{\frac{1 + \tanh({w_i}/{2})}{2}}\enspace.
\end{align}
Using the fact that $\rho'_{\PureM}$ is supported in $(-\infty, 0]$, and that $\Pfun \le 1$, we obtain $K \le 2^{-r_{\NoKid}-r_{\PureM}}$. The
monotonicity of $\fun_\lambda$ with respect to the parameter $\lambda$ concludes the proof. 
\end{proof}

Using \Cl~\ref{cl_Dbound}, and maximising each of the functions $\fun_\lambda$ appearing in \eqref{eq:BoundOfDiffAllKid}--\eqref{eq:BoundOfDiffPureM}, we can recover
the bounds of~\cite{MM}. To obtain sharper bounds, a natural idea is to optimise groups of summands, instead of optimising each $\diffr$-summand of
$W_1(\hat\rho_{\PureP},\hat{\rho}'_{\PureP}) + W_1(\hat\rho_{\PureM},\hat{\rho}'_{\PureM})$ in isolation. In particular, 
it is tempting to pair terms of the form $\diffr(\cdot,-1)$ with corresponding terms of the form $\diffr(\cdot,+1)$, as \Lem~\ref{lem:OptimizeAntipodal} suggests.
\begin{lemma}
\label{lem:OptimizeAntipodal}
Let $ \sfun_\lambda : [0,1] \to \RR$ to be the function
$\sfun_\lambda(w) =
\fun_\lambda(w) + \fun_\lambda(1-w)
$.
For every $\lambda \in (0,1]$, we have that  $\sfun_\lambda(w) \le \sfun_\lambda(1/2) = \frac{\lambda/2}{1-\lambda/2}$, for all $w\in [0,1]$.
\end{lemma}

\begin{proof}
For $\lambda = 1$, we have that $\fun_\lambda(w) = w$ implying $\sfun_\lambda(w) = 1$, and thus, the result holds trivially. Let now $\lambda \in (0,1)$. 
Differentiating $\fun_\lambda$ gives
\begin{align*}
\fun^\prime_\lambda(w)
=
\frac{\lambda^{2} w^{2} - 2 \lambda w + \lambda}{\lambda^{2} w^{2} - 2 \, \lambda w + 1}
=
1 - \frac{ 1 - \lambda }{(1-\lambda w)^2}
\enspace.
\end{align*}
Therefore,
\begin{align*}
\sfun^\prime_\lambda(w)
=
1 - \frac{ 1 - \lambda }{(1-\lambda w)^2}
-\left(1 - \frac{ 1 - \lambda }{(1-\lambda (1-w))^2}\right)
=
\frac{ 1 - \lambda }{(1-\lambda (1-w))^2}
- \frac{ 1 - \lambda }{(1-\lambda w)^2}
\enspace.
\end{align*}
It is straightforward to check that the above expression has only one root at $w=1/2$, being non-negative for $w \in [0, 1/2)$, and non-positive for $w \in (1/2, 1]$. Therefore, $\sfun_\lambda(1/2) = \frac{\lambda/2}{1-\lambda/2}$ is the maximum value of $\sfun_\lambda$.
\end{proof}

However, directly applying \Lem~\ref{lem:OptimizeAntipodal} to 
$W_1(\hat\rho_{\PureP},\hat{\rho}'_{\PureP}) + W_1(\hat\rho_{\PureM},\hat{\rho}'_{\PureM})$
seems hopeless, since in the bounds supplied by \Cl~\ref{cl_Dbound}, the parameters of the functions $\fun$ bounding 
$\diffr(\cdot,r,+1)$-terms in $W_1\left(\hat\rho_{\PureP},\hat{\rho}'_{\PureP}\right)$
are quite different from the parameters of the functions $\fun$ bounding the corresponding $\diffr(\cdot,r,-1)$-terms in
$W_1\left(\hat\rho_{\PureM},\hat{\rho}'_{\PureM}\right)$.

The following lemma reveals, a somewhat unexpected, symmetry between $W_1(\hat\rho_{\PureP},\hat{\rho}'_{\PureP})$
and $W_1(\hat\rho_{\PureM},\hat{\rho}'_{\PureM})$, that facilitates our pairing strategy. 

Some additional notation is in order. We denote with $\cR(k)$ for the set of all vectors $r = (r_{\AllKid}, r_{\PureP}, r_{\PureM}, r_{\NoKid})$ of 
non-negative integer entries which sum to $k-1$. For every $r \in \cR(k)$ we use the shorthand %
\begin{align*}
P({r})
= \frac{(k-1)!}
{{r}_{\AllKid}! {r}_{\PureP}!{r}_{\PureM}! {r}_{\NoKid}!}
\cdot
p_{\AllKid}^{{r}_{\AllKid}}
p_{\PureP}^{{r}_{\PureP}}
p_{\PureM}^{{r}_{\PureM}}
p_{\NoKid}^{{r}_{\NoKid}}
\enspace,
\end{align*}
where $p_{\AllKid}, p_{\PureP}, p_{\PureM}, p_{\NoKid}$ are the probabilities defined in \eqref{eq_def_multipr}. Finally, we define
\begin{align}
\Ex{\AllKid}
&=
\sum_{\substack{r \in \cR(k) \\ r_{\AllKid}\ge 1}}
P({r})
\cdot
{r}_{\AllKid}
\cdot
\Erw
\left[
\int_{{\ETA}^\wedge_{\AllKid, 1, 1}}^{{\ETA}^\vee_{\AllKid, 1, 1}}
\sfun_{2^{-r_{\NoKid}-r_{\PureM}}} \left(\frac{1+\tanh(w/2)}{2}\right)
\;
\dd w
\right]
\enspace,
\label{eq:ExpAllKid}
\\
\Ex{\PureP}
&=
\sum_{\substack{r \in \cR(k) \\ r_{\PureP}\ge 1}}
P({r})
\cdot
{r}_{\PureP}
\cdot
\Erw
\left[
\int_{{\ETA}^\wedge_{\PureP, 1, 1}}^{{\ETA}^\vee_{\PureP, 1, 1}}
\sfun_{2^{-r_{\NoKid}-r_{\PureM}}}
\left(\frac{1 + \tanh({y}/{2})}{2}\right)
\;
\dd y
\right]
\enspace,
\label{eq:ExpPureP}
\\
\Ex{\PureM}
&=
\sum_{\substack{r \in \cR(k) \\ r_{\PureM}\ge 1}}	P({r})
\cdot
{r}_{\PureM}
\cdot
\Erw
\left[ 
\int_{{\ETA}^\wedge_{\PureM, 1, 1}}^{{\ETA}^\vee_{\PureM, 1, 1}}
\sfun_{2^{-r_{\NoKid}-r_{\PureP}}}
\left(\frac{1 + \tanh({z}/{2})}{2}\right)
\;
\dd z
\right]
\label{eq:ExpPureM}
\enspace.
\end{align}

\begin{lemma}\label{lem_PairingE}
We have that
\begin{align}
W_1(\hat\rho_{\PureP},\hat{\rho}'_{\PureP})
+
W_1(\hat{\rho}_{\PureM},\hat{\rho}'_{\PureM})
\le
\frac{d/2}{1-e^{-d/2}}
\left(
\Ex{\AllKid}
+
\Ex{\PureP}
+
\Ex{\PureM}
\right)
\enspace.
\label{eq:ExpBreak3}
\end{align}
\end{lemma}

\begin{proof}
Expanding the expectation in \eqref{eq:FinalBoundForW1P} with respect
to $\vec{r}=\left(\vec{r}_{\AllKid}, \vec{r}_{\PureP},
\vec{r}_{\PureM}, \vec{r}_{\NoKid}\right)$ , and using the shorthand
\begin{align*}
E^{\pm}_{\AllKid}(r) &=  \ex\brk{
\int_{{\ETA}^\wedge_{\AllKid, 1, 1}}^{{\ETA}^\vee_{\AllKid, 1, 1}}
\diffr^{\AllKid}_{1}
\left(w, r; \pm{1}
\right) \dd w
} ,\enspace 
E^{\pm}_{\PureP}(r) =  \ex\brk{
\int_{{\ETA}^\wedge_{\PureP, 1, 1}}^{{\ETA}^\vee_{\PureP, 1, 1}}
\diffr^{\PureP}_{1}
\left(y, r; \pm{1}
\right) \dd y} 
, \enspace 
E^{\pm}_{\PureM}(r) =  \ex\brk{
\int_{{\ETA}^\wedge_{\PureM, 1, 1}}^{{\ETA}^\vee_{\PureM, 1, 1}}
\diffr^{\PureM}_{1}
\left(z, r; \pm\vec{1}
\right) \dd z},
\end{align*}
we see that
\begin{align}\label{eq_W1Pb}
W_1(\hat\rho_{\PureP},\hat\rho'_{\PureP})
&\le \frac{d/2}{1-e^{-d/2}}
\bc{
\sum_{r \in \cR(k)} P(r) \cdot r_{\AllKid} \cdot E^{+}_{\AllKid}(r) +
\sum_{r \in \cR(k)} P(r) \cdot r_{\PureP} \cdot E^{+}_{\PureP}(r) +
\sum_{r \in \cR(k)} P(r) \cdot r_{\PureM} \cdot E^{+}_{\PureM}(r)
}\nonumber\\
&= \frac{d/2}{1-e^{-d/2}}
\bc{
\sum_{\substack{r \in \cR(k) \\ r_{\AllKid}\ge1}} P(r) \cdot r_{\AllKid} \cdot E^{+}_{\AllKid}(r) +
\sum_{\substack{r \in \cR(k) \\ r_{\PureP}\ge1}} P(r) \cdot r_{\PureP} \cdot E^{+}_{\PureP}(r) +
\sum_{\substack{r \in \cR(k) \\ r_{\PureM}\ge1}} P(r) \cdot r_{\PureM} \cdot E^{+}_{\PureM}(r)
}
\enspace.
\end{align}
In a similar manner, we derive
\begin{align}\label{eq_W1Mb}
W_1(\hat\rho_{\PureM},\hat\rho'_{\PureM})
\le \frac{d/2}{1-e^{-d/2}}
\bc{
\sum_{\substack{r \in \cR(k) \\ r_{\AllKid}\ge1}} P(r) \cdot r_{\AllKid} \cdot E^{-}_{\AllKid}(r) +
\sum_{\substack{r \in \cR(k) \\ r_{\PureP}\ge1}} P(r) \cdot r_{\PureP} \cdot E^{-}_{\PureP}(r) +
\sum_{\substack{r \in \cR(k) \\ r_{\PureM}\ge1}} P(r) \cdot r_{\PureM} \cdot E^{-}_{\PureM}(r)
}\enspace.
\end{align}

Let us now consider the bound on the sum $W_1(\hat\rho_{\PureP},\hat\rho'_{\PureP}) + W_1(\hat\rho_{\PureM},\hat\rho'_{\PureM})$ obtained by summing 
\eqref{eq_W1Pb}, \eqref{eq_W1Mb}. We next group each of the three sums in \eqref{eq_W1Pb} with the corresponding sum in \eqref{eq_W1Mb}, carefully pairing their terms.
Specifically, for the $\rAllKid$--sums we match  the term of $\sum_{r} P(r) \cdot r_{\AllKid} \cdot E^{+}_{\AllKid}(r)$ corresponding to
$r = ({r}_{\AllKid}, {r}_{\PureP},
{r}_{\PureM}, {r}_{\NoKid})$ with the term of $\sum_{r'} P(r') \cdot r'_{\AllKid} \cdot E^{-}_{\AllKid}(r')$ that corresponds to
$r' = ({r}_{\AllKid}, {r}_{\PureM}, {r}_{\PureP}, {r}_{\NoKid})$. Since $r \mapsto r'$ is a bijection of
$\mathcal{R}(k) \cap \{r: r_{\AllKid}\ge1\}$, and $r'_{\AllKid}= r_{\AllKid}$, and $P(r') = P(r)$ we see that
\begin{align}\label{eq_MatchAllKid}
\sum_{\substack{r \in \cR(k)\\ r_{\AllKid \ge 1}}} P(r)\cdot r_{\AllKid} \bc{E^{+}_{\AllKid}(r) + E^{-}_{\AllKid}(r)}
= \sum_{\substack{r \in \cR(k)\\ r_{\AllKid \ge 1}}} P(r)\cdot r_{\AllKid} \bc{E^{+}_{\AllKid}(r) + E^{-}_{\AllKid}(r')} \enspace.
\end{align}
Invoking the bounds \eqref{eq:BoundOfDiffAllKid} of \Cl~\ref{cl_Dbound}, and recalling the definitions of $\sfun$, $\Ex{\AllKid}$, we upper bound the r.h.s. of
\eqref{eq_MatchAllKid}~by
\begin{align}\label{eq_SumEA}
\sum_{\substack{r \in \cR(k)\\ r_{\AllKid \ge 1}}} P(r)r_{\AllKid} \bc{\Erw
\brk{
\int_{{\ETA}^\wedge_{\AllKid, 1, 1}}^{{\ETA}^\vee_{\AllKid, 1, 1}}
\fun_{2^{-r_{\NoKid}-r_{\PureM}}} \left(\frac{1+\tanh(w/2)}{2}\right)
\; \dd w}
+
\Erw\brk{
\int_{{\ETA}^\wedge_{\AllKid, 1, 1}}^{{\ETA}^\vee_{\AllKid, 1, 1}}
\fun_{2^{-r_{\NoKid}-r_{\PureM}}} \left(\frac{1-\tanh(w/2)}{2}\right)
\;\dd w
}}	
= \Ex{\AllKid}
\enspace.
\end{align}
The matchings between the terms for the $\rPureP, \rPureM$--sums of \eqref{eq_W1Pb}, \eqref{eq_W1Mb} are more delicate.
In particular, for the $\rPureP$--sum it turns out that we can pull off the same trick as above by pairing the term of
$\sum_{r} P(r) \cdot r_{\PureP} \cdot E^{+}_{\PureP}(r)$ corresponding to the vector $r = ({r}_{\AllKid}, {r}_{\PureP},
{r}_{\PureM}, {r}_{\NoKid})$ with the term of $\sum_{r''} P(r'') \cdot r''_{\PureP} \cdot E^{-}_{\PureP}(r'')$ that corresponds to 
$r'' = ({r}_{\AllKid}, {r}_{\PureM}+1, {r}_{\PureP}-1, {r}_{\NoKid})$. To see this, note that the mapping $r \mapsto r''$ is a bijection of
$\mathcal{R}(k) \cap \{r: r_{\PureP}\ge1\}$, leaving the quantity $P(r)\cdot r_{\PureP}$ invariant as  
\begin{align*}
P(r) \cdot r_{\PureP} =
\frac{(k-1)!}
{{r}_{\AllKid}! ({r}_{\PureP}-1)!{r}_{\PureM}! {r}_{\NoKid}!}
\cdot
p_{\AllKid}^{{r}_{\AllKid}}
p_{\PureP}^{{r}_{\PureP}}
p_{\PureM}^{{r}_{\PureM}}
p_{\NoKid}^{{r}_{\NoKid}} =
\frac{(k-1)!}
{{r}_{\AllKid}! ({r}_{\PureM}+1)!({r}_{\PureP}-1)! {r}_{\NoKid}!}
\cdot
p_{\AllKid}^{{r}_{\AllKid}}
p_{\PureP}^{{r}_{\PureP}}
p_{\PureM}^{{r}_{\PureM}}
p_{\NoKid}^{{r}_{\NoKid}}
\cdot
(r_{\PureM}+1)=
P(r'') \cdot r''_{\PureP}
\enspace.
\end{align*}
Invoking the bounds \eqref{eq:BoundOfDiffPureP} of \Cl~\ref{cl_Dbound}, recalling the definitions of $\sfun$, $\Ex{\PureP}$, and arguing as above, 
we obtain
\begin{align}\label{eq_SumEP}
\sum_{\substack{r \in \cR(k)\\ r_{\PureP}\ge1}} P(r)\cdot r_{\PureP} \bc{E^{+}_{\PureP}(r) + E^{-}_{\PureP}(r'')} \le \Ex{\PureP}\enspace.
\end{align}
Similarly, using the mapping $r\mapsto r'''$, with $r'''\! = (r_{\AllKid}, r_{\PureM} - 1, r_{\PureP}+1, r_{\NoKid})$, and following the same steps as above, we get
\begin{align}\label{eq_SumEM}
\sum_{\substack{r \in \cR(k)\\ r_{\PureM}\ge1}} P(r)\cdot r_{\PureM} \bc{E^{+}_{\PureM}(r) + E^{-}_{\PureM}(r''')} \le \Ex{\PureM}\enspace.
\end{align}
Summing \eqref{eq_SumEA}--\eqref{eq_SumEM} concludes the proof.
\end{proof}

In light of the above, we are now ready to finish the proof of \Prop~\ref{prop_BPPureLit}.

\begin{proof}[Proof of \Prop~\ref{prop_BPPureLit}]
Applying \Lem~\ref{lem:OptimizeAntipodal} on the function $\sfun$ in the r.h.s. of \eqref{eq:ExpAllKid} gives
\begin{align}
\sfun_{2^{-r_{\NoKid}-r_{\PureM}}} \left(\frac{1+\tanh(w/2)}{2}\right)
\le
\frac{2^{-\vec{r}_{\PureM}- \vec{r}_{\NoKid}-1}}{1 - 2^{-\vec{r}_{\PureM}- \vec{r}_{\NoKid}-1}}
\le
\left(\frac{1}{2}\right)^{\vec{r}_{\PureM} + \vec{r}_{\NoKid}} \enspace.
\label{eq:SfunAllKidBound}
\end{align}
Plugging the above into  \eqref{eq:ExpAllKid} and applying the binomial theorem, further yields 
\begin{align}
\Ex{\AllKid} \le (k-1) \cdot p_{\AllKid} \bc{1 - \frac{e^{-\frac{d}{2}}}{2}}^{k-2} \ex\brk{|\ETA_{\AllKid, 1,1 }-\ETA'_{\AllKid, 1,1 }|}
\enspace.
\label{eq:ExpAllKidFin}
\end{align}
Working in a similar manner, we obtain 
\begin{align}
\Ex{\PureP}
\le
(k-1)\cdot p_{\PureP}
\left(1-\frac{e^{-\frac{d}{2}}}{2}\right)^{k-2}
\!\!
\ex\brk{|\ETA_{\PureP, 1,1 }-\ETA'_{\PureP, 1,1 }|},
&&
\text{and}
&&
\Ex{\PureM}
\le
(k-1)\cdot p_{\PureM}
\left(1-\frac{e^{-\frac{d}{2}}}{2}\right)^{k-2}
\!\!
\ex\brk{|\ETA_{\PureM, 1,1 }-\ETA'_{\PureM, 1,1 }|}
\label{eq:ExpPureMFin}
\enspace.
\end{align}
Finally, plugging the bounds \eqref{eq:ExpAllKidFin},
and \eqref{eq:ExpPureMFin} into \eqref{eq:ExpBreak3} we see that 
$ W_1(\hat\rho_{\PureP},\hat\rho'_{\PureP})
+
W_1(\hat\rho_{\PureM},\hat\rho'_{\PureM})$
is upper bounded by 
\begin{align}
\frac{d\cdot (k-1)}{2}
\cdot
\left(1-\frac{e^{-\frac{d}{2}}}{2}\right)^{k-2}
\!\! 
\left[
\left(1-e^{-\frac{d}{2}}\right)
\ex\brk{|\ETA_{\AllKid, 1,1 }-\ETA'_{\AllKid, 1,1 }|}
+
e^{-\frac{d}{2}}
\ex\brk{|\ETA_{\PureP, 1,1 }-\ETA'_{\PureP, 1,1 }|}
+
e^{-\frac{d}{2}}
\ex\brk{|\ETA_{\PureM, 1,1 }-\ETA'_{\PureM, 1,1 }|}
\right]
\enspace.
\label{eq:ExpBound3}
\end{align}
Recall that we established \eqref{eq:ExpBound3} assuming an arbitrary coupling between the coordinates of each pair of distributions $(\rho_{\AllKid},\rho'_{\AllKid})$, 
$(\rho_{\PureP},\rho'_{\PureP})$, and $(\rho_{\PureM},\rho'_{\PureM})$. Therefore, the definition of $W_1$ norm and \eqref{eq:ExpBound3}, imply the first inequality
below, while \eqref{eq:ExpBound4} follows by the definition \eqref{eq:DistDef} of $\dist_d$  
\begin{align}
W_1(\hat\rho_{\PureP},\hat\rho'_{\PureP})+W_1(\hat\rho_{\PureM},\hat\rho'_{\PureM}) 	
&\le
\frac{d (k-1)}{2}
\left(1-\frac{e^{-\frac{d}{2}}}{2}\right)^{k-2}
\left[
\left(1-e^{-\frac{d}{2}}\right)
W_1(\rho_{\AllKid}, \rho'_{\AllKid})
+
e^{-\frac{d}{2}}
W_1(\rho_{\PureP}, \rho'_{\PureP})
+
e^{-\frac{d}{2}}
W_1(\rho_{\PureM}, \rho'_{\PureM})
\right]
\nonumber\\
&\le
\frac{d (k-1)}{2}
\left(1-\frac{e^{-\frac{d}{2}}}{2}\right)^{k-2}
\dist_d(\rho, \rho') \enspace.
\label{eq:ExpBound4}
\end{align}
Moreover, as per the triangle inequality we see that
\begin{align}
W_1(\hat\rho_{\AllKid},\hat\rho'_{\AllKid}) \le W_1(\hat\rho_{\PureP},\hat\rho'_{\PureP})+ W_1(\hat\rho_{\PureM},\hat\rho'_{\PureM})
\le\frac{d (k-1)}{2}
\left(1-\frac{e^{-\frac{d}{2}}}{2}\right)^{k-2}
\dist_d(\rho, \rho')
\enspace.  
\label{eq_AllBTri}
\end{align}
Plugging the bounds \eqref{eq:ExpBound4} and \eqref{eq_AllBTri} into the expression of $\dist_{d} \bc{\hat\rho, \hat\rho'}$ yields
\begin{align*}
\dist_{d} \bc{\hat\rho, \hat\rho'} 
&= 
(1-e^{-d/2})
\cdot
W_1(\hat\rho_{\AllKid}, \hat{\rho}_{\AllKid}')
+
e^{-d/2}
\bc{
W_1(\hat\rho_{\PureP}, \hat{\rho}_{\PureP}')
+
W_1(\hat\rho_{\PureM}, \hat{\rho}_{\PureM}')}
\le
\frac{d (k-1)}{2}
\left(1-\frac{e^{-\frac{d}{2}}}{2}\right)^{k-2}
\dist_d(\rho, \rho') \enspace.
\end{align*}
Recalling the definition of ${\dours}$, we see that for $d<\dours(k)$, the operator $\LDELitS{d,k}$
contracts with respect to the metric $\dist_{d}$, as desired.
\end{proof}

\subsection{Proof of \Prop~\ref{prop:condMarg}}
To get a handle on the $\ETA_{x}^{(\ell)}$ from \eqref{eta}, we show that these quantities
can be calculated by propagating the extremal boundary condition $\SIGMA^+$
bottom-up toward the root of the tree. Specifically, we consider the operator
\begin{align*}
  \Uplambda^{+}_{{\TT}^{(\ell)}}
&:
(-\infty,\infty]^{V({\TT}^{(\ell)})}
\to
(-\infty,\infty]^{V({\TT}^{(\ell)})}
\enspace,&
\eta&\mapsto\hat\eta=\Uplambda^{+}_{{\TT}^{(\ell)}}(\eta)
\enspace,
\end{align*}
defined as follows. 
For all $x\in\partial^{2\ell}\root$ we set $\hat\eta_x=\infty$. 
Moreover, for a variable $x\in\partial^{2q}\root$ with $q<\ell$ having children clauses $a_1,\ldots,a_t$, and grandchildren variables $y_{1,1}, \ldots, y_{1,(k-1)}, \ldots, y_{t,1}, \ldots, y_{t,(k-1)}$ we define
\begin{align}\label{eqhatetax}
\hat\eta_x &= -\sum_{i=1}^{t} {\TAU^+(x)}\sign(x,a_i)
\cdot
\log\left(
1-
{\Pfun\left({\TAU^+(x)}\sign(x,a_i)\cdot
(
\eta_{y_{i,1}}, \ldots, \eta_{y_{1,(k-1)}}
)
\right)}
\right)
\enspace.
\end{align}
It may not be apparent that the sum above is well-defined as $-\infty$ summands may manifest. 
The following lemma rules out such possibility and shows that the $\ell$-fold iteration of $\Uplambda^{+\, (\ell)}_{{\TT}^{(\ell)}}$, initiated all-$(+\infty)$ yields $\ETA^{(\ell)}=(\ETA^{(\ell)}_x)_{x \in V({\TT}^{(\ell)})}$.

\begin{lemma}\label{lem:upthetree}
The operator $\Uplambda^{+}_{{\TT}^{(\ell)}}$ is well-defined and
$\Uplambda^{+\, (t)}_{{\TT}^{(\ell)}}(+\infty,\ldots,+\infty)= \ETA^{(\ell)}$
for every $t \ge \ell$.
\end{lemma}

\begin{proof}
To show that $\Uplambda^{+}_{{\TT}^{(\ell)}}$ is well defined we verify that, in the notation of \eqref{eqhatetax},
$\hat\eta_x\in(-\infty,\infty]$ for all $x$. Indeed, in the expression on the r.h.s.\ of \eqref{eqhatetax} a
$\pm\infty$ summand can arise only from variables $y_{i,j}$ with $\eta_{y_{i,j}}=\infty$. But the definition of
$\TAU^+$ ensures that such $y_{i,j}$ either render a zero summand if $\TAU^+(x)\sign(x,a_i)=-1$, or a $+\infty$ summand if $\TAU^+(x)\sign(x,a_i)=1$.
Thus, the sum is well-defined and $\hat\eta_x\in(-\infty,\infty]$.

Further, to verify the identity $\ETA^{(\ell)}=\Uplambda^{+\, (\ell)}_{{\TT}^{(\ell)}}(\infty,\ldots,\infty)$, consider a variable $x$ of $\TT^{(\ell)}$.
Let $a_1^+,\ldots,a_g^+$ be the children (clauses) of $x$ with $\sign(x,a_i^+)=\TAU^+(x)$. Also let $y_{11},\ldots,y_{1(k-1)}, \ldots,
y_{g1}, \ldots, y_{g(k-1)}$ be the children of $a_1^+,\ldots,a_g^+$. 
Similarly, let $a_1^-,\ldots,a_h^-$ be the children of $x$ with
$\sign(x,a_i^-)=-\TAU^+(x)$ and let
$z_{11},\ldots,z_{1(k-1)}, \ldots, z_{h1}, \ldots, z_{h(k-1)}$ be their
children. Then \eqref{eq:SplusMax}, and \eqref{eq:SminusMax}  yield
\begin{align*}
\ETA_x^{(\ell)}
&=-
\sum_{i=1}^g
\log\left(1-\prod_{q=1}^{k-1}
\frac{Z({\TT}_{y_{iq}}^{(\ell)},\TAU^+,\TAU^+({y_{iq}}))}
{Z({\TT}_{y_{iq}}^{(\ell)},\TAU^+)}
\right)
+\sum_{j=1}^h\log
\left(
1-
\prod_{q=1}^{k-1}
\frac{Z({\TT}_{z_{jq}}^{(\ell)},\TAU^+, -\TAU^+({z_{jq}}))}
{Z({\TT}_{z_{jq}}^{(\ell)},\TAU^+)}
\right)
\\
& =-\sum_{i=1}^g\log
\left(1- \Pfun\left(\sign(x,a_i^+)\TAU^+(x)\cdot\left(\ETA^{(\ell)}_{y_{i1}}, \ldots, \ETA^{(\ell)}_{y_{i(k-1)}}\right) \right)\right)
+\sum_{j=1}^h \log
\left(1- \Pfun\left( \sign(x,a_i^-)\TAU^+(x)\cdot\left(\ETA^{(\ell)}_{z_{j1}}, \ldots, \ETA^{(\ell)}_{z_{j(k-1)}}\right)\right)\right)
\enspace.
\end{align*}
The assertion follows because $\sign(x,a_i^+)\TAU^+(x)=1$ and $\sign(x,a_i^-)\TAU^+(x)=-1$.
\end{proof}

The next aim is to approximate the $\ell$-fold iteration of $\Uplambda^{+}_{{\TT}^{(\ell)}}$, and more specifically the distribution of $\ETA_\root^{(\ell)}$, using a non-random operator.
To this end, we need to cope with the $\pm\infty$-entries of the vector $\ETA^{(\ell)}$. 
This is addressed by \Lem~\ref{lem_SmallOnTop}, proven in \Sec~\ref{sec:lem:SmallOnTop}, which provides a bound on $\ETA^{(\ell)}_x$ for variables $x$ near the root of the tree.

In the following we continue to write $c$ and $(\varepsilon_t)_t$ for the number and the sequence supplied by Lemma~\ref{lem_SmallOnTop}.
Guided by Lemma \ref{lem_SmallOnTop} we consider the vector $\ETA_{\wedge t}^{(\ell)}$ of truncated log-likelihood ratios
\begin{align}\label{eq:truncEta}
  \left(\ETA_{\wedge t}^{(\ell)}\right)_x
&=
\begin{cases}-2t^c&\mbox{ if $x\in\partial^{2t}\root$ and } \ETA_{x}^{(\ell)}<-2t^c \enspace, \\
2t^c&\mbox{ if $x\in\partial^{2t}\root$ and } \ETA_{x}^{(\ell)}>2t^c \enspace,\\
\ETA_{x}^{(\ell)}&\mbox{ otherwise} \enspace.
\end{cases}
\end{align}
Further, let
$
\ETA^{(\ell,t)}
$
be the result of $t$ iterations of $\Uplambda^{+}_{{\TT}^{(\ell)}}(\nix)$ starting from $\ETA_{\wedge t}^{(\ell)}$.
The following corollary is a direct consequence of \Lem~\ref{lem:upthetree} and  \Lem~\ref{lem_SmallOnTop}.

\begin{corollary}\label{Cor_boundary}
For any $\ell>ct^c$ we have $\Pr[\ETA^{(\ell,t)}_\root \neq \ETA^{(\ell)}_\root]<\eps_t.$
\end{corollary}

\begin{proof}
Due to Lemma~\ref{lem_SmallOnTop}, the truncation in \eqref{eq:truncEta} is
inconsequential with probability at least $1-\varepsilon_t$, in which case
\begin{align*}
\ETA^{(\ell,t)}
=
\Uplambda^{+\, (t)}_{{\TT}^{(\ell)}}\left(\ETA_{\wedge t}^{(\ell)}\right)
=
\Uplambda^{+\, (t)}_{{\TT}^{(\ell)}}(\ETA^{(\ell)})
=
\Uplambda^{+\, (\ell+ t)}_{{\TT}^{(\ell)}}(+\infty, \ldots, +\infty)
=
\ETA^{(\ell)}
\enspace,
\end{align*}
where the last equality follows from Lemma~\ref{lem:upthetree}.
\end{proof}

Recall that we defined the non-random operator $\LDELitS{d,k}$ from~\eqref{eqLDELitS}, mimicking $\Uplambda^{+}_{{\TT}^{(\ell)}}$.
To make the connection between the
random operator $\Uplambda^{+}_{{\TT}^{(\ell)}}$ and $\LDELitS{d,k}$ precise, we introduce the
following concepts. 
Given a tree formula $T$ we write $V_{\AllKid}(T)$, for the set of $x$ variables of $T$ that appear both as positive and negative literals in the sub-tree $T_x$ comprising $x$ and its the progeny. 
We define  $V_{\NoKid}(T), V_{\PureP}(T)$, and $V_{\PureM}(T)$ similarly.
Note that the above sets constitute a partition of $V(T)$. 
We use $\type: V(T) \to \{\rAllKid, \rPureP, \rPureM, \rNoKid\}$ to indicate the part each vertex belongs to.
We denote with $\TT^{(\ell)}_{\AllKid}$ the random Galton-Watson formula $\TT$ conditioned on the root satisfying $\type(\root) = \rAllKid$. 
We define $\TT^{(\ell)}_{\PureP}$, and $\TT^{(\ell)}_{\PureM}$ analogously. 
Degenerately, we also write $\TT^{(\ell)}_{\NoKid}$ for the formula comprised by a single variable $\root$.
Let us denote with $\hat\eta_{\AllKid}^{(\ell,t)}$ the distribution of $(\ETA_{\wedge t}^{(\ell)})_\root$ in $\TT^{(\ell)}_{\AllKid}$. 
Moreover, let $\bar\eta^{(\ell-t)}_{\AllKid}$ be the distribution of
\begin{align*}
\ETA_{\root}^{(\ell-t)} \cdot\Ind\left\{|\ETA_{\root}^{(\ell-t)}| \le 2t^c\right\}  + 2t^c \cdot\Ind\left\{\ETA_{\root}^{(\ell-t)} > 2t^c\right\}
-2t^c \cdot\Ind\left\{\ETA_{\root}^{(\ell-t)} < -2t^c\right\}
\enspace,
\end{align*}
i.e., the truncation of $\ETA_{\root}^{(\ell-t)}$ in
$\TT^{(\ell)}_{\AllKid}$.
Analogously we define the distributions $ \hat\eta^{(\ell, t)}_{\PureP}, \hat\eta^{(\ell, t)}_{\PureM}$, and
$\bar\eta^{(\ell-t)}_{\PureP}, \bar\eta^{(\ell- t)}_{\PureM}$. Notice that, degenerately,
$\hat\eta^{(\ell,t)}_{\NoKid} = \bar\eta^{(\ell-t)}_{\NoKid} = \delta_{0}$.

\begin{lemma}\label{lem_NoR_R_OP}
For $\ell > ct^c$ we have that
$ \left(\hat\eta^{(\ell, t)}_{\AllKid}, \hat\eta^{(\ell, t)}_{\PureP}, \hat\eta^{(\ell, t)}_{\PureM}\right)
= \LDELitS{d,k} \left(\bar\eta^{(\ell- t)}_{\AllKid}, \bar\eta^{(\ell- t)}_{\PureP}, \bar\eta^{(\ell- t)}_{\PureM}\right).
$
\end{lemma}

\begin{proof}
We use induction on $t$. Specifically, let $\nu = \left(\nu_{\AllKid}, \nu_{\PureP}, \nu_{\PureM}\right)$ be any triplet in
$\cP(-\infty, \infty] \times \cP[0, +\infty] \times \cP[-\infty,0]$, and $\nu^{(t)} =\LDELit^{\star (t)}_{d,k}(\nu)$ be the outcome
of the $t$-fold application of $\LDELitS{d,k}$. Moreover, let $(\ETA_x)_{x\in V(\TT^{(t)})}$ be a vector of independent samples
with  $\ETA_x \disteq \nu_{\type(x)}$.
We claim that root value $\ETA^{(t)}_\root$ of the random operator $\Uplambda^{+\, (t)}_{{\TT^{(t)}}}$,  coincides with $\nu_{\type(\root)}$. Indeed, for
$t=1$ the claim follows readily from the definitions. For the inductive step, we notice that the $t$-fold application of $\LDELitS{d,k}$ is
obtained by applying $\LDELitS{d,k}$ to the $(t-1)$-fold application. 
Per the induction hypothesis
\begin{align}
\bc{\Uplambda^{+\, (t-1)}_{\TT^{(t-1)}}(\ETA_x)_{x}}_{\root} \disteq \nu^{(t-1)}_{\type(\root)}
\enspace.
\end{align}
Applying $\LDELitS{d,k}$ to $\nu^{(t-1)}$ implies the result as the first layer of $\TT^{(t)}$ is independent of the subtrees rooted at the
grandchildren $\partial^2 \root$ of the root, which are distributed as i.i.d.\ copies of $\TT^{(t-1)}$. 
The lemma follows from applying the above identity to
$\nu = \left(\bar\eta^{(\ell- t)}_{\AllKid}, \bar\eta^{(\ell- t)}_{\PureP}, \bar\eta^{(\ell- t)}_{\PureM}\right)$.
\end{proof}

Refining the definition of the $\BP_{d,k}$ operator in \eqref{eqBPop}, we write  $\BP_{d,k}^{\AllKid}$ for the operator obtained from
$\BP_{d,k}$ upon conditioning on $\vd^+,\vd^- \ge 1$. Similarly $\BP_{d,k}^{\PureP}$ and $\BP_{d,k}^{\PureM}$ are obtained from
$\BP_{d,k}$ upon conditioning on $\vd^+\ge 1,\vd^- = 0$, and $\vd^+ = 0 ,\vd^- \ge 1$, respectively.
We define
\begin{align*}
\pi^{\AllKid}_{d,k}=\BP_{d,k}^{\AllKid} \left(\pi_{d,k}\right),\enspace
\pi^{\PureP}_{d,k}=\BP_{d,k}^{\PureP} \left(\pi_{d,k}\right),\enspace
\pi^{\PureM}_{d,k}=\BP_{d,k}^{\PureM} \left(\pi_{d,k}\right)\enspace.
\end{align*}
Let us write $\fratio, \fprop$ for the continuous and mutually inverse real functions
\begin{align}
\fratio: \RR \to (0,1), \enspace z \mapsto (1 + \tanh(z/2))/2, &&
\fprop: (0,1) \to \RR, \enspace p \mapsto \log (p/(1-p))
\enspace.
\end{align}
Let $\rho^{\AllKid}_{d,k}=\fprop(\pi_{d,k}^{\AllKid})$,  and define $\rho^{\PureP}_{d,k}, \rho^{\PureM}_{d,k}$ similarly.

\begin{claim}\label{cl_FixedPoints_LL}
The vector $\bc{\rho^{\AllKid}_{d,k}, \rho^{\PureP}_{d,k}, \rho^{\PureM}_{d,k}}$ is a fixed point of the operator $\LDELitS{d,k}$.
\end{claim}

\begin{proof}
Let $\rho_{d,k}=\fprop\bigl(\pi_{d,k}\bigr)$.
First, we claim that
\begin{align}\label{eq_dixq}
\LDELitS{d,k}\bc{\rho_{d,k}, \rho_{d,k}, \rho_{d,k}} = \bc{\rho^{\AllKid}_{d,k}, \rho^{\PureP}_{d,k}, \rho^{\PureM}_{d,k}}
\enspace.
\end{align}
Indeed, since all input distributions are the same, by \Prop~\ref{prop_arnab}, the two summands in 
the left term of \eqref{eq:AllKidRecGUnq} corresponding to $\vdpc$ and $\vdmc$ are identically distributed, and also identically distributed to the sums that appear in the other two terms.
Therefore,
\eqref{eq_dixq} follows directly from the definitions of $\BP_{d,k}^{\AllKid}, \BP_{d,k}^{\PureP}$, and $ \BP_{d,k}^{\PureM}$.
The claim now follows from \eqref{eq_dixq}, the definition of the operator $\LDELitS{d,k}$, and the law of total probability.
\end{proof}

Let $\rho^{(\ell)}$ be the distribution of the log-likelihood ratio $\ETA^{(\ell)}_\root$.

\begin{corollary}\label{cor_wk_prp}
For $d < \dours(k)$ the sequence $\bc{\fratio\bc{\rho^{(\ell)}}}_{\ell}$ converges weakly to $\pi_{d,k}$.
\end{corollary}

\begin{proof}
The result follows by combining \Cor~\ref{Cor_boundary}, \Lem~\ref{lem_NoR_R_OP}, \Prop~\ref{prop_BPPureLit}, \Cl~\ref{cl_FixedPoints_LL},
and applying the continuous mapping theorem and the law of total probability.
\end{proof}

\begin{proof}[Proof of \Prop~\ref{prop:condMarg}]
Recall that we write $\Uplambda_{\TT^{(\ell)}}^{+ \, (\ell)}$ for the $\ell$-fold iteration of the operator $\Uplambda^{+}_{{\TT}}$.
Let us write $\THETA^{(\ell)}_\root = \bigl(\Uplambda^{+\, (\ell)}_{\TT^{(\ell)}}(0, \ldots, 0)\bigr)_\root$.
Using arguments similar to \Fact~\ref{fac_BPexact}, we can show that  $\THETA^{(\ell)}_\root$ is nothing but the distribution of the random variable
$\fprop(\pr\brk{\TAU^{(\ell)}(\root)=1\mid\TT})$. Therefore,
\begin{align*}
{\pr[\TAU^{(\ell)}(\root)=1\mid\TT]}\disteq
\fratio(\THETA^{(\ell)}_\root) \enspace, &&
\text{ and }
&&
{\pr[\TAU^{(\ell)}(\root)=1 \mid \forall y\in\partial^{2\ell}\root: \TAU^{(\ell)}(y) = \TAU^+(y),\TT]} \disteq \fratio(\ETA^{(\ell)}_{\root})
\enspace.
\end{align*}
Due to \Lem~\ref{lem:ExtremalBnd}, $0 \le \fratio(\THETA^{(\ell)}_\root) \le \fratio(\ETA^{(\ell)}_\root) \le 1$. 
Moreover, from \Lem~\ref{lem_NoR_R_OP}, \Prop~\ref{prop_BPPureLit}, and \Cl~\ref{cl_FixedPoints_LL}, we see that for $d < \dours(k)$ the sequence $\bigl(\fratio(\THETA^{(\ell)}_\root)\bigr)_\ell$ converges weakly to $\pi_{d,k}$. 
Finally, \Cor~\ref{cor_wk_prp} implies that $\bigl(\fratio(\ETA^{(\ell)}_\root)\bigr)_\ell$ also converges weakly to $\pi_{d,k}$, and thus,
\begin{align*}
\lim_{\ell \to \infty} \Erw \left|\fratio(\THETA^{(\ell)}_\root)- \fratio(\ETA^{(\ell)}_\root)\right|
=
\lim_{\ell \to \infty} \left|\Erw\brk{\fratio(\THETA^{(\ell)}_\root)} - \Erw\brk{\fratio(\ETA^{(\ell)}_\root)}\right| = 0 \enspace,
\end{align*}
implying the assertion.
\end{proof}

\section*{Acknowledgements}

We would like to thank the anonymous referees for thoroughly reviewing our paper and for suggesting valuable corrections and improvements.

Amin Coja-Oghlan is supported by DFG CO 646/3, DFG CO 646/5 and DFG CO 646/6. Catherine Greenhill is supported by ARC DP250101611.
Vincent Pfenninger is supported by the Austrian Science Fund (FWF) [10.55776~/ 16502]. Pavel Zakharov is supported by DFG CO 646/6. Kostas Zampetakis is supported by DFG CO 646/5.
For open access, the authors have applied a CC BY public copyright licence to any Author Accepted Manuscript version arising from this submission.

\end{document}